\documentclass[11pt]{article}

\usepackage{siunitx}
\usepackage{fancyhdr}
\usepackage{extramarks}
\usepackage{amsmath}
\usepackage{amsthm}
\usepackage{arydshln}
\usepackage{mathtools}
\usepackage{amsfonts}
\usepackage{enumerate}
\usepackage{graphicx}
 \graphicspath{ {figures/} }
\usepackage{algpseudocode}
\usepackage{textcomp}
\usepackage{import}
\usepackage{natbib}
\usepackage[font={footnotesize, it}]{caption}
\usepackage{bbm}
\usepackage{subcaption}
\usepackage[ruled,vlined]{algorithm2e} 
\usepackage[font={footnotesize, it}]{caption}
\usepackage{xr}
 \usepackage{hyperref}
 \hypersetup{
     linktoc=all,     
     linkcolor=black,  
 }
 \usepackage{geometry}
 
 \usepackage[T1]{fontenc}
\usepackage[utf8]{inputenc}
\usepackage{authblk}
\usepackage{blindtext}
\usepackage{booktabs}
\usepackage{array, makecell}

\newcommand{\pderiv}[2]{\frac{\partial #1}{\partial #2}}

\newcommand{\E}{\mathbb{E}}
\newcommand{\Var}{\mathrm{Var}}
\newcommand{\Cov}{\mathrm{Cov}}

\newtheorem{definition}{Definition}
\newtheorem{theorem}{Theorem}
\theoremstyle{definition}
\newtheorem{remark}{Remark}
\theoremstyle{definition}

\newtheorem{lemma}[theorem]{Lemma}
\theoremstyle{definition}
\newtheorem{proposition}[theorem]{Proposition}

\numberwithin{theorem}{section}

\begin{document}

\title{Predicting Rare Events by Shrinking Towards Proportional Odds}
\date{\today}

\author{Gregory Faletto\thanks{Corresponding author: gregory.faletto@marshall.usc.edu} }
\author{Jacob Bien}
\affil{Department of Data Sciences and Operations \\ University of Southern California Marshall School of Business}



%
       



\maketitle

\begin{abstract}

Training classifiers is difficult with severe class imbalance, but many rare events are the culmination of a sequence with much more common intermediate outcomes. For example, in online marketing a user first sees an ad, then may click on it, and finally may make a purchase; estimating the probability of purchases is difficult because of their rarity. We show both theoretically and through data experiments that the more abundant data in earlier steps may be leveraged to improve estimation of probabilities of rare events. We present \textsc{PRESTO}, a relaxation of the proportional odds model for ordinal regression. Instead of estimating weights for one separating hyperplane that is shifted by separate intercepts for each of the estimated Bayes decision boundaries between adjacent pairs of categorical responses, we estimate separate weights for each of these transitions. We impose an L1 penalty on the differences between weights for the same feature in adjacent weight vectors in order to shrink towards the proportional odds model. We prove that \textsc{PRESTO} consistently estimates the decision boundary weights under a sparsity assumption. Synthetic and real data experiments show that our method can estimate rare probabilities in this setting better than both logistic regression on the rare category, which fails to borrow strength from more abundant categories, and the proportional odds model, which is too inflexible.

\end{abstract}

\section{Introduction}

Estimating probabilities of rare events is known to be difficult due to class imbalance. However, sometimes these events are the culmination of a sequential process with intermediate outcomes. For example:
\begin{enumerate}
\item In online marketing, a customer is first served an ad, then may click on it, then may indicate interest in making a purchase (by ``liking" the product, for example), and finally may make a purchase. 
\item In health and medicine, many outcomes can be encoded as ordered categorical variables, like reported quality of life and disease progression \citep{norris2006ordinal}. 
\item Sales of high-price durable goods typically follow a \textit{sales funnel} \citep{10.1145/2783258.2788578}. For example, when buying a car often a potential buyer first comes in to see a car, may take a test drive, and finally may buy the car.
\end{enumerate}
In many of these cases, the intermediate events are much more common than the rare events. Though these intermediate events may not be of direct interest, if the features that contribute to the probability of advancing through earlier classes also contribute to the probability of advancing through later classes, then the more abundant intermediate events can be leveraged to improve estimation of the rare event probabilities.

The \textit{proportional odds model} \citep{mccullagh1980regression}, also called the \textit{ordered logit model} \citep[Section 15.9.1]{cameron2005microeconometrics}, satisfies, for ordinal outcomes \(k \in \{1, \ldots, K - 1\}\),
\begin{equation}\label{prop_odds}
\log \left( \frac{\mathbb{P} \left( y \leq k \mid \boldsymbol{x} \right)}{\mathbb{P} \left( y > k  \mid \boldsymbol{x} \right)} \right) = \alpha_k + \boldsymbol{\beta}^\top \boldsymbol{x}  ,
\end{equation}
where \(\boldsymbol{\beta} \in \mathbb{R}^p\) is a vector of weights and \(\boldsymbol{x} \in \mathbb{R}^p\) is a vector of features. This implies that for all \(k \in \{1, \ldots, K - 1\}\)
\begin{equation}\label{prop_odds2}
p_k ( \boldsymbol{x} ) := \mathbb{P} \left( y \leq k \mid \boldsymbol{x} \right) = F \left( \alpha_k + \boldsymbol{\beta}^\top \boldsymbol{x}  \right),  
\end{equation}
where \(F(\cdot)\) is the logistic cumulative distribution function, \(F(t) = \exp\{t\}/[1 + \exp\{t\}]\). Notice that \( \alpha_k + \boldsymbol{\beta}^\top \boldsymbol{x} \) is the Bayes decision boundary for the binary random variable \(\mathbbm{1} \left\{ y \leq k \right\} \mid \boldsymbol{x}  \). This problem could instead be cast as \(K - 1\) binary classification problems of the form \eqref{prop_odds2} for adjacent classes:
\begin{equation}\label{prop_odds_gen}
\log \left( \frac{\mathbb{P} \left( y \leq k \mid \boldsymbol{x} \right)}{\mathbb{P} \left( y > k \mid \boldsymbol{x} \right)} \right) = \alpha_k + \boldsymbol{\beta}_k^\top \boldsymbol{x} , \qquad k \in \{1, \ldots, K - 1 \}
.
\end{equation}
The condition that the weight vectors \(\boldsymbol{\beta}_k\) of the separating hyperplanes in \eqref{prop_odds_gen} are all equal, as in \eqref{prop_odds}, has been called the \textit{proportional odds assumption} \citep{mccullagh1980regression} or the \textit{parallel regression assumption} \citep[Section 18.3.2]{Greene2012Econometric}. One way to motivate this model is by supposing that the response is driven by a latent (unobserved) variable \(U\),
\begin{equation}\label{prop.odds.utilty}
U = \boldsymbol{\beta}^\top \boldsymbol{x} + \epsilon,
\end{equation}
where \(\epsilon\) has a standard logistic distribution and is independent of \(\boldsymbol{x}\). Response \(k\) is observed if and only if \( -\alpha_{k} \leq  U_i  < -\alpha_{k - 1} \) (where we define \(\alpha_0 := - \infty\) and \(\alpha_K := \infty\)). This model leads to \eqref{prop_odds2}. (See Section 3.3.2 of \citealt{agresti2010analysis} for a more detailed explanation.)

Because the proportional odds model assumes that the decision boundaries between adjacent classes are all governed by the same hyperplane defined by \(\boldsymbol{\beta}\) (separated only by different intercepts \(\alpha_k\)), it assumes that the decision boundary between any two classes perfectly explains the decision boundary between any two other classes, other than an intercept term. If a rare event has much more common intermediate events before it, this model can therefore be very useful for better estimating the parameters of the model, and therefore better estimating the rare event probabilities. However, it could be that the proportional odds assumption is too rigid to be realistic, because observed features may have varying influence at different decision boundaries. For example: 
\begin{enumerate}
\item In online marketing, users may click on an ad only to realize that the product is not what they were expecting, resulting in a particularly low probability of purchase.
\item For expensive goods like a home or car, potential buyers may express interest by going on a tour or taking a test drive purely out of curiosity; this may be distinct from their level of interest in actually making a purchase.
\item Students may place weights on different factors when deciding whether to apply to graduate school than they did when deciding whether to apply to an undergraduate program---they may have more appealing alternatives to additional schooling, they may face new financial or personal constraints because they are older, etc.
\end{enumerate}
In each of these settings, if specific features vary in relevance for different decision boundaries while other features have about the same influence at every boundary, the proportional odds assumption may be too strong. Violations or relaxations of the proportional odds assumption along the lines of \eqref{prop_odds_gen} have previously been considered by, for example, \citet{brant1990assessing}. \citet{peterson1990partial} developed \textit{partial proportional odds models}, which allow the proportional odds assumption to hold for some features but not others, an idea previously mentioned by \citet{armstrong1989ordinal}. (See Section 3.6.1 of \citealt{agresti2010analysis} for a textbook-level discussion). These relaxations have not been widely adopted because fitting separate weights for each outcome is too flexible unless \(p(K-1) \ll n\) and all classes are reasonably common (and we discuss additional difficulties of this kind of model in Sections \ref{presto.cons.sec} and \ref{sparse.sim}).

\subsection{Our Contributions}

In this paper we propose relaxing the proportional odds assumption as in \eqref{prop_odds_gen}, but controlling the amount of relaxation by placing \(\ell_1\) penalties on the differences in weights corresponding to the same features in adjacent \(\boldsymbol{\beta}_k\) vectors, in a way that is reminiscent of the fused lasso \citep{tibshirani2005sparsity}. This model allows us to borrow strength from outcomes where data is much more abundant to improve rare probability estimates when outcomes are much more rare without making the strong assumption that the weights in these models are exactly equal. In particular, it allows for the proportional odds model to hold for some specific features in some adjacent pairs of decision boundaries, but not others.

We formalize the intuitive argument we outline above---that the proportional odds model allows for precise estimation of the \(\boldsymbol{\beta}\) vector as long as at least one decision boundary is surrounded by reasonably well-balanced outcomes, and this allows for improved estimation of rare probabilities at the end of the sequence---through theoretical results in Section \ref{sec.theory.presto}. Motivated by this argument but skeptical of the proportional odds assumption holding exactly, we propose \textsc{PRESTO} in Section \ref{sec.method} and prove that it consistently estimates \(\boldsymbol{\beta}_1, \ldots, \boldsymbol{\beta}_{K-1}\) under a sparsity assumption in Section \ref{presto.cons.sec}. In Section \ref{sec.sim} we demonstrate through synthetic and real data experiments that \textsc{PRESTO} can outperform both logistic regression on the rare class and the proportional odds model, both in settings where the differences in adjacent \(\boldsymbol{\beta}_k\) vectors are sparse, as \textsc{PRESTO} assumes, and in settings where these differences are not sparse. Before we move on from the introduction, we review related literature.

\subsection{Related Work}

The difficulty of classification with class imbalance has been well-known for decades. \citet{johnson2019} provide a recent review focusing on deep learning methods for handling class imbalance, and they also provide references for many other ways of dealing with class imbalance. One particularly closely related work is \citet{owen2007infinitely}, which explores how logistic regression handles a vanishingly rare class. A particularly popular approach, SMOTE \citep{chawla2002smote}, has its own recent review paper \citep{fernandez2018smote}.

\citet{tutz2016regularized} discuss the possibility of penalizing differences in weights between adjacent models, including briefly proposing an \(\ell_1\) penalty between weights in corresponding categories for proportional hazard models, though this is not the focus of their article and they only mention the idea very briefly without investigating it. 

\citet{ordnet} propose a generalization of a proportional odds model (and implement it in the R package \texttt{ordinalNet}) that allows for the possibility that adjacent categories have equal (or very close) weights, but their method differs from ours. The most closely related model \citeauthor{ordnet} propose is an over-parameterized \textit{semi-parallel} model with both a matrix of separate parameters for each level, an approach reminiscent of \citet{peterson1990partial}. This results in more flexible, less structured models than our approach, which assumes similarity between adjacent \(\boldsymbol{\beta}_k\) vectors. Further, \citet{ordnet} do not investigate the theoretical properties of their model, or the use of their model for improving estimates of rare event probabilities.

\citet{ugba2021smoothing} and \citet{epub26912} implement an \(\ell_2\) rather than \(\ell_1\) penalty between weights in models for adjacent decision boundaries. However, these works also focus on ordinal regression more generally, while we focus both theoretically and in simulations on leveraging common classes to improve estimated probabilities of rare events. Further, the \(\ell_1\) penalty, which imposes sparse differences, allows the proportional odds assumption to hold for some features and decision boundaries and not others, while the \(\ell_2\) ridge penalty used by \citet{ugba2021smoothing} (and previously proposed by \citealt[Section 4.2.2]{tutz2016regularized}) relaxes the proportional odds assumption for all features but regularizes the relaxation. The \(\ell_2\) group lasso penalty used by \citet{epub26912} can impose the proportional odds assumption for a given feature either at all decision boundaries or none of them, making it less flexible than \textsc{PRESTO}.

Besides the fused and generalized \citep{tibshirani2011solution} lasso, our work relates more specifically to the generalized fused lasso \citep{hofling2010coordinate, xin2014efficient}. \citet{xin2014efficient} in particular propose and analyze an algorithm to solve a class of optimization problems similar to the PRESTO optimization problem, \eqref{ordinal.pen.opt}, with \(\ell_1\) fusion penalties. In contrast to the present work, \citet{xin2014efficient} focus almost entirely on the properties of their algorithm. Further, PRESTO lies outside their class of optimization problems because PRESTO directly penalizes only the coefficients in the first decision boundary, not all of the coefficients. This distinction is central to our proof strategy for Theorem \ref{cons.thm}; \citet{xin2014efficient} do not prove the consistency of their method. \citet{viallon2013adaptive} provide theoretical results for the generalized fused lasso specifically in the cases of linear and logistic regression, though not for ordinal regression. \citet{viallon2016robustness} prove theoretical results for a broader class of generalized linear models that still does not include the proportional odds model or a generalization like PRESTO. Lastly, \citet{Ekvall2022} prove theoretical results for a class of models in which PRESTO can be expressed, and indeed we leverage their results in proving our own theory, though they do not directly consider fusion penalties.

\section{Motivating Theory}\label{sec.theory.presto}

We present the following theoretical results to motivate \textsc{PRESTO}. The thrust of our motivation is as follows: 
\begin{enumerate}
\item Logistic regression does arbitrarily badly as class imbalance worsens (Theorem \ref{log.imb}).
\item However, as one would expect, a logistic regression model's ability to estimate probabilities improves when the parameters \(\boldsymbol{\beta}\) are known (Theorem \ref{est.known.beta}). 
\item The proportional odds model allows for precise estimation of \(\boldsymbol{\beta}\) as long as two adjacent classes are reasonably common, even if the remaining classes are arbitrarily rare (Theorem \ref{main.cov.thm.2}). 
\item Taking 2 and 3 together, our conclusion is that we can better estimate probabilities of rare events by using a method that leverages data from decision boundaries between abundant classes to better estimate decision boundaries near rare classes. (Both the proportional odds model and \textsc{PRESTO} leverage the data in this way.)
\end{enumerate}

Before we present our results, we discuss the metrics we will use in our results and some of the assumptions we will make.

\subsection{Preliminaries}

Our goal is to characterize and compare the prediction error of estimated conditional probabilities of a rare class from both logistic regression and the proportional odds model. There are many settings where estimating rare probabilities accurately (as opposed to, for example, predicting class labels accurately) is important. For example, in online advertising, advertisers bid on the price to display an ad to a given user. Advertisers could bid optimally if they knew the true probability each user would click a given ad, so they'd like to estimate these probabilities as precisely as possible \citep{he2014practical, zhang2014optimal}. Another example is public policy, where scarce resources may be allocated based on estimated probability of bad outcomes \citep{von2019predicting}. To prioritize optimally, precisely estimated probabilities are needed, not just accurate labels. 

A natural metric in an estimation setting is mean squared error, \( \E \left[  \left( \hat{\pi}(\boldsymbol{x}) - \pi(\boldsymbol{x})  \right)^2 \right]\), where \(\pi(\boldsymbol{x}) \) is the actual probability of a rare event conditional on \(\boldsymbol{x}\) and \(\hat{\pi}(\boldsymbol{x})\) is an estimate. Further, we leverage asymptotic statistics and present results for \textit{large-sample} estimators. We define the notions of asymptotic mean squared error we will use below:

\begin{definition}\label{def.asym.mse}

Let \(\hat{\theta}_n\) be a maximum likelihood estimator for a parameter \(\theta \in \mathbb{R}\) from a sample size of \(n\). Under regularity conditions, the sequence of random variables \(\{\sqrt{n} \cdot ( \hat{\theta}_n - \theta)\}\) converges in distribution to a Gaussian random variable. Then we define the \textbf{asymptotic mean squared error} of \(\hat{\theta}_n\) to be (suppressing \(n\) from the notation)
\[
\mathrm{Asym. MSE}(\hat{\theta}):= \E \left[ \left( \lim_{n \to \infty}  \sqrt{n} \left[\hat{\theta}_n - \theta \right] \right)^2 \right]
.
\]

\end{definition}
Asymptotic metrics are commonly used to compare the performance of estimators. The \textit{asymptotic relative efficiency} of two estimators is the ratio of their asymptotic variances,
\[
\mathrm{Asym. }\Var(\hat{\theta}):=  \Var \left( \lim_{n \to \infty}  \sqrt{n} \left[\hat{\theta}_n - \theta \right] \right),
\]
which is equal to \(\mathrm{Asym. MSE}(\hat{\theta}_n)\) for the (asymptotically unbiased) maximum likelihood estimators we consider. See Section 10.1.3 of \citet{casella2021statistical}, Section 8.2 of \citet{van2000asymptotic}, or Section 4.4.5 of \citet{Greene2012Econometric} for textbook-level discussions. The asymptotic MSE could also be used as an estimator of the MSE for large (but finite) \(n\), under the heuristic reasoning that for large \(n\),
\begin{align*}
 \mathrm{MSE}(\hat{\theta}) =  ~ & \frac{1}{n} \E \left[  \left(  \sqrt{n} \cdot \left[ \hat{\theta} - \theta \right] \right)^2 \right] 
\\  \approx   ~ &    \frac{1}{n}  \E \left[ \left( \lim_{n \to \infty}  \sqrt{n} \cdot \left[ \hat{\theta} - \theta \right] \right)^2 \right] 
\\= ~ &   \frac{1}{n}   \mathrm{Asym. MSE}(\hat{\theta})
.
\end{align*}
See Section 4.4 of \citet{Greene2012Econometric}, Section 7.3 of \citet{hansen2022econometrics}, or Section 3.5 of \citet{wooldridge2010econometric} for more discussion of this kind of finite-sample estimation using asymptotic quantities.

We briefly present and discuss some of our assumptions.

\begin{itemize}

\item \textbf{Assumption \(X(\mathcal{A})\):} The random vectors \(\boldsymbol{x}_i \in \mathbb{R}^p\) are independent and identically distributed (iid) for \(i \in \{1, \ldots, n\}\), each with probability measure \(dF(\boldsymbol{x})\) with measurable, bounded support \(\mathcal{S} \subset \mathcal{A} \subseteq \mathbb{R}^p\), with \(\Cov \left(\boldsymbol{X} \right)\) positive definite.

\item \textbf{Assumption \(Y(K)\):} The response \(y_i \in \{1, \ldots, K\}\) is distributed conditionally on \(\boldsymbol{x}_i\) as in the proportional odds model \eqref{prop_odds}. (Note that if \(K = 2\), this is equivalent to the logistic regression model.) All classes have positive probability for all \(\boldsymbol{x}\) on the support of \(\boldsymbol{x}_i\) (equivalently, the intercepts strictly differ: \(\alpha_1 < \ldots < \alpha_{K-1}\).) 

\end{itemize}

Assumption \(X(\mathcal{A})\) allows a very broad class of distributions, including both discrete and continuous random variables. Notice that the boundedness assumption within \(X(\mathcal{A})\) implies that the matrix \(\boldsymbol{\tilde{X}} := (\boldsymbol{1}, \boldsymbol{X})\) (where \(\boldsymbol{1}\) is an \(n\)-vector of ones) has a finite maximum eigenvalue. When we will refer to it, we call it \(\lambda_{\text{max}}\) and write \textbf{Assumption \(X(\mathcal{A}, \lambda_{\text{max}})\)}.

From \eqref{prop_odds2} we see that in the proportional odds model if the intercepts strictly differ (\(\alpha_1 < \ldots < \alpha_{K-1}\)) then for any \(\boldsymbol{x}\) all of the classes have conditional probability strictly between 0 and 1. That said, if the support of \(\boldsymbol{X}\) is unbounded then all of the probabilities for individual classes can become arbitrarily close to 0 or 1. Under Assumption \(X(\mathcal{A})\), however, we can strictly bound quantities like \(\sup_{\boldsymbol{x} \in \mathcal{S}} \{\pi_k(\boldsymbol{x})\}\) (where \(\pi_k(\boldsymbol{x}) := p_k(\boldsymbol{x}) - p_{k-1}(\boldsymbol{x}) = \mathbb{P}( y = k \mid \boldsymbol{x})\)) away from 0 or 1.

Theorem \ref{est.known.beta} holds under Assumption \(X([0, \infty)^p)\), though for any bounded \(\mathcal{S} \subseteq \mathbb{R}^p\), there is some finite \(a\) one could add to each coordinate to shift \(\mathcal{S}\) to a subset of \([0, \infty)^p\); Theorem \ref{est.known.beta} would then apply to these translated features. 

\subsection{Theorem \ref{log.imb}}\label{thm.log.imb.sec}

It is well-known that class imbalance poses a major challenge for classifiers. Theorem \ref{log.imb} exhibits this concretely for logistic regression.

\begin{theorem}\label{log.imb}

Assume \(X(\mathbb{R}^p, \lambda_{\text{max}})\) and \(Y(2)\) hold. Let \(\pi(\boldsymbol{x}) := \mathbb{P}(y =2 \mid \boldsymbol{x})\), and assume that \(\sup_{\boldsymbol{x} \in \mathcal{S}} \pi(\boldsymbol{x}) = \pi_{\text{rare}}\) for some \(\pi_{\text{rare}} \leq 1/2\). Then 

\begin{enumerate}

\item for any fixed \(\boldsymbol{v} \in \mathbb{R}^{p+1}\),
\[
\frac{1}{\lVert \boldsymbol{v} \rVert_2^2}\mathrm{Asym. MSE}\left(   (\hat{\alpha}, \boldsymbol{\hat{\beta}}^\top)  \boldsymbol{v} \right) \geq  \frac{1}{\lambda_{\text{max}}  \pi_{\text{rare}} }
,
\]

and

\item for any fixed \(\boldsymbol{z} \in \mathcal{S}\),

\[
\mathrm{Asym. MSE}\left( \frac{\hat{\pi}(\boldsymbol{z}) }{ \pi (\boldsymbol{z})} \right) \geq \frac{1 - \pi_{\text{rare}} }{ \pi_{\text{rare}}} \frac{1}{\lambda_{\text{max}}}.
\]

\end{enumerate}
\end{theorem}

\begin{proof} Provided in Section \ref{thms.1.2.proofs}. \end{proof}

To give an example of applying part 1 of this result, consider the choice \(\boldsymbol{v} = (0, 1, 0, \ldots, 0)\). Then we have that \( \mathrm{Asym. MSE}\left( \hat{\beta}_1 \right) \geq 1/(\lambda_{\text{max}} \pi_{\text{rare}})\), so \(\hat{\beta}_1\) (or any other estimated coefficient) has arbitrarily large asymptotic mean squared error as \(\pi_{\text{rare}}\) vanishes. Part 2 shows that the same thing happens to the asymptotic mean squared error for the estimated probabilities of the logistic regression estimator, when scaled by \(\pi(\boldsymbol{z})\).

\subsection{Theorem \ref{est.known.beta}}

Theorem \ref{est.known.beta} suggests a possible way to circumvent the problem of class imbalance. We compare the typical logistic regression intercept estimate \(\hat{\alpha}\) to the \textit{quasi-estimated} estimator \(\hat{\alpha}_q\) obtained when one estimates only the intercept of the logistic regression model with a known \(\boldsymbol{\beta}\). We also compare the resulting estimators of conditional probabilities for any \(\boldsymbol{z} \in \mathbb{R}^p\): the usual logistic regression estimator \(\hat{\pi}(\boldsymbol{z}) \) and \(\hat{\pi}_q(\boldsymbol{z}) \), the estimator when \(\boldsymbol{\beta}\) is known. Theorem \ref{est.known.beta} proves the reasonable intuition that \(\hat{\alpha}_q\) must be a better estimator than \(\hat{\alpha}\), and likewise for \(\hat{\pi}_q(\boldsymbol{z}) \) and \(\hat{\pi}(\boldsymbol{z})\). 

\begin{theorem}\label{est.known.beta}

Assume \(X([0, \infty)^p, \lambda_{\text{max}})\) and \(Y(2)\) hold. Let \(\pi(\boldsymbol{x}) := \mathbb{P}(y =2 \mid \boldsymbol{x})\), and let  \(\pi_{\text{min}} := \inf_{\boldsymbol{x} \in \mathcal{S}} \left\{ \pi(\boldsymbol{x}) \wedge 1 - \pi(\boldsymbol{x})  \right\} \). Then

\begin{enumerate}
\item 
\begin{align*}
\frac{\mathrm{Asym. MSE}(\hat{\alpha})    -  \mathrm{Asym. MSE}(\hat{\alpha}_q)  }{ \left[ \mathrm{Asym. MSE}(\hat{\alpha}_q) \right]^2} \geq   \Delta
\end{align*}
where
\[
\Delta :=   \frac{4  \pi_{\text{min}}^2(1 - \pi_{\text{min}})^2 \left \lVert \E \left[ \boldsymbol{X} \right] \right \rVert_2^2} {\lambda_{\text{max}}} ,
\]
and

\item For any \(\boldsymbol{z} \in \mathbb{R}^p \setminus \{\boldsymbol{z}^*\}\), where
\[
\boldsymbol{z}^* := \frac{\E \left[ \boldsymbol{X} \pi(\boldsymbol{X})[1 - \pi(\boldsymbol{X})] \right] }{ \E \left[ \pi(\boldsymbol{X})[1 - \pi(\boldsymbol{X})] \right]},
\]
it holds that
\[
\mathrm{Asym. MSE}(\hat{\pi}_q(\boldsymbol{z})) < \mathrm{Asym. MSE}(\hat{\pi}(\boldsymbol{z}))
.
\]
(For \(\boldsymbol{z}^*\), the above holds with \(\leq\) rather than \(<\).)
\end{enumerate}
\end{theorem}

Examining the first result, it is sensible that the lower bound for the gap between the asymptotic variances of the two estimators vanishes as \(\pi_{\text{min}}\) vanishes because if \(\min\{\pi_1(\boldsymbol{x}), 1 - \pi_1(\boldsymbol{x})\}\) becomes very small on the bounded support, then the imbalance between the two classes potentially becomes very large, and estimating the intercept becomes difficult regardless of whether or not \(\boldsymbol{\beta}\) is known. As the class balance improves (\(\pi_{\text{min}}\) becomes closer to its upper bound \(1/2\)), the guaranteed gap between \(  \mathrm{Asym. MSE}(\hat{\alpha}) \) and \(  \mathrm{Asym. MSE}(\hat{\alpha}_q) \) becomes larger. 

In addition to formally verifying intuition, Theorem \ref{est.known.beta} also quantifies the estimation gap between the rare intercept estimators in terms of noteworthy parameters and shows that the assumptions needed for this intuition to hold are minimal.

\subsection{Theorem \ref{main.cov.thm.2}}

Theorem \ref{est.known.beta} suggests that if only we could estimate \(\boldsymbol{\beta}\) very well, we could improve our estimated probabilities even in the face of class imbalance. Theorem \ref{main.cov.thm.2} suggests a way to leverage abundant data among other classes to do this.

In the proportional odds model \eqref{prop_odds}, \(\mathbb{R}^p\) is partitioned into regions with separating hyperplanes defined by \(\boldsymbol{\beta}\), which we note are Bayes decision boundaries: for \(\boldsymbol{x} \in \mathbb{R}^p\) such that \( \alpha_k + \boldsymbol{\beta}^\top \boldsymbol{x} = 0\), we have \(p_k ( \boldsymbol{x} ) = 1/2\).

Consider the setting of ordered categorical data generated by the proportional odds model with categories 1 and 2 similarly common over the support of a bounded distribution of \(\boldsymbol{x}_i\) and categories \(3, \ldots, K\) all rare. In this setting, for many of the observed values of \(\boldsymbol{x}_i\), the probabilities of being in class 1 or 2 will both be close to 1/2. Intuitively it should be relatively easy to estimate \(\boldsymbol{\beta}\) and \(\alpha_1\), the parameters that define the Bayes decision boundary between classes 1 and 2, and therefore \( p_1  ( \boldsymbol{x}_i )  \). Theorem \ref{est.known.beta} suggests this should help us in estimating the rare class probabilities. In Theorem \ref{main.cov.thm.2}, we prove that even if class \(K\) becomes arbitrarily rare, as long as the first two classes are reasonably well balanced, the proportional odds model still learns \(\boldsymbol{\beta}\) quite well.

\begin{theorem}\label{main.cov.thm.2} Assume \(X(\mathbb{R}^p)\) and \(Y(3)\) hold. Assume for all \(\boldsymbol{x} \in \mathcal{S}\) it holds that \(|\pi_k(\boldsymbol{x}) - 1/2 | \leq \Delta\) for \(k \in \{1, 2\}\) for some \(\Delta \in (0, 1/2)\) and let \(M := \sup_{\boldsymbol{x} \in \mathcal{S}} \lVert \boldsymbol{x} \rVert_2 \) (notice that \(X(\mathbb{R}^p)\) ensures that \(M < \infty\)). Suppose \(\sup_{\boldsymbol{x} \in \mathcal{S}} \{ \pi_3(\boldsymbol{x}) \} = \pi_{\text{rare}}\), where \(\pi_{\text{rare}}\) is no greater than
\begin{align}
\min \left\{ \frac{1}{2} \left(\frac{1}{2} - \Delta \right) \left(\frac{1}{2} + \Delta \right) ,
 \frac{  \lambda_{\text{min}} \left(I_{\beta \beta} -   2 \frac{I_{\beta \alpha_1} I_{\beta \alpha_1}^\top}{I_{\alpha_1 \alpha_1}}  \right)  }{3 M^2   \left(  M + 2 \right) }
 \right\} 
 \label{delta.small.assum.fin}
 ,
 \end{align}
where  \(\lambda_{\text{min}}(\cdot)\) denotes the minimum eigenvalue of \(\cdot\) and \( I_{\beta \beta} -   2 \frac{I_{\beta \alpha_1} I_{\beta \alpha_1}^\top}{I_{\alpha_1 \alpha_1}} \) is a symmetric matrix composed of terms from the Fisher information matrix for the proportional odds model (see the definitions of these terms in Equations \ref{alpha.block}, \ref{alpha.beta.block}, and \ref{beta.block} in the appendix). Then there exists \(C < \infty\) not depending on \(\pi_{\text{rare}}\) such that for any fixed \(\boldsymbol{v} \in \mathbb{R}^p\), 
 \begin{align*}
\frac{1}{\lVert \boldsymbol{v} \rVert_2^2} \mathrm{Asym. MSE} \left(  \boldsymbol{v}^\top \boldsymbol{\hat{\beta}}^{\text{prop. odds}}  \right) 
\leq  C 
.
\end{align*}
 
\end{theorem}
Theorem \ref{main.cov.thm.2} shows that In contrast to logistic regression, the proportional odds model still learns \(\boldsymbol{\beta}\) within a fixed precision even as \(\pi_{\text{rare}}\) vanishes.

\begin{remark}
We briefly discuss the upper bound \eqref{delta.small.assum.fin}. For this bound to make sense, it must hold that the symmetric matrix \( I_{\beta \beta} -   2 \frac{I_{\beta \alpha_1} I_{\beta \alpha_1}^\top}{I_{\alpha_1 \alpha_1}} \) is positive definite so that its minimum eigenvalue is strictly positive. The matrix \(\boldsymbol{S} := I_{\beta \beta} -    \frac{I_{\beta \alpha_1} I_{\beta \alpha_1}^\top}{I_{\alpha_1 \alpha_1}}  \) is the Schur complement of \(I_{\alpha_1 \alpha_1} = M_1\) in the submatrix
\begin{equation}\label{submat}
\begin{pmatrix}
I_{\alpha_1 \alpha_1} &  I_{\beta \alpha_1}^\top
\\ I_{\beta \alpha_1} &  I_{\beta \beta}
\end{pmatrix}
\end{equation}
of the Fisher information matrix \(I^{\text{prop. odds}} (\boldsymbol{\alpha}, \boldsymbol{\beta})\) for the proportional odds model (see Lemma \ref{asym.matrix} in the appendix). Note \eqref{submat} is a principal submatrix of the positive definite \(I^{\text{prop. odds}} (\boldsymbol{\alpha}, \boldsymbol{\beta})\), so is positive definite by Observation 7.1.2 in \citet{horn_johnson_2012}. From \eqref{alpha.block} we also know that \(I_{\alpha_1 \alpha_1} > 0\),  so \(\boldsymbol{S}\) is positive definite by Theorem 1.12 in \citet{zhang2005schur}. It seems plausible that
\[
I_{\beta \beta} -   2 \frac{I_{\beta \alpha_1} I_{\beta \alpha_1}^\top}{I_{\alpha_1 \alpha_1}}   = \boldsymbol{S} -  \frac{I_{\beta \alpha_1} I_{\beta \alpha_1}^\top}{I_{\alpha_1 \alpha_1}}   
\]
is also positive definite because \( I_{\beta \beta}\) is the inverse of the asymptotic covariance matrix of \(\boldsymbol{\hat{\beta}}^{\text{ideal}}\), the maximum likelihood estimator of \(\boldsymbol{\beta}\) when \(\alpha_1\) and \(\alpha_2\) are known. We expect that \(\Cov (\boldsymbol{\hat{\beta}}^{\text{ideal}})\) would be small (and the eigenvalues of \( I_{\beta \beta}\) would be large) in this setting because we can estimate \(\boldsymbol{\beta}\) well due to the abundant observations in classes 1 and 2 (ensured if \(\Delta\) is not too large), so we should be able to learn the decision boundary between these classes well. If the eigenvalues of \( I_{\beta \beta}\) are indeed large, it might be reasonable to expect \(I_{\beta \beta} -   2 \frac{I_{\beta \alpha_1} I_{\beta \alpha_1}^\top}{I_{\alpha_1 \alpha_1}}  \) to be positive definite. In Sections \ref{min.eigen.bound} and \ref{min.eigen.cond.sim} in the appendix, we present more detailed analysis as well as the results of synthetic experiments that indicate that it is plausible both that \( I_{\beta \beta} -   2 \frac{I_{\beta \alpha_1} I_{\beta \alpha_1}^\top}{I_{\alpha_1 \alpha_1}} \) is positive definite and that the upper bound \eqref{delta.small.assum.fin} is reasonable.

\end{remark}

\section{Predicting Rare Events by Shrinking Towards proportional Odds (\textsc{PRESTO})}\label{sec.method}

Theorems \ref{est.known.beta} and \ref{main.cov.thm.2} suggest a path to improve estimated probabilities for a rare event that is at the end of an ordered sequence: use the more common events that come before it to improve the estimation of the decision boundary affecting the rare class. In practice, however, the proportional odds model assumption is strong and unlikely to hold in many settings. \textsc{PRESTO} allows for this assumption to be relaxed; instead of assuming the \(\boldsymbol{\beta}\) vectors governing the decision boundaries are identical, we assume they are in general different, but with differences that are (approximately) sparse. 

One concrete model to motivate this is a relaxation of \eqref{prop.odds.utilty} along the lines of \eqref{prop_odds_gen}. Suppose that \(U_{1} := U\) as defined in \eqref{prop.odds.utilty} (with \(\boldsymbol{\beta}_1 = \boldsymbol{\beta}\)), and it still holds that an observation is in class 1 if \(U_1 \geq  -\alpha_1\). However, for \(k \in \{2, \ldots, K - 1 \}\), outcome \(k\) is observed if and only if \( -\alpha_{k}  \leq U_k < -\alpha_{k-  1} +  \boldsymbol{\psi}_k^\top \boldsymbol{x}\) for sparse vectors \(\boldsymbol{\psi}_2, \ldots, \boldsymbol{\psi}_{K-1} \in \mathbb{R}^p\) satisfying \(\boldsymbol{\psi}_k = \boldsymbol{\beta}_k - \boldsymbol{\beta}_{k-1}\), so \(U_k = U_{k-1} + \boldsymbol{\psi}_k^\top \boldsymbol{x}\) for \(k \in \{2, \ldots, K-1\}\). Note that this is within the scope of \eqref{prop_odds_gen}, but we assume a structure on the differing \(\boldsymbol{\beta}_k\) vectors rather than allowing for arbitrary differences.

Assuming sparse differences in adjacent \(\boldsymbol{\beta}_k\) vectors in this way suggests the following optimization problem for data \(\boldsymbol{X} = (\boldsymbol{x}_1, \ldots, \boldsymbol{x}_n )^\top\) and \(\boldsymbol{y} = (y_1, \ldots, y_n)\):
\begin{align}
 \underset{\boldsymbol{\beta}, \boldsymbol{\alpha}}{\arg \min}
&\bigg\{ - \frac{1}{n} \sum_{i=1}^n \log \bigg[  F \left(\alpha_{y_i} + \boldsymbol{\beta}_{y_i}^\top \boldsymbol{x}_i \right) \nonumber
\\ & 
 - F \left( \alpha_{y_i - 1} + \boldsymbol{\beta}_{y_i - 1}^\top \boldsymbol{x}_i  \right) \bigg]    \nonumber
 \\ &  + \lambda_n \left( \sum_{j=1}^p   \left|  \beta_{j1} \right| +  \sum_{j=1}^p  \sum_{k=2}^{K-1} \left|  \beta_{jk} - \beta_{j,k-1}\right| \right) \bigg\} 
 ,
 \label{ordinal.pen.opt}
\end{align}
where we define \(\alpha_K := \infty, \alpha_0 := - \infty\) and \(\boldsymbol{\beta}_0 := \boldsymbol{0}\). The penalties on the \(  \left|  \beta_{j1} \right| \) terms are sufficient to regularize all of the weights given the penalties on the difference terms starting from the \(\boldsymbol{\beta}_1\) vector, improving parameter estimation. Like the proportional odds model and the generalized lasso \citep{tibshirani2011solution} optimization problem, \eqref{ordinal.pen.opt} is strictly convex if and only if \( \alpha_{y_i} + \boldsymbol{\beta}_{y_i}^\top \boldsymbol{x}_i > \alpha_{y_i - 1} + \boldsymbol{\beta}_{y_i - 1}^\top \boldsymbol{x}_i \) for all \(i\) \citep{Pratt1981}. This can be violated if the decision boundaries, which are not parallel, cross in the support of \(\boldsymbol{X}\). In Section \ref{sparse.sim}, we discuss the practical issues this presents when implementing relaxed proportional odds models like \textsc{PRESTO}, and in the next section, we prove \textsc{PRESTO} is consistent relying in part on an assumption that these decision boundaries do not cross in the support of \(\boldsymbol{X}\). See Appendix \ref{sec.est.presto} for details on how we estimate \textsc{PRESTO} in practice.

\subsection{Theoretical Analysis}\label{presto.cons.sec}

In this section, we present Theorem \ref{cons.thm}, which shows that \textsc{PRESTO} is a consistent estimator of \(\boldsymbol{\beta}_1, \ldots, \boldsymbol{\beta}_{K-1}\) under suitable assumptions. Before stating Theorem \ref{cons.thm}, we present and briefly discuss some of the new assumptions we will make.

\begin{itemize}
\item \textbf{Assumption \(S(s, c)\)}: The distribution of \(y_i \mid \boldsymbol{x}_i\) is distributed according to the \textsc{PRESTO} likelihood \eqref{ordinal.pen.opt}, where the true coefficients \( \boldsymbol{\theta}_* = \left( \boldsymbol{\beta}_1^\top, \boldsymbol{\psi}_2^\top, \ldots, \boldsymbol{\psi}_{K-1}^\top  \right)^\top \in \mathbb{R}^{p(K-1)}\) are \(s\)-sparse (have \(s\) nonzero entries for a fixed \(s\) not increasing in \(n\) or \(p\)). Further, \(\lVert \boldsymbol{\theta}_* \rVert_\infty \leq c\) for a fixed \(c\).

\item \textbf{Assumption \(T(c)\)}:  For all small enough \(\rho > 0\), for all \(\boldsymbol{\theta} \in \mathbb{R}^{p(K-1)}\) with \(\lVert \boldsymbol{\theta} - \boldsymbol{\theta}_* \rVert_1 \leq \rho\) it holds that none of the decision boundaries defined by \(\boldsymbol{\theta}\) and the true \(\alpha_1, \ldots, \alpha_{K-1}\) cross in \(\mathcal{S}\). Also, \(\max_{k \in \{1, \ldots, K-1\}} \left| \alpha_k \right| \leq c\).
\end{itemize}

The fixed sparsity assumption \(S(s, c)\) is helpful theoretically and also because without it in higher dimensions it becomes increasingly difficult to have nonparallel decision boundaries that do not cross. The first part of Assumption \(T(c)\) can be interpreted to mean that none of the decision boundaries cross ``too closely" to \(\mathcal{S}\). Other than these aspects, Assumptions \(S(s, c)\) and \(T(c)\) are mild. 

\begin{theorem}\label{cons.thm}
In a setting with fixed \(K \geq 3\) and \(p = p_n \to \infty\) as \(n \to \infty\) and satisfying \(p_n \leq C_1 n^{C_2}\) for some \(C_1 > 0\) and \(C_2 \in (0, 1)\), consider estimating \textsc{PRESTO} with penalty \(\lambda_n = C_3 \log(p_n[K-1]) /n    \) for some \(C_3 > 0\). Suppose Assumption \(X(\mathbb{R}^{p_n})\) holds and there is some \(C_4 < \infty\) such that \(\sup_{\boldsymbol{x} \in \mathcal{S}} \lVert \boldsymbol{x} \rVert_\infty \leq C_4\) and Assumptions \(S(s, C_4)\) and \(T(C_4)\) hold. Assume for some fixed \(b >0\) it holds that \( \lambda_{\text{min}}^* := \min_{k \in \{1, \ldots, K\}}  \lambda_{\text{min}} \left(  \boldsymbol{\Sigma}_k  \right)  >  b\), where \(\boldsymbol{\Sigma}_k := \E \left[  \boldsymbol{x}_i \boldsymbol{x}_i^\top \mid y_i = k\right] \). Then \textsc{PRESTO} is a consistent estimator of \(  \boldsymbol{\beta}_1, \ldots, \boldsymbol{\beta}_{K-1} \).

\end{theorem}

Theorem \ref{cons.thm} shows that under fairly mild regularity conditions and a sparsity assumption in a high-dimensional setting, \textsc{PRESTO} consistently estimates all of the decision boundaries. That is, it is consistent both if the proportional odds assumption holds and in more flexible settings, where the proportional odds model is unrealistic, under sparsity. Theorem \ref{est.known.beta} suggests this should be helpful for estimating rare class probabilities. The proof of Theorem \ref{cons.thm} leverages recent theory developed for \(\ell_1\)-penalized ordinal regression \citep{Ekvall2022}.

\begin{figure}[htbp]
\begin{center}
\includegraphics[width=\linewidth]{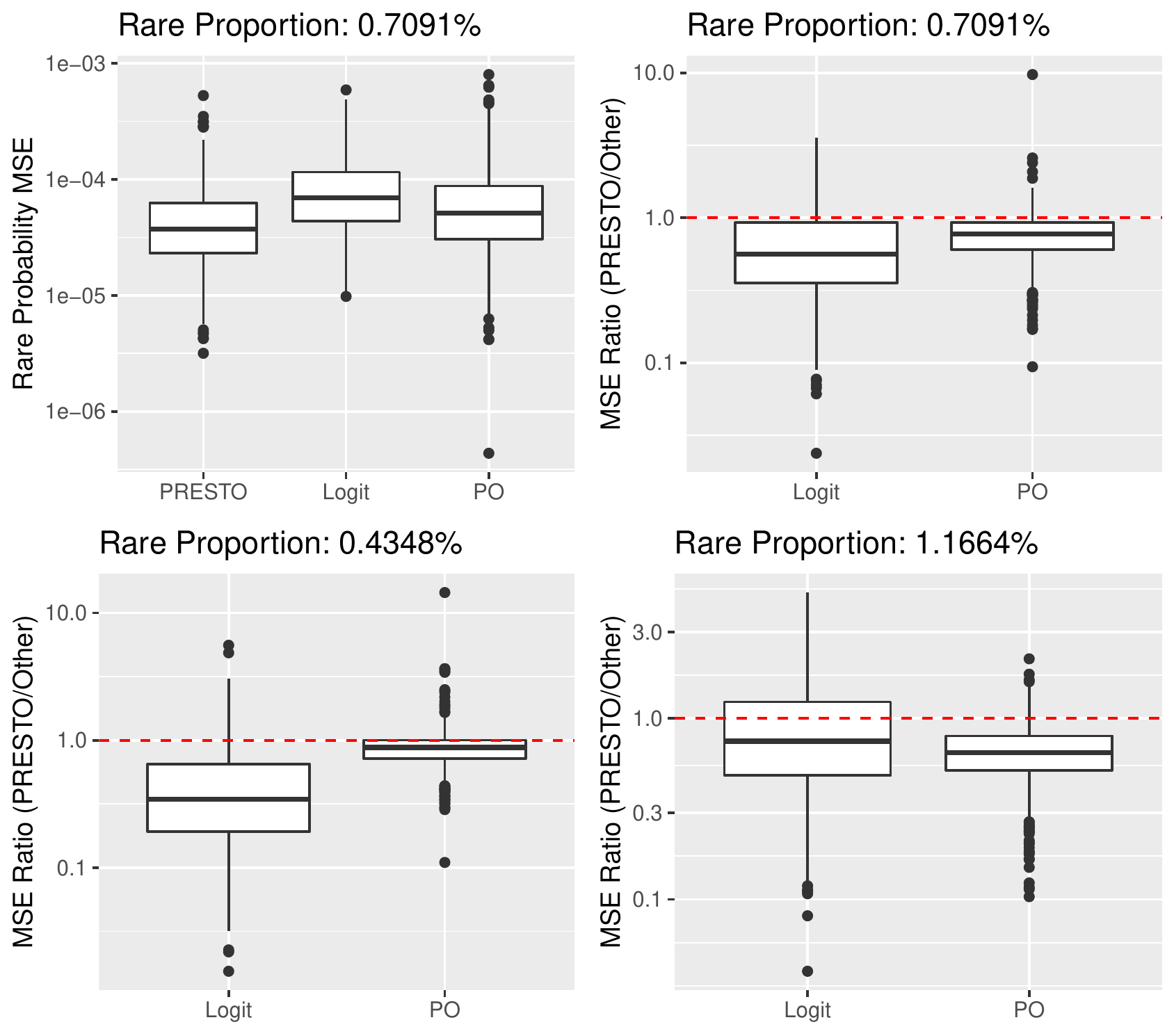}
\caption{Top left: MSE of estimated rare class probabilities for each method across all \(n = 2500\) observations, across 700 simulations, in sparse differences simulation setting of Section \ref{sparse.sim}, for intercept setting yielding rare class proportions of about \(0.71\%\) on average and sparsity \(1/2\). Remaining plots: ratios of MSE for \textsc{PRESTO} divided by MSE of each other method for each of three sets of intercepts with sparsity \(1/2\) (\textsc{PRESTO} performs better if ratio is less than 1). All plots on log scale.}
\label{fig.sparse.3}
\end{center}
\end{figure}

\begin{figure}[htbp]
\begin{center}
\includegraphics[width=\linewidth]{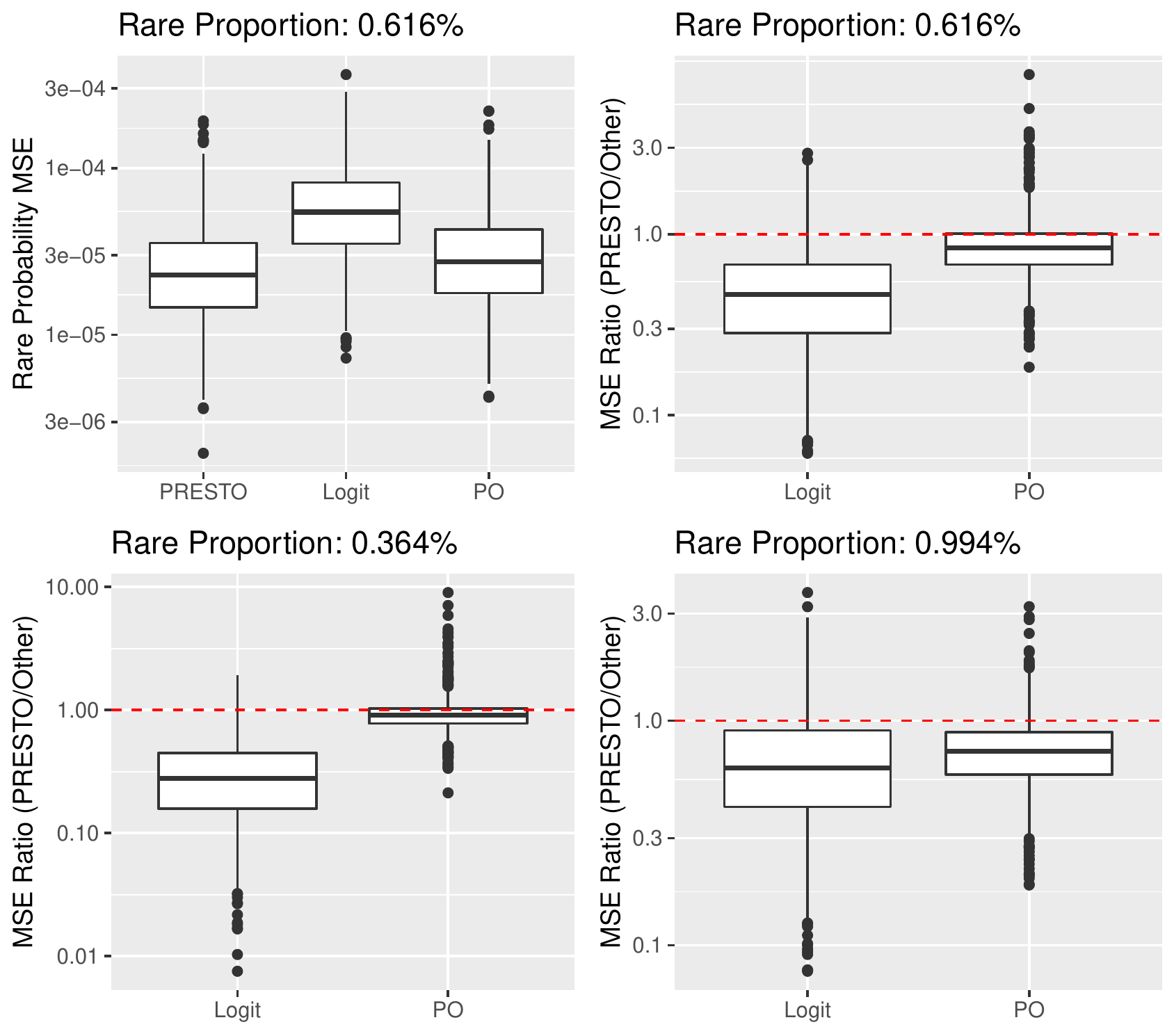}
\caption{Same plots as in Figure \ref{fig.sparse.3}, but for uniform differences synthetic experiment in Section \ref{dense.sim}.}
\label{fig.unif}
\end{center}
\end{figure}

\section{Experiments}\label{sec.sim}

To illustrate the efficacy of \textsc{PRESTO}, we conduct two synthetic experiments and also examine two real data sets. In Section \ref{sparse.sim}, we generate random \(\boldsymbol{y}\) that have conditional probabilities based on a relaxation of the proportional odds model with sparse differences between adjacent decision boundary parameter vectors, rather than parameterizing all decision boundaries with the same \(\boldsymbol{\beta}\). This setting is well-suited to \textsc{PRESTO}. In Section \ref{dense.sim}, we show that \textsc{PRESTO} also performs well in a less favorable setting, where the differences between adjacent decision boundaries are instead dense; nonetheless, \textsc{PRESTO} still outperforms logistic regression and proportional odds models. In Section \ref{real.data} we compare the performance of \textsc{PRESTO} to logistic regression and the proportional odds model at estimating rare probabilities in a real data experiment. Finally, in Section \ref{real.data.2} we conduct a second real data experiment on a data set of patients diagnosed with diabetes, where we vary the rarity of the outcome of interest. See Section \ref{sec.est.presto} of the appendix for all implementation details. The code generating all plots and tables is available at \url{https://github.com/gregfaletto/presto}.

\subsection{Synthetic Data: Sparse Differences Setting}\label{sparse.sim}

We repeat the following procedure for 700 simulations. First we generate data using \(n = 2500\), \(p = 10\), and \(K =  4\). We draw a random \(\boldsymbol{X} \in [-1, 1]^{n \times p}\), where \(X_{ij} \sim \operatorname{Uniform}(-1, 1)\) for all \(i \in \{1, \ldots, n\}\) and \(j \in \{1, \ldots, p\}\). Then \(\boldsymbol{y} \in \{1, \ldots, K\}^n\) is generated according to a relaxation of the proportional odds model; instead of \eqref{prop_odds}, we generate probabilities according to \eqref{prop_odds_gen} where the \(\boldsymbol{\beta}_k\) are generated in the following way for sparsity settings of \(\eta \in\{ 1/3, 1/2\}\): first, we generate \(\boldsymbol{\beta}_1\) by taking the vector \((0.5, \ldots, 0.5)^\top\), but setting all of the entries equal to 0 randomly with probability \(1 - \eta\) for each entry independently. Then we set \( \boldsymbol{\beta}_k =  \boldsymbol{\beta}_{k-1} + \boldsymbol{\psi}_k, \qquad k \in \{2, \ldots, K - 1\},\) where \(\boldsymbol{\psi}_k \in \mathbb{R}^p\) are iid random vectors for each \(k \in \{2, \ldots, K - 1\}\) generated according to the following distribution:
\[
\psi_{kj} = \begin{cases}
0, & \text{with probability } 1- \eta,
\\ 0.5, & \text{with probability } \eta/2,
\\ -0.5, & \text{ with probability } \eta/2,
\end{cases} \ \ \  j \in \{1, \ldots, p\}.
\]
We consider three possible sets of intercepts: \(\boldsymbol{\alpha} =  (0,3,5), (0, 3.5, 5.5)\), and \((0, 4, 6)\), so that the first two categories are common and the remaining categories are rare. The final rare class is the one of interest; in the three settings, the average proportions of observations falling in the rare class are \(1.00\%\), \(0.62\%\), and \(0.37\%\), respectively, for the \(\eta = 1/3\) setting and \(1.17\%\), \(0.71\%\), and \(0.43\%\) for the \(\eta = 1/2\) setting.

The fact that the decision boundaries may cross in the support of \(\boldsymbol{X}\), which would mean that for such \(\boldsymbol{x}\) some class probabilities are defined to be negative, puts practical limits on the magnitude of \(\boldsymbol{\psi}_k\) in simulations. (See Section 3.6.1 of \citealt{agresti2010analysis} for a discussion of this point.) Also, for this reason, in each simulation we check whether or not the conditional probabilities are positive for each class for every sampled \(\boldsymbol{x}\); if not, we generate new \(\boldsymbol{\psi}_2, \ldots, \boldsymbol{\psi}_{K-1}\) for a limited number of iterations, ending the simulation study in failure if no suitable \(\boldsymbol{\psi}_k\) can be found in a reasonable number of attempts. The parameters we used generated positive probabilities for all observations across all simulations. 

We then estimate a model for each method; for logistic regression, we estimate the binary classification problem of whether or not each observation is in class \(K\), and for proportional odds and \textsc{PRESTO}, we fit a full model on all \(K\) responses. For \textsc{PRESTO}, we use 5-fold cross-validation to choose a value of \(\lambda_n\) among 20 choices, selecting the \(\lambda_n\) with the best out-of-fold Brier score (other metrics, like negative log likelihood, failed because some values of \(\lambda_n\) in some folds resulted in models yielding negative probabilities, so these other metrics were undefined). The 20 candidate values of \(\lambda_n\) are generated in the following way: the largest \(\lambda_n\) value, \(\lambda_n^{(20)}\), is the smallest \(\lambda_n\) for which all of the estimated sparse differences equal 0; the smallest \(\lambda_n\) value is set to \(\lambda_n^{(1)} = 0.01 \cdot \lambda_n^{(20)}\), and the remaining \(\lambda_n\) values are generated at equal intervals on a logarithmic scale between these two values.

Each of these models yields estimated probabilities that each observation lies in class \(K\). In the final step of each simulation run, we compute the mean squared error of these estimated probabilities for each method.

In Figure \ref{fig.sparse.3}, we show boxplots of the empirical mean squared errors for each method in the setting where the rare class is observed in \(0.71\%\) of observations when \(\eta = 1/2\). In order to see how the methods compare pairwise on each simulation, we also show boxplots of the ratio between the mean squared error of \textsc{PRESTO} and the other two methods in each of the three simulation settings. We also conduct one-tailed paired \(t\)-tests of the alternative hypothesis that the mean MSE for \textsc{PRESTO} is lower than each of the competitor methods in each setting; all 12 of the \(p\)-values (provided in Table \ref{sparse.wilcox.tab} of Appendix \ref{sim.output}) are below \(0.01\). Finally, in Appendix \ref{sim.output} we also provide   the means and standard errors for the MSE of each method in each simulation setting in Table \ref{sparse.sum.stats.tab}, as well as boxplots like the one in the top left corner of Figure \ref{fig.sparse.3} for the other two intercept settings and all boxplots for the \(\eta = 1/3\) setting.

We see that \textsc{PRESTO} typically estimates these rare probabilities better than logistic regression, which despite being correctly specified struggles with class imbalance and does not draw strength from estimating the easier decision boundary between classes 1 and 2, and the proportional odds model, whose assumptions are not satisfied in this setting.

\subsection{Synthetic Data: Dense Differences Setting}\label{dense.sim}

In real data sets the differences between adjacent decision boundary parameter vectors may not always be exactly sparse, so we conduct another synthetic experiment in the same way as in Section \ref{sparse.sim}, except \(\boldsymbol{\beta}_{1j} \sim \operatorname{Uniform}(-.5, .5)\) and each \(\psi_{kj} \sim \operatorname{Uniform}(-.5, .5)\), iid across \(j \in \{1, \ldots, p\}\) and \(k \in \{2, \ldots, K - 1\}\). We also add an extra intercept setting of \((0, 2.5, 4.5)\). This yields average rare class proportions of \(0.99\%\), \(0.62\%\), and \(0.36\%\) using the same intercepts as the experiments in Section \ref{sparse.sim} and \(1.60\%\) in the new intercept setting. The uniformly distributed differences can be considered ``approximately" sparse in the sense that while no deviations will exactly equal 0, some will be large and important to estimate, and some will be essentially negligible.

Figure \ref{fig.unif} and Table \ref{unif.wilcox.tab} summarize the results, along with additional figures and tables in Appendix \ref{sim.output}. We again see that \textsc{PRESTO} outperforms both competitor methods by statistically significant margins.

\begin{table}[!htbp] \centering 
  \caption{Calculated \(p\)-values for one-tailed paired \(t\)-tests for uniform differences simulation setting of Section \ref{dense.sim} testing the alternative hypothesis that PRESTO's rare probability MSE is less than each competitor method in each rarity setting. (Statistically significant \(p\)-values indicate better performance for PRESTO).} 
  \label{unif.wilcox.tab} 
\begin{tabular}{@{\extracolsep{5pt}} ccc} 
\\[-1.8ex]\hline 
\hline \\[-1.8ex] 
 Rare Class Proportion & Logit \(p\)-value & PO \(p\)-value \\ 
\hline \\[-1.8ex] 
 1.6\% & \(<\) 1e-10 & \(<\) 1e-10  \\ 
0.99\% & \(<\) 1e-10  & \(<\) 1e-10  \\ 
0.62\% & \(<\) 1e-10  &\(<\) 1e-10  \\ 
0.36\% & \(<\) 1e-10  & 0.000242 \\ 
\hline \\[-1.8ex] 
\end{tabular} 
\end{table}

\subsection{Real Data Experiment 1: Soup Tasting}\label{real.data}

We conduct a real data experiment using the \texttt{soup} data set from the R \texttt{ordinal} package \citep{ordinal}. The data come from a study \citep{Christensen2011} of participants who tasted soups and responded whether they thought each soup was a reference product they had previously been familiarized with or a new test product. The respondents also stated how sure they were in their response on a three-level scale, yielding a total of \(K = 6\) possible ordered outcomes for \(n = 1847\) observations. The outcome of interest corresponds to the respondent being sure the tasted soup was the reference and is observed in 228 observations (about \(12\%\) of the total). All of the features are categorical, and after one-hot encoding we have \(p = 22\) binary features related to the soup, the respondent, and the testing environment\footnote{The categorical predictors \texttt{PRODID} and \texttt{RESP} are omitted because in some splits not all levels of these features are observed in the training set, making it impossible to estimate parameters for these features.}. This may be a promising setting for \textsc{PRESTO} because, while the responses have a well-defined ordering, it's plausible that different features could have varying impacts at different levels of respondent certainty.

We complete the following procedure 350 times: first, we randomly split the data into training (\(90\%\) of the data) and test (\(10\%\)) sets. We estimate models using PRESTO, logistic regression, and the proportional odds model on the training data and evaluate on the test set. 

We are interested in the accuracy of the rare class probabilities, but we can't evaluate rare probability MSE directly since we don't observe the true probabilities. Brier score could be a reasonable proxy, but it is known to be a poor metric in the presence of class imbalance \citep{benedetti2010scoring}. Instead we estimate rare probability MSE using the following procedure. For each method, we sort the estimated test set rare class probabilities in ascending order and assign the observations into 10 bins: the first \(1/10\) observations go in the first bin, and so on. Then we estimate the mean squared error of the estimated probabilities by \(\frac{1}{n} \sum_{i=1}^n ( \hat{\pi}_1^{(i)} - o_{b(i)}  )^2\), where \(\hat{\pi}_1^{(i)}\) is the estimated rare class probability for observation \(i\) and \(o_{b(i)}\) is the observed rare class proportion in the bin containing observation \(i\). This is similar to \textit{expected calibration error} \citep{naeini2015obtaining}, though we use squared error rather than absolute error. 10 equal frequency bins follows the default of the R \texttt{CalibratR} package that implements expected calibration error \citep{10.1093/bioinformatics/bty984}.

By this metric, the mean error for \textsc{PRESTO} is 0.0096, 0.0157 for logistic regression and 0.0135 for proportional odds. Figure \ref{fig.soup} displays boxplots of the results as in the synthetic experiments which indicate that \textsc{PRESTO} typically outperforms the other methods. We do not report \(p\)-values or standard errors since the observed samples are dependent (random splits of the same data set).

\begin{figure}[htbp]
\begin{center}
\includegraphics[width=\linewidth]{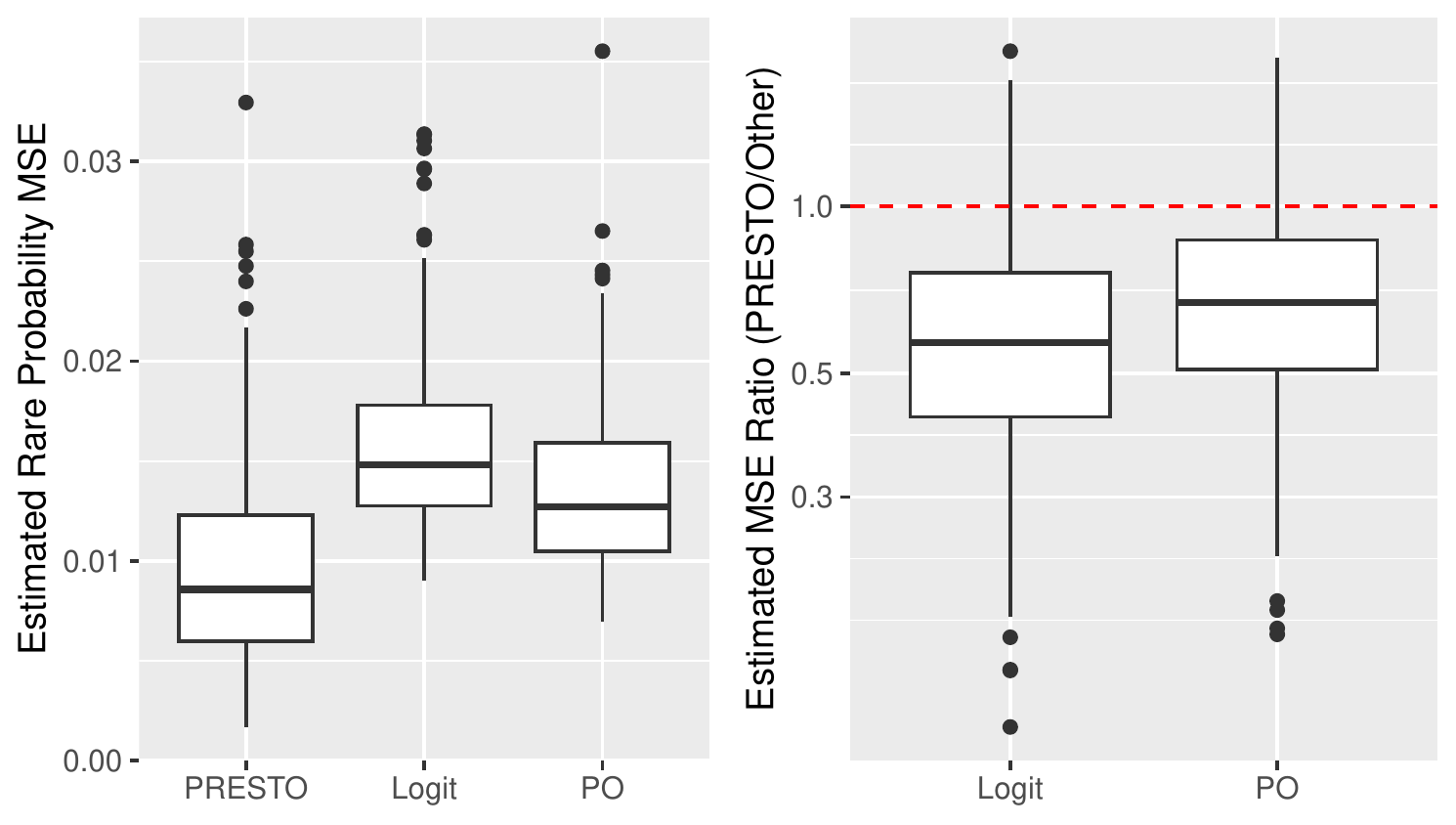}
\caption{Left: Estimated MSEs of estimated rare class probabilities for each method across 350 random draws of training and test sets in real data experiment from Section \ref{real.data}. Right: ratios of estimated MSE for \textsc{PRESTO} divided by MSE of each other method (\textsc{PRESTO} performs better if ratio is less than 1).}
\label{fig.soup}
\end{center}
\end{figure}

\subsection{Real Data Experiment 2: Diabetes}\label{real.data.2}

We present another real data experiment using the data set \texttt{PreDiabetes} from the R \texttt{MLDataR} package \citep{mldatar}. This data set contains \(n = 3059\) observations of patients who were eventually diagnosed with diabetes. Each observation consists of the age at which the patient was diagnosed with prediabetes and diabetes as well as \(p =5\) covariates. Given an age \(a\), we form an ordinal variable based on the patient's status of non-diabetic, prediabetic, or diabetic at age \(a -1\). We do this for ages \(a \in \{30, 35, 40, \ldots, 65\}\). The number of patients diagnosed with diabetes increases with \(a\), so varying \(a\) allows us to change the rarity of the rarest class in a natural way. \(0.92\%\) of patients in the data were diagnosed with diabetes before age \(a = 30\), and \(50.93\%\) of the patients were diagnosed with diabetes before age \(a = 65\).

We use PRESTO, logistic regression, and the proportional odds model to estimate the probability that each patient was diagnosed with diabetes before age \(a\) for each \(a\). Much like our soup tasting data application, in each setting we take repeated random splits of the data, using \(90\%\) of the data selected at random for training and \(10\%\) for testing. In each iteration we again evaluate each method on the test data using the same estimator for mean squared error of the estimated rare class probabilities. We repeat this procedure 49 times in each of the 8 settings. 

\begin{figure}[htbp]
\begin{center}
\includegraphics{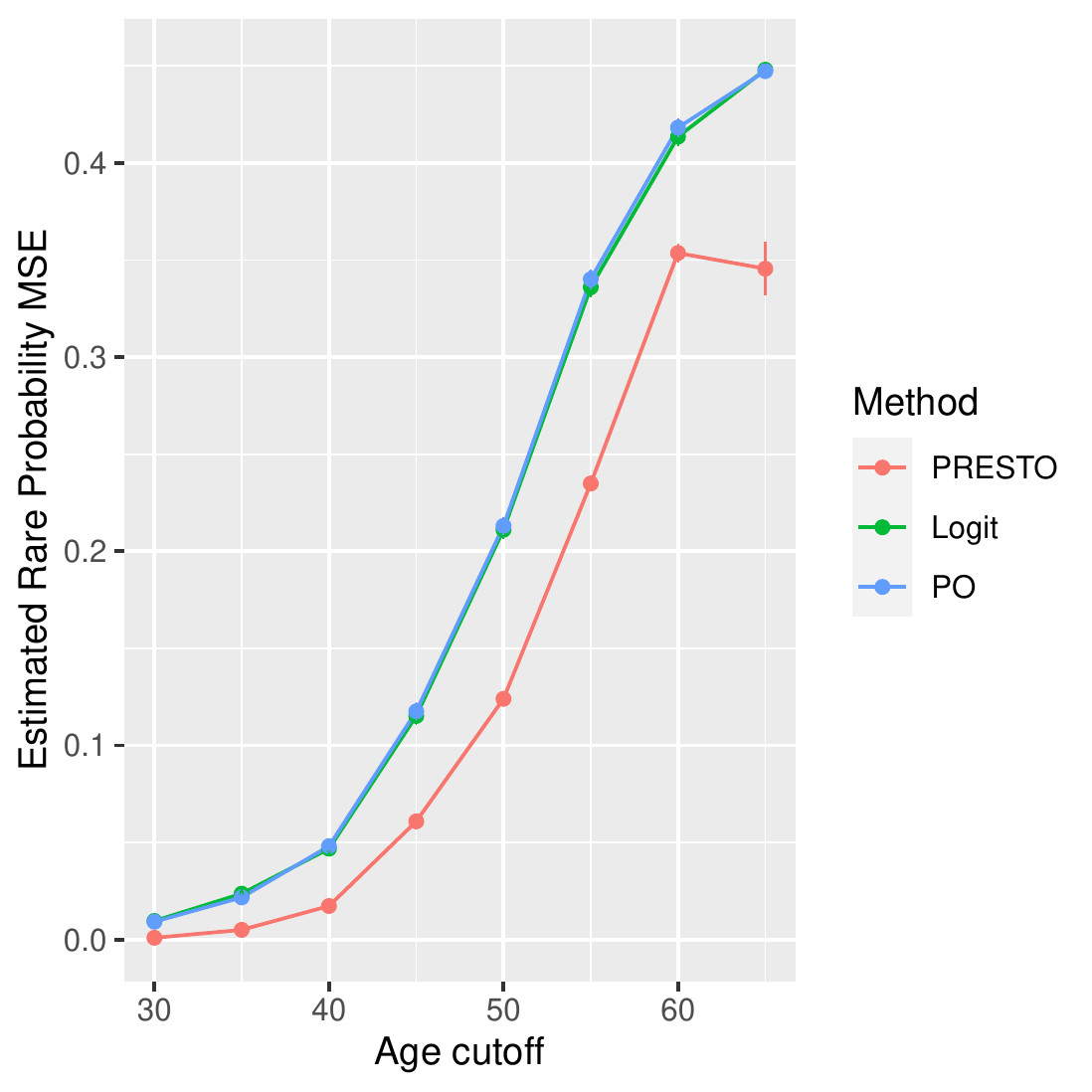}
\caption{Estimated MSEs of estimated rare class probabilities for each method and each age cutoff across 49 random draws of training and test sets in real data experiment from Section \ref{real.data.2}.}
\label{fig.diabetes}
\end{center}
\end{figure}

\begin{table}
\caption{Estimated rare class MSE for each method at each age cutoff in prediabetes real data experiment from Section \ref{real.data.2}.}
\label{tab.diabetes}
\centering
\begin{tabular}[t]{l|l|l|l}
\hline
Age cutoff  & PRESTO & Logit & PO\\
\hline
30 & 0.000943 & 0.009609 & 0.009217\\
\hline
35 & 0.005013 & 0.023658 & 0.021740\\
\hline
40 & 0.017307 & 0.046828 & 0.048123\\
\hline
45 & 0.060896 & 0.115189 & 0.117525\\
\hline
50 & 0.124000 & 0.211083 & 0.213059\\
\hline
55 & 0.234906 & 0.336009 & 0.340130\\
\hline
60 & 0.353615 & 0.413534 & 0.418179\\
\hline
65 & 0.345535 & 0.448015 & 0.447290\\
\hline
\end{tabular}
\end{table}

We display the results in a plot in Figure \ref{fig.diabetes}. We also provide the mean MSEs for each method at each age cutoff in Table \ref{tab.diabetes}. We see that PRESTO outperforms both logistic regression and the proportional odds model in all of these settings. (For age cutoffs \(a=29\) and below we were unable to estimate the proportional odds model on all subsamples because of the difficulty of having at least one observation from each class in both the training and test sets for 49 random draws.) PRESTO seems to outperform the other methods at all class rarities, and the absolute performance gap increases as the rare class becomes less rare.

\section{Conclusion}\label{sec.conclusion}

By leveraging data from earlier decision boundaries, but relaxing the rigid proportional odds assumption, \textsc{PRESTO} can substantially improve estimation of the probability of rare events, even when the assumption of sparse differences between adjacent decision boundary weight vectors does not exactly hold. Future work could explore \(\ell_1\) penalties for the coefficients themselves, not just the differences between the coefficients, to allow for simultaneous feature selection and model estimation. Inference for \textsc{PRESTO} could also be possible by extending the method for exact post-selection inference for the generalized lasso path by \citet{hyun2018exact}, or similar work on the fused lasso by \citet{chen2022more}, to our generalized linear model setting. Future work could also explore the empirical performance of \textsc{PRESTO} in even more depth, perhaps by using large-scale real world data sets like those used in \citet{10.1145/2783258.2788578}.

There are other possible extensions that could improve estimation. For example, we set the first decision boundary as the one that is directly penalized, with differences from this boundary assumed to be sparse. This makes sense if the classes become increasingly rare and the first decision boundary is the most balanced. However, it may make more sense to directly penalize whichever decision boundary has the best balance of observed responses on each side. Penalizing the differences from this boundary might improve estimation since this decision boundary ought to be the easiest to estimate. This could improve estimation in settings like the real data experiment from Section \ref{real.data} where the most balanced decision boundary is closer to the center of the responses. 

Also, in cases where the final categories are very rare, a better bias/variance tradeoff might be achieved by reimposing the proportional odds assumption, imposing an exact equality constraint for the last few decision boundaries. In these settings, data might be too rare to hope for better estimation by relaxing the proportional odds assumption even with regularization. 

Lastly, in Section \ref{sec.est.presto} of the appendix we discuss possible faster approaches than the one used in the present work for solving the PRESTO optimization problem.

\bibliography{mybib2fin}

\begin{thebibliography}{53}
\providecommand{\natexlab}[1]{#1}
\providecommand{\url}[1]{\texttt{#1}}
\expandafter\ifx\csname urlstyle\endcsname\relax
  \providecommand{\doi}[1]{doi: #1}\else
  \providecommand{\doi}{doi: \begingroup \urlstyle{rm}\Url}\fi

\bibitem[Agresti(2010)]{agresti2010analysis}
A.~Agresti.
\newblock \emph{Analysis of ordinal categorical data}, volume 656.
\newblock John Wiley \& Sons, 2010.

\bibitem[Armstrong and Sloan(1989)]{armstrong1989ordinal}
B.~G. Armstrong and M.~Sloan.
\newblock Ordinal regression models for epidemiologic data.
\newblock \emph{American Journal of Epidemiology}, 129\penalty0 (1):\penalty0
  191--204, 1989.

\bibitem[Arnold and Tibshirani(2016)]{Arnold2016}
T.~B. Arnold and R.~J. Tibshirani.
\newblock {Efficient Implementations of the Generalized Lasso Dual Path
  Algorithm}.
\newblock \emph{Journal of Computational and Graphical Statistics}, 25\penalty0
  (1):\penalty0 1--27, 2016.
\newblock ISSN 15372715.
\newblock \doi{10.1080/10618600.2015.1008638}.

\bibitem[Benedetti(2010)]{benedetti2010scoring}
R.~Benedetti.
\newblock Scoring rules for forecast verification.
\newblock \emph{Monthly Weather Review}, 138\penalty0 (1):\penalty0 203--211,
  2010.

\bibitem[Bickel et~al.(2009)Bickel, Ritov, and Tsybakov]{PeterBickel2009}
P.~J. Bickel, Y.~Ritov, and A.~B. Tsybakov.
\newblock {Simultaneous Analysis of LASSO and Dantzig Selector}.
\newblock \emph{The Annals of Statistics}, 37\penalty0 (4):\penalty0
  1705--1732, 2009.
\newblock \doi{10.1214/08-AOS620}.
\newblock URL
  \url{https://projecteuclid-org.libproxy1.usc.edu/download/pdfview_1/euclid.aos/1245332830}.

\bibitem[Brant(1990)]{brant1990assessing}
R.~Brant.
\newblock Assessing proportionality in the proportional odds model for ordinal
  logistic regression.
\newblock \emph{Biometrics}, pages 1171--1178, 1990.

\bibitem[Cameron and Trivedi(2005)]{cameron2005microeconometrics}
A.~C. Cameron and P.~K. Trivedi.
\newblock \emph{Microeconometrics: methods and applications}.
\newblock Cambridge university press, 2005.

\bibitem[Casella and Berger(2021)]{casella2021statistical}
G.~Casella and R.~L. Berger.
\newblock \emph{Statistical inference}.
\newblock Cengage Learning, 2021.

\bibitem[Chawla et~al.(2002)Chawla, Bowyer, Hall, and
  Kegelmeyer]{chawla2002smote}
N.~V. Chawla, K.~W. Bowyer, L.~O. Hall, and W.~P. Kegelmeyer.
\newblock Smote: synthetic minority over-sampling technique.
\newblock \emph{Journal of artificial intelligence research}, 16:\penalty0
  321--357, 2002.

\bibitem[Chen et~al.(2022)Chen, Jewell, and Witten]{chen2022more}
Y.~Chen, S.~Jewell, and D.~Witten.
\newblock More powerful selective inference for the graph fused lasso.
\newblock \emph{Journal of Computational and Graphical Statistics}, pages
  1--11, 2022.

\bibitem[Christensen et~al.(2011)Christensen, Cleaver, and
  Brockhoff]{Christensen2011}
R.~H.~B. Christensen, G.~Cleaver, and P.~B. Brockhoff.
\newblock {Statistical and Thurstonian models for the A-not A protocol with and
  without sureness}.
\newblock \emph{Food Quality and Preference}, 22\penalty0 (6):\penalty0
  542--549, 2011.
\newblock ISSN 09503293.
\newblock \doi{10.1016/j.foodqual.2011.03.003}.
\newblock URL \url{http://dx.doi.org/10.1016/j.foodqual.2011.03.003}.

\bibitem[Cordeiro and McCullagh(1991)]{Cordeiro1991}
G.~M. Cordeiro and P.~McCullagh.
\newblock {Bias Correction in Generalized Linear Models}.
\newblock \emph{Journal of the Royal Statistical Society: Series B
  (Methodological)}, 53\penalty0 (3):\penalty0 629--643, 1991.
\newblock \doi{10.1111/j.2517-6161.1991.tb01852.x}.

\bibitem[Duncan and Elkan(2015)]{10.1145/2783258.2788578}
B.~A. Duncan and C.~P. Elkan.
\newblock Probabilistic modeling of a sales funnel to prioritize leads.
\newblock In \emph{Proceedings of the 21th ACM SIGKDD International Conference
  on Knowledge Discovery and Data Mining}, KDD '15, pages 1751--1758, New York,
  NY, USA, 2015. Association for Computing Machinery.
\newblock ISBN 9781450336642.
\newblock \doi{10.1145/2783258.2788578}.
\newblock URL \url{https://doi.org/10.1145/2783258.2788578}.

\bibitem[Ekvall and Bottai(2022)]{Ekvall2022}
K.~Ekvall and M.~Bottai.
\newblock {Concave likelihood‐based regression with finite‐support response
  variables}.
\newblock \emph{Biometrics}, \penalty0 (March):\penalty0 1--12, 2022.
\newblock ISSN 0006-341X.
\newblock \doi{10.1111/biom.13760}.

\bibitem[Fern{\'a}ndez et~al.(2018)Fern{\'a}ndez, Garcia, Herrera, and
  Chawla]{fernandez2018smote}
A.~Fern{\'a}ndez, S.~Garcia, F.~Herrera, and N.~V. Chawla.
\newblock Smote for learning from imbalanced data: progress and challenges,
  marking the 15-year anniversary.
\newblock \emph{Journal of artificial intelligence research}, 61:\penalty0
  863--905, 2018.

\bibitem[Greene(2012)]{Greene2012Econometric}
W.~H. Greene.
\newblock \emph{Econometric Analysis}.
\newblock Pearson Education, 7th edition, 2012.

\bibitem[Hansen(2022)]{hansen2022econometrics}
B.~Hansen.
\newblock \emph{Econometrics}.
\newblock Princeton University Press, 2022.
\newblock ISBN 9780691235899.
\newblock URL \url{https://books.google.com/books?id=Pte7zgEACAAJ}.

\bibitem[Hastie et~al.(2015)Hastie, Tibshirani, and
  Wainwright]{hastie2015statistical}
T.~Hastie, R.~Tibshirani, and M.~Wainwright.
\newblock Statistical learning with sparsity.
\newblock \emph{Monographs on statistics and applied probability},
  143:\penalty0 143, 2015.

\bibitem[He et~al.(2014)He, Pan, Jin, Xu, Liu, Xu, Shi, Atallah, Herbrich,
  Bowers, et~al.]{he2014practical}
X.~He, J.~Pan, O.~Jin, T.~Xu, B.~Liu, T.~Xu, Y.~Shi, A.~Atallah, R.~Herbrich,
  S.~Bowers, et~al.
\newblock Practical lessons from predicting clicks on ads at facebook.
\newblock In \emph{Proceedings of the eighth international workshop on data
  mining for online advertising}, pages 1--9, 2014.

\bibitem[H{\"o}fling et~al.(2010)H{\"o}fling, Binder, and
  Schumacher]{hofling2010coordinate}
H.~H{\"o}fling, H.~Binder, and M.~Schumacher.
\newblock A coordinate-wise optimization algorithm for the fused lasso.
\newblock \emph{arXiv preprint arXiv:1011.6409}, 2010.

\bibitem[Horn and Johnson(2012)]{horn_johnson_2012}
R.~A. Horn and C.~R. Johnson.
\newblock \emph{Matrix Analysis}.
\newblock Cambridge University Press, 2 edition, 2012.
\newblock \doi{10.1017/9781139020411}.

\bibitem[Hutson et~al.(2022)Hutson, Laldin, and Vel{\'a}squez]{mldatar}
G.~Hutson, A.~Laldin, and I.~Vel{\'a}squez.
\newblock \emph{{MLDataR: Collection of Machine Learning Datasets for
  Supervised Machine Learning}}, 2022.
\newblock URL \url{https://CRAN.R-project.org/package=MLDataR}.
\newblock R package version 0.1.3.

\bibitem[Hyun et~al.(2018)Hyun, G'Sell, and Tibshirani]{hyun2018exact}
S.~Hyun, M.~G'Sell, and R.~J. Tibshirani.
\newblock Exact post-selection inference for the generalized lasso path.
\newblock \emph{Electronic Journal of Statistics}, 12\penalty0 (1):\penalty0
  1053--1097, 2018.

\bibitem[Johnson and Khoshgoftaar(2019)]{johnson2019}
J.~M. Johnson and T.~M. Khoshgoftaar.
\newblock Survey on deep learning with class imbalance.
\newblock \emph{Journal of Big Data}, 6\penalty0 (1):\penalty0 27, 2019.
\newblock \doi{10.1186/s40537-019-0192-5}.
\newblock URL \url{https://doi.org/10.1186/s40537-019-0192-5}.

\bibitem[Ko et~al.(2019)Ko, Yu, and Won]{ko2019easily}
S.~Ko, D.~Yu, and J.-H. Won.
\newblock Easily parallelizable and distributable class of algorithms for
  structured sparsity, with optimal acceleration.
\newblock \emph{Journal of Computational and Graphical Statistics}, 28\penalty0
  (4):\penalty0 821--833, 2019.

\bibitem[Lehmann(1999)]{lehmann1999elements}
E.~L. Lehmann.
\newblock \emph{Elements of large-sample theory}.
\newblock Springer, 1999.

\bibitem[McCullagh(1980)]{mccullagh1980regression}
P.~McCullagh.
\newblock Regression models for ordinal data.
\newblock \emph{Journal of the Royal Statistical Society: Series B
  (Methodological)}, 42\penalty0 (2):\penalty0 109--127, 1980.

\bibitem[Naeini et~al.(2015)Naeini, Cooper, and
  Hauskrecht]{naeini2015obtaining}
M.~P. Naeini, G.~Cooper, and M.~Hauskrecht.
\newblock Obtaining well calibrated probabilities using bayesian binning.
\newblock In \emph{Twenty-Ninth AAAI Conference on Artificial Intelligence},
  2015.

\bibitem[Norris et~al.(2006)Norris, Ghali, Saunders, Brant, Galbraith, Faris,
  Knudtson, Investigators, et~al.]{norris2006ordinal}
C.~M. Norris, W.~A. Ghali, L.~D. Saunders, R.~Brant, D.~Galbraith, P.~Faris,
  M.~L. Knudtson, A.~Investigators, et~al.
\newblock Ordinal regression model and the linear regression model were
  superior to the logistic regression models.
\newblock \emph{Journal of clinical epidemiology}, 59\penalty0 (5):\penalty0
  448--456, 2006.

\bibitem[Owen(2007)]{owen2007infinitely}
A.~B. Owen.
\newblock Infinitely imbalanced logistic regression.
\newblock \emph{Journal of Machine Learning Research}, 8\penalty0 (4), 2007.

\bibitem[Peterson and Harrell~Jr(1990)]{peterson1990partial}
B.~Peterson and F.~E. Harrell~Jr.
\newblock Partial proportional odds models for ordinal response variables.
\newblock \emph{Journal of the Royal Statistical Society: Series C (Applied
  Statistics)}, 39\penalty0 (2):\penalty0 205--217, 1990.

\bibitem[P{\"o}{\ss}necker and Tutz(2016)]{epub26912}
W.~P{\"o}{\ss}necker and G.~Tutz.
\newblock A general framework for the selection of effect type in ordinal
  regression, 2016.
\newblock URL
  \url{http://nbn-resolving.de/urn/resolver.pl?urn=nbn:de:bvb:19-epub-26912-0}.

\bibitem[Pratt(1981)]{Pratt1981}
J.~W. Pratt.
\newblock {Concavity of the log likelihood}.
\newblock \emph{Journal of the American Statistical Association}, 76\penalty0
  (373):\penalty0 103--106, 1981.
\newblock ISSN 1537274X.
\newblock \doi{10.1080/01621459.1981.10477613}.

\bibitem[{R. H. B. Christensen}(2019)]{ordinal}
{R. H. B. Christensen}.
\newblock {ordinal---Regression Models for Ordinal Data}, 2019.
\newblock URL \url{https://CRAN.R-project.org/package=ordinal}.

\bibitem[Schwarz and Heider(2018)]{10.1093/bioinformatics/bty984}
J.~Schwarz and D.~Heider.
\newblock {GUESS: projecting machine learning scores to well-calibrated
  probability estimates for clinical decision-making}.
\newblock \emph{Bioinformatics}, 35\penalty0 (14):\penalty0 2458--2465, 11
  2018.
\newblock ISSN 1367-4803.
\newblock \doi{10.1093/bioinformatics/bty984}.
\newblock URL \url{https://doi.org/10.1093/bioinformatics/bty984}.

\bibitem[Serfling(1980)]{serfling1980}
R.~Serfling.
\newblock \emph{Approximation theorems of mathematical statistics}.
\newblock Wiley series in probability and mathematical statistics : Probability
  and mathematical statistics. Wiley, New York, NY [u.a.], [nachdr.] edition,
  1980.
\newblock ISBN 0471024031.
\newblock URL
  \url{http://gso.gbv.de/DB=2.1/CMD?ACT=SRCHA&SRT=YOP&IKT=1016&TRM=ppn+024353353&sourceid=fbw_bibsonomy}.

\bibitem[Tibshirani et~al.(2005)Tibshirani, Saunders, Rosset, Zhu, and
  Knight]{tibshirani2005sparsity}
R.~Tibshirani, M.~Saunders, S.~Rosset, J.~Zhu, and K.~Knight.
\newblock Sparsity and smoothness via the fused lasso.
\newblock \emph{Journal of the Royal Statistical Society: Series B (Statistical
  Methodology)}, 67\penalty0 (1):\penalty0 91--108, 2005.

\bibitem[Tibshirani and Taylor(2011)]{tibshirani2011solution}
R.~J. Tibshirani and J.~Taylor.
\newblock The solution path of the generalized lasso.
\newblock \emph{The annals of statistics}, 39\penalty0 (3):\penalty0
  1335--1371, 2011.

\bibitem[Tutz and Gertheiss(2016)]{tutz2016regularized}
G.~Tutz and J.~Gertheiss.
\newblock Regularized regression for categorical data.
\newblock \emph{Statistical Modelling}, 16\penalty0 (3):\penalty0 161--200,
  2016.

\bibitem[Ugba et~al.(2021)Ugba, M{\"o}rlein, and Gertheiss]{ugba2021smoothing}
E.~R. Ugba, D.~M{\"o}rlein, and J.~Gertheiss.
\newblock Smoothing in ordinal regression: An application to sensory data.
\newblock \emph{Stats}, 4\penalty0 (3):\penalty0 616--633, 2021.

\bibitem[van~der Vaart(2000)]{van2000asymptotic}
A.~van~der Vaart.
\newblock \emph{Asymptotic Statistics}.
\newblock Asymptotic Statistics. Cambridge University Press, 2000.
\newblock ISBN 9780521784504.
\newblock URL \url{https://books.google.com/books?id=UEuQEM5RjWgC}.

\bibitem[Vershynin(2012)]{vershynin_2012}
R.~Vershynin.
\newblock \emph{Introduction to the non-asymptotic analysis of random
  matrices}, pages 210--268.
\newblock Cambridge University Press, 2012.
\newblock \doi{10.1017/CBO9780511794308.006}.

\bibitem[Vershynin(2018)]{vershynin2018high}
R.~Vershynin.
\newblock \emph{High-Dimensional Probability: An Introduction with Applications
  in Data Science}.
\newblock Cambridge Series in Statistical and Probabilistic Mathematics.
  Cambridge University Press, 2018.
\newblock ISBN 9781108415194.
\newblock URL \url{https://books.google.com/books?id=J-VjswEACAAJ}.

\bibitem[Viallon et~al.(2013)Viallon, Lambert-Lacroix, H{\"o}fling, and
  Picard]{viallon2013adaptive}
V.~Viallon, S.~Lambert-Lacroix, H.~H{\"o}fling, and F.~Picard.
\newblock Adaptive generalized fused-lasso: Asymptotic properties and
  applications.
\newblock 2013.

\bibitem[Viallon et~al.(2016)Viallon, Lambert-Lacroix, Hoefling, and
  Picard]{viallon2016robustness}
V.~Viallon, S.~Lambert-Lacroix, H.~Hoefling, and F.~Picard.
\newblock On the robustness of the generalized fused lasso to prior
  specifications.
\newblock \emph{Statistics and Computing}, 26\penalty0 (1-2):\penalty0
  285--301, 2016.

\bibitem[Von~Wachter et~al.(2019)Von~Wachter, Bertrand, Pollack, Rountree, and
  Blackwell]{von2019predicting}
T.~Von~Wachter, M.~Bertrand, H.~Pollack, J.~Rountree, and B.~Blackwell.
\newblock Predicting and preventing homelessness in los angeles.
\newblock \emph{California Policy Lab and University of Chicago Poverty Lab},
  2019.

\bibitem[Wooldridge(2010)]{wooldridge2010econometric}
J.~M. Wooldridge.
\newblock \emph{Econometric analysis of cross section and panel data}.
\newblock MIT press, 2010.

\bibitem[Wurm et~al.(2017)Wurm, Rathouz, and Hanlon]{wurm2017regularized}
M.~J. Wurm, P.~J. Rathouz, and B.~M. Hanlon.
\newblock Regularized ordinal regression and the ordinalnet r package.
\newblock \emph{arXiv preprint arXiv:1706.05003}, 2017.

\bibitem[Wurm et~al.(2021)Wurm, Rathouz, and Hanlon]{ordnet}
M.~J. Wurm, P.~J. Rathouz, and B.~M. Hanlon.
\newblock Regularized ordinal regression and the {ordinalNet} {R} package.
\newblock \emph{Journal of Statistical Software}, 99\penalty0 (6):\penalty0
  1--42, 2021.

\bibitem[Xin et~al.(2014)Xin, Kawahara, Wang, and Gao]{xin2014efficient}
B.~Xin, Y.~Kawahara, Y.~Wang, and W.~Gao.
\newblock Efficient generalized fused lasso and its application to the
  diagnosis of alzheimer's disease.
\newblock In \emph{Proceedings of the AAAI Conference on Artificial
  Intelligence}, volume~28, 2014.

\bibitem[Zhang(2005)]{zhang2005schur}
F.~Zhang.
\newblock \emph{The Schur Complement and Its Applications}.
\newblock Numerical Methods and Algorithms. Springer, 2005.
\newblock ISBN 9780387242712.
\newblock URL \url{https://books.google.com/books?id=Wjd8\_AwjiIIC}.

\bibitem[Zhang et~al.(2014)Zhang, Yuan, and Wang]{zhang2014optimal}
W.~Zhang, S.~Yuan, and J.~Wang.
\newblock Optimal real-time bidding for display advertising.
\newblock In \emph{Proceedings of the 20th ACM SIGKDD international conference
  on Knowledge discovery and data mining}, pages 1077--1086, 2014.

\bibitem[Zhu(2017)]{zhu2017augmented}
Y.~Zhu.
\newblock An augmented admm algorithm with application to the generalized lasso
  problem.
\newblock \emph{Journal of Computational and Graphical Statistics}, 26\penalty0
  (1):\penalty0 195--204, 2017.

\end{thebibliography}
\bibliographystyle{abbrvnat}

\newpage
\appendix
\onecolumn

In Section \ref{sim.output}, we display summary statistics and additional figures for the observed mean squared errors (MSEs) for each method from the synthetic data experiments from Sections \ref{sparse.sim} and \ref{dense.sim}. We also briefly investigate the effect of implementing PRESTO with a squared \(\ell_2\) (ridge) penalty rather than an \(\ell_1\) penalty in Section \ref{ridge.supp}. We provide the proofs of Theorems \ref{est.known.beta} and \ref{log.imb} in Section \ref{main.proofs}. In Section \ref{main.thm.sec}, we present synthetic data experiments and analysis justifying the validity of one of the assumptions of Theorem \ref{main.cov.thm.2} in Sections \ref{min.eigen.bound} and \ref{min.eigen.cond.sim}, and we then prove Theorem \ref{main.cov.thm.2}. Theorems \ref{est.known.beta}, \ref{log.imb}, and \ref{main.cov.thm.2} depend on Lemma \ref{asym.matrix}, which is stated at the beginning of Section \ref{main.proofs} and proven in Section \ref{lemmas.sec}. We prove Theorem \ref{cons.thm} in Section \ref{cons.thm.proof}. Finally, in Section \ref{sec.est.presto} we provide implementation details for estimating \textsc{PRESTO}.

\section{More Simulation Results}\label{sim.output}

For more results from the synthetic experiments, see Tables \ref{sparse.wilcox.tab}, \ref{sparse.sum.stats.tab}, and \ref{dense.sum.stats.tab}, along with Figures \ref{fig.sparse.4}, \ref{fig.sparse}, \ref{fig.sparse.2}, \ref{fig.dense.2}, and \ref{fig.dense.3}.

\begin{table}[!htbp] \centering 
  \caption{Similar to Table \ref{unif.wilcox.tab}; calculated \(p\)-values for one-tailed paired \(t\)-tests for sparse differences simulation setting of Section \ref{sparse.sim} (statistically significant \(p\)-values indicate better performance for PRESTO).} 
  \label{sparse.wilcox.tab} 
\begin{tabular}{@{\extracolsep{5pt}} cccc} 
\\[-1.8ex]\hline 
\hline \\[-1.8ex] 
  Rare Prop. & Sparsity & Logit \(p\)-value & PO \(p\)-value \\ 
\hline \\[-1.8ex] 
1\%&  1/3 & 1.69e-33 & 6.42e-41 \\ 
1.17\%& 1/2 & 1.61e-15 & 2.78e-66 \\ 
0.61\%&  1/3 & 5.19e-74 & 4.21e-19 \\ 
0.71\%& 1/2 & 8.68e-48 & 3.38e-35 \\ 
0.37\%&  1/3 & 3.08e-61 & 0.00165 \\ 
0.43\% & 1/2& 3.75e-64 & 2.57e-11 \\ 
\hline \\[-1.8ex] 
\end{tabular} 
\end{table}

\begin{table}[!htbp] \centering 
  \caption{Means and standard errors of empirical MSEs for each method in each of three intercept settings in the sparse differences synthetic experiment setting of Section \ref{sparse.sim}.} 
  \label{sparse.sum.stats.tab} 
\begin{tabular}{@{\extracolsep{5pt}} ccccc} 
\\[-1.8ex]\hline 
\hline \\[-1.8ex] 
 Rare Class Proportion & Sparsity &  PRESTO & Logistic Regression & Proportional Odds \\ 
\hline \\[-1.8ex] 
1\% &  1/3& 6.05e-05 (2.1e-06) & 9.38e-05 (2.5e-06) & 8.62e-05 (3.1e-06) \\ 
1.17\%& 1/2 & 9.87e-05 (2.9e-06) & 1.25e-04 (3.3e-06) & 1.66e-04 (5.5e-06) \\ 
0.61\%&  1/3 & 3.03e-05 (1.1e-06) & 6.90e-05 (2.1e-06) & 3.64e-05 (1.4e-06) \\ 
0.71\%& 1/2 & 5.22e-05 (1.9e-06) & 8.89e-05 (2.5e-06) & 7.50e-05 (3e-06) \\ 
0.37\% &  1/3& 1.40e-05 (6e-07) & 5.66e-05 (2.4e-06) & 1.49e-05 (6.1e-07) \\ 
0.43\%& 1/2 & 2.63e-05 (1e-06) & 7.21e-05 (2.7e-06) & 3.17e-05 (1.4e-06) \\ 
\hline \\[-1.8ex] 
\end{tabular}
\end{table}

\begin{figure}[htbp]
\begin{center}
\includegraphics{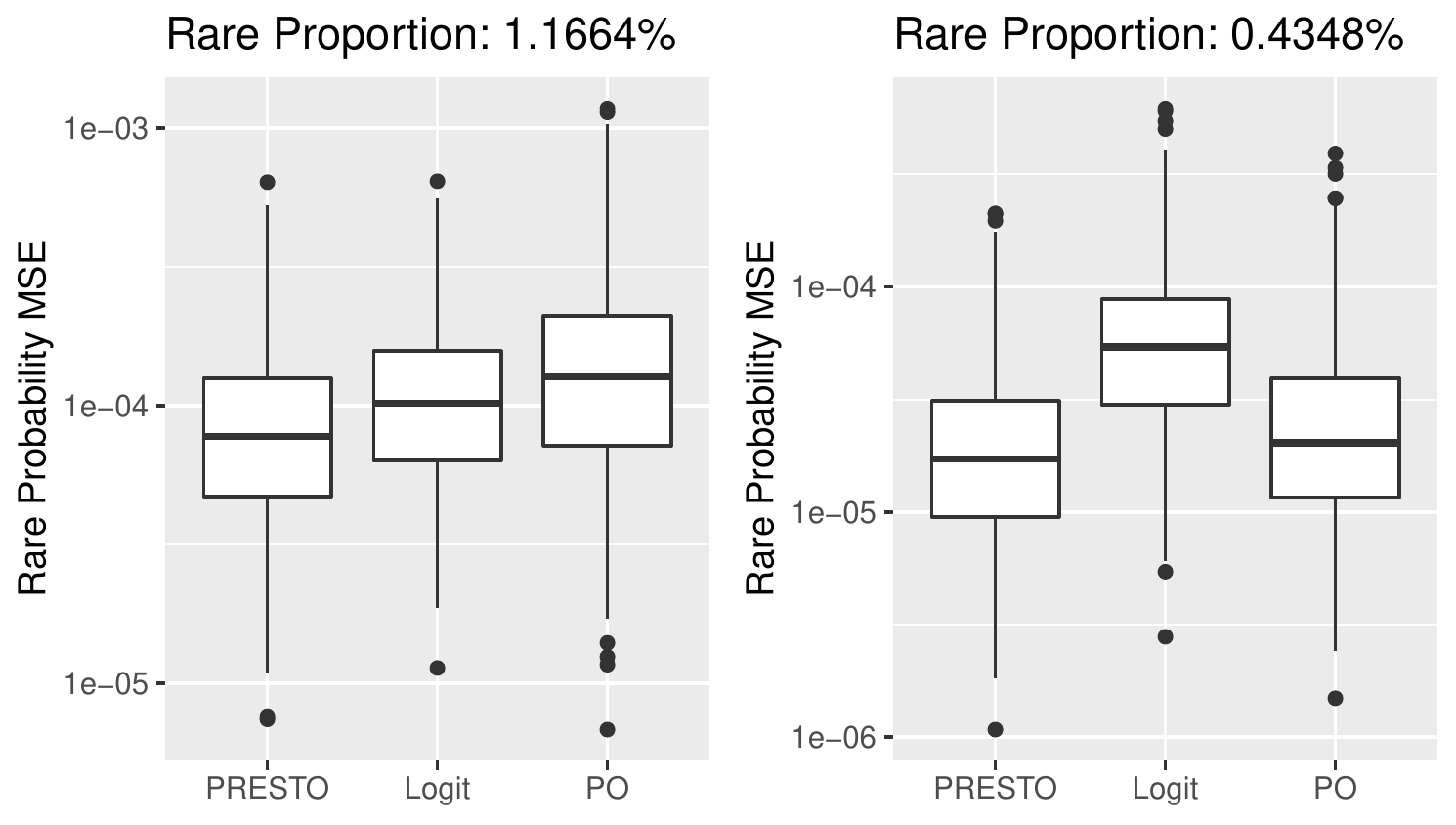}
\caption{MSE of predicted rare class probabilities for each method across all \(n = 2500\) observations, across 700 simulations, in sparse differences synthetic experiment setting of Section \ref{sparse.sim} with sparsity \(1/3\). (These plots are for the two intercept settings that weren't shown in the main text for the sparsity setting of \(1/3\))}
\label{fig.sparse.4}
\end{center}
\end{figure}

\begin{table}[!htbp] \centering 
  \caption{Means and standard errors of empirical MSEs for each method in each of four intercept settings in the uniform differences synthetic experiment setting of Section \ref{dense.sim}.} 
  \label{dense.sum.stats.tab} 
\begin{tabular}{@{\extracolsep{5pt}} cccc} 
\\[-1.8ex]\hline 
\hline \\[-1.8ex] 
  Rare Class Proportion & PRESTO & Logistic Regression & Proportional Odds \\ 
\hline \\[-1.8ex] 
 1.6\% & 1.13e-04 (2.7e-06) & 1.35e-04 (2.9e-06) & 1.85e-04 (5.2e-06) \\ 
0.99\% & 5.82e-05 (1.7e-06) & 9.25e-05 (2.3e-06) & 8.53e-05 (2.6e-06) \\ 
0.62\% & 2.86e-05 (8.3e-07) & 6.51e-05 (1.7e-06) & 3.43e-05 (9.8e-07) \\ 
0.36\% & 1.33e-05 (4.4e-07) & 5.36e-05 (2.2e-06) & 1.43e-05 (5e-07) \\ 
\hline \\[-1.8ex] 
\end{tabular} 
\end{table} 

\begin{figure}[htbp]
\begin{center}
\includegraphics{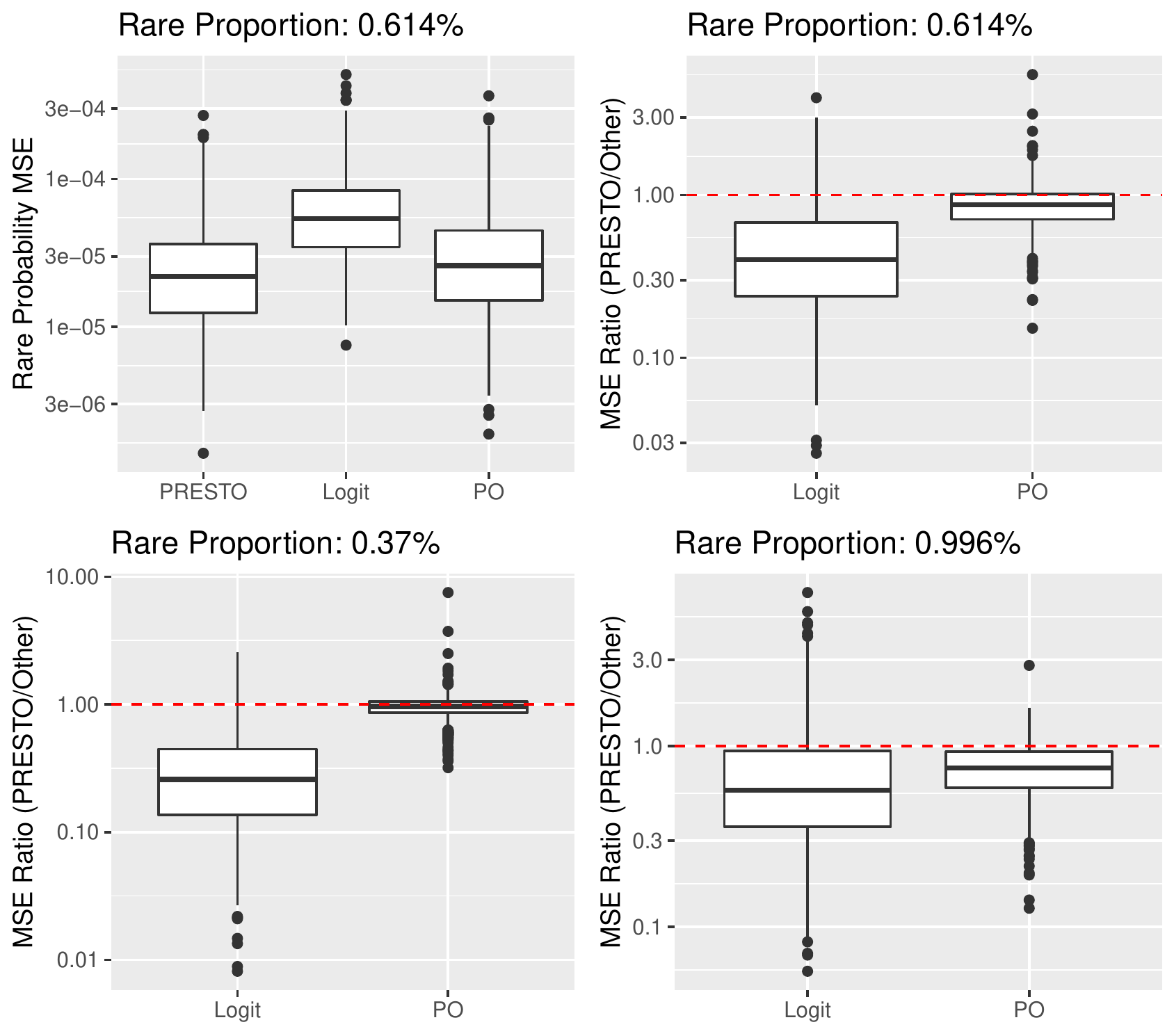}
\caption{Same as Figure \ref{fig.sparse.3}, but for the simulations with sparsity \(1/2\).}
\label{fig.sparse}
\end{center}
\end{figure}

\begin{figure}[htbp]
\begin{center}
\includegraphics{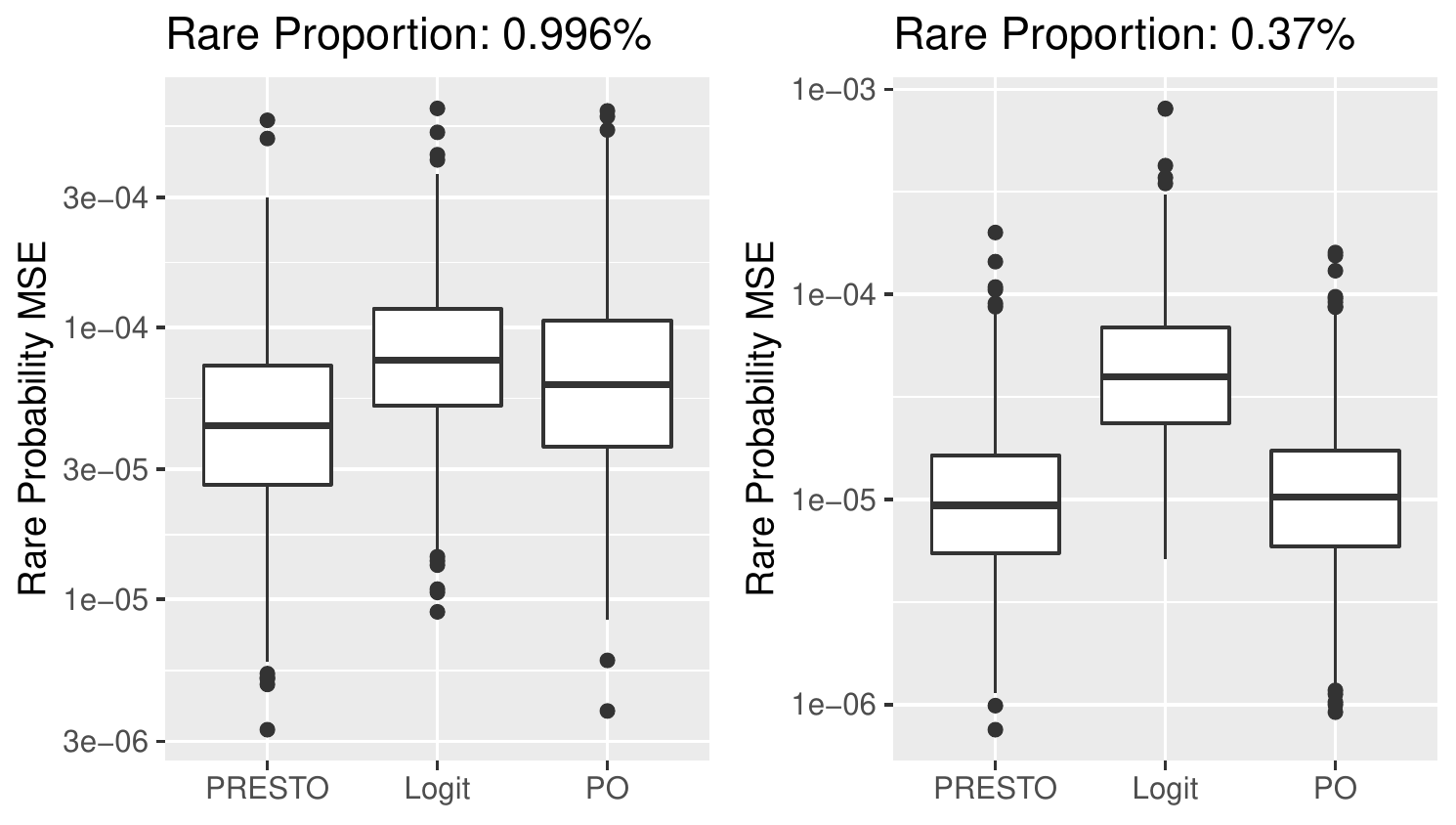}
\caption{Same as Figure \ref{fig.sparse.4}, but for the simulations with sparsity \(1/2\).}
\label{fig.sparse.2}
\end{center}
\end{figure}

\begin{figure}[htbp]
\begin{center}
\includegraphics{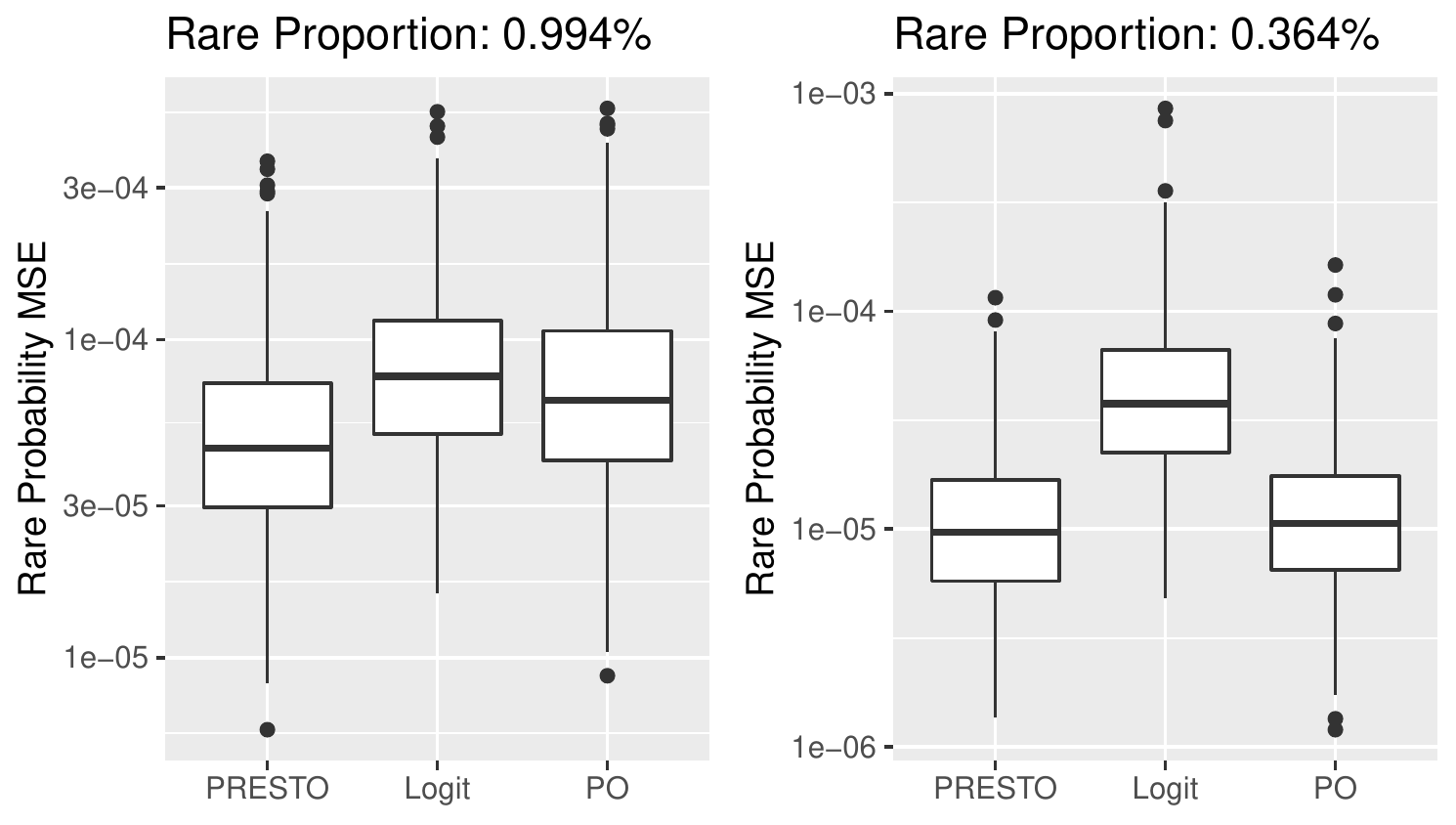}
\caption{MSE of predicted rare class probabilities for each method across all \(n = 2500\) observations, across 700 simulations, in uniform differences synthetic experiment setting of Section \ref{dense.sim}. (These plots are for two of the intercept settings that weren't shown in the main text.)}
\label{fig.dense.2}
\end{center}
\end{figure}

\begin{figure}[htbp]
\begin{center}
\includegraphics{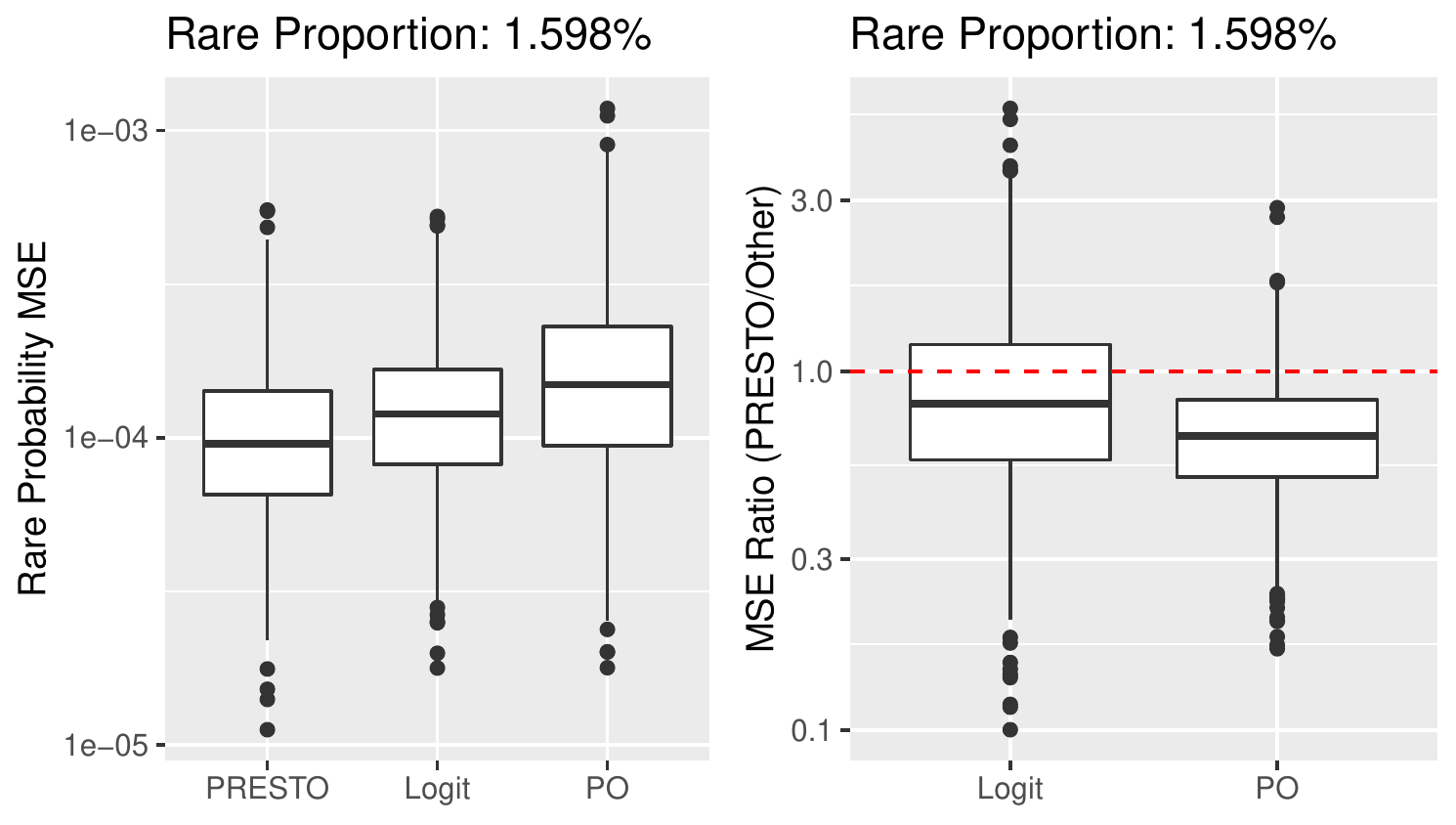}
\caption{MSE of predicted rare class probabilities for each method across all \(n = 2500\) observations, across 700 simulations, in uniform differences synthetic experiment setting of Section \ref{dense.sim} for intercept setting of \((0, 2.5, 4.5)\).}
\label{fig.dense.3}
\end{center}
\end{figure}

\subsection{Ridge PRESTO}\label{ridge.supp}

We briefly investigate the effect of implementing PRESTO with a ridge penalty instead of an \(\ell_1\) penalty, similarly to proposals by \citet[Section 4.2.2]{tutz2016regularized} and \citet[Equation 8]{ugba2021smoothing},
\[
\lambda_n \left( \sum_{j=1}^p     \beta_{j1}^2 +  \sum_{j=1}^p  \sum_{k=2}^{K-1} \left(  \beta_{jk} - \beta_{j,k-1}\right)^2 \right) 
.
\]
We implement this method (``PRESTO\_L2") in the sparse differences synthetic data experiment of Section \ref{sparse.sim} on the same simulated data that was used for the other methods in the intercept setting \((0,3,5)\) for both sparsity levels. The implementation is identical to PRESTO in every way except for the ridge penalty---the method is implemented using our modification of the \texttt{ordinalNet} R package and the tuning parameter is selected in the same way. 

Figures \ref{fig.sparse.10} and \ref{fig.sparse.11} display the results. (The results for all methods but PRESTO\_L2 are identical to previous plots and are only displayed for reference.) We also present the means and standard deviations of the MSEs for each method in these settings in Table \ref{tab.presto.l2.mean.se}, and \(p\)-values for one-tailed paired \(t\)-tests of the alternative hypothesis that PRESTO has a lower MSE than the competitor methods in Table \ref{tab.presto.l2.p.values}. We see that in practice, PRESTO and PRESTO\_L2 seem to perform similarly in our setting, though the \(t\)-tests show that PRESTO does outperform PRESTO\_L2 at a \(5\%\) significance level in both settings. We might expect PRESTO to better outperform PRESTO\_L2 in settings where \(p/n\) is larger and where the sparsity is lower. We also note it is not clear if PRESTO\_L2 enjoys a high-dimensional consistency guarantee similar to Theorem \ref{main.cov.thm.2}.

\begin{figure}[htbp]
\begin{center}
\includegraphics{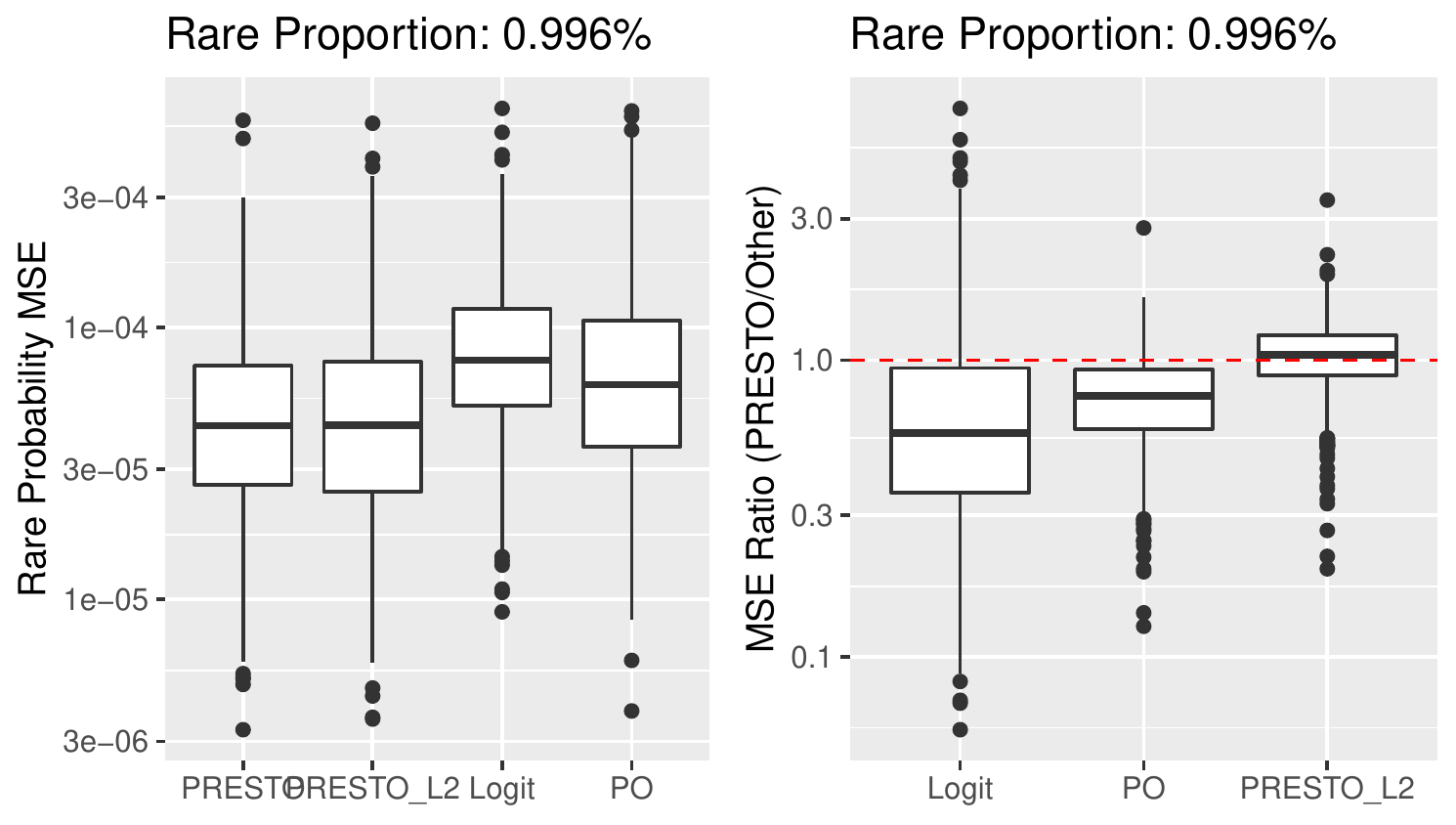}
\caption{MSE of predicted rare class probabilities for each method across all \(n = 2500\) observations, across 700 simulations, in sparse differences synthetic experiment setting of Section \ref{sparse.sim} with sparsity \(1/3\) for intercept setting of \((0, 3, 5)\), with PRESTO\_L2 implemented as well, as described in Section \ref{ridge.supp}.}
\label{fig.sparse.10}
\end{center}
\end{figure}

\begin{figure}[htbp]
\begin{center}
\includegraphics{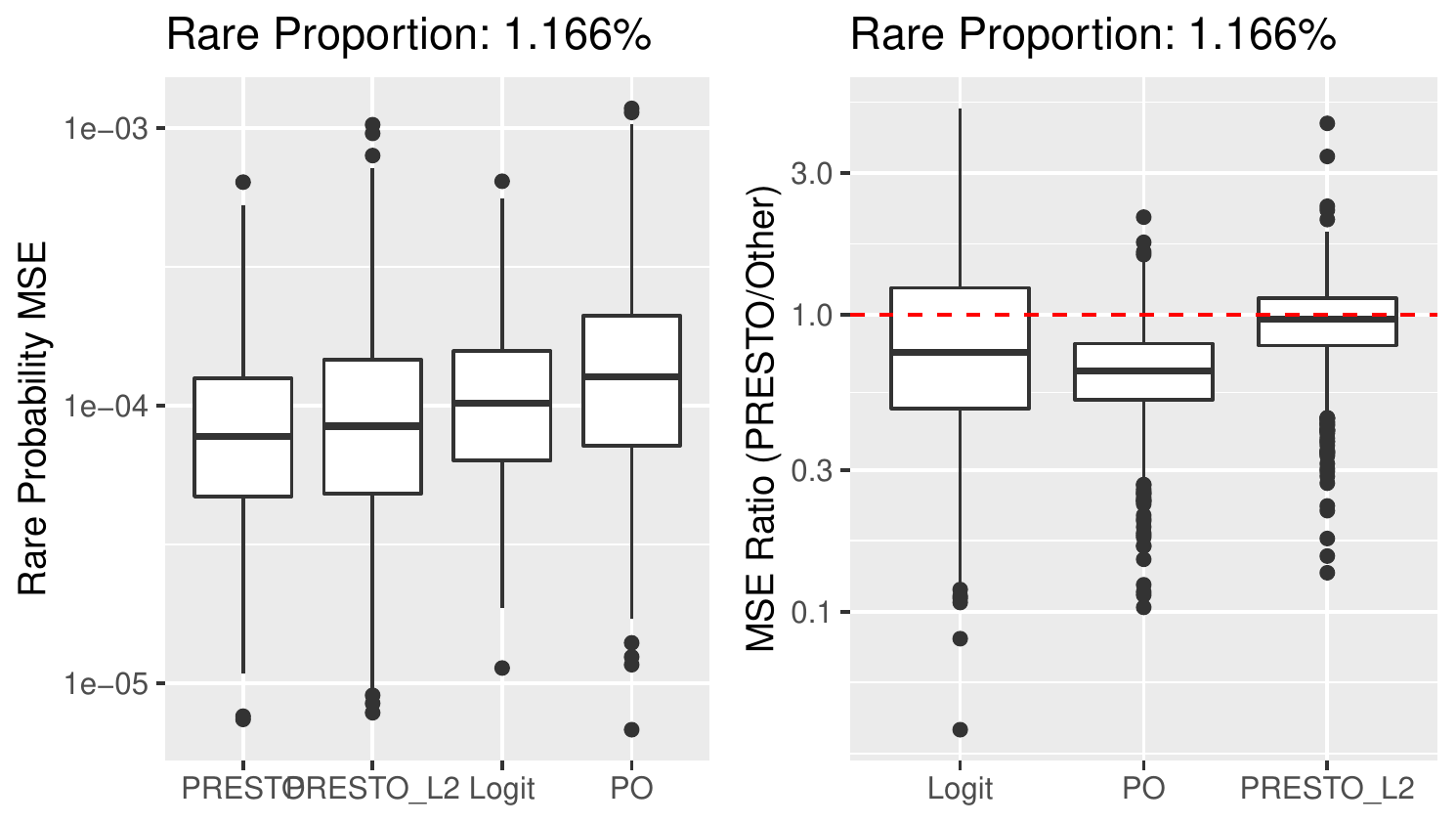}
\caption{MSE of predicted rare class probabilities for each method across all \(n = 2500\) observations, across 700 simulations, in sparse differences synthetic experiment setting of Section \ref{sparse.sim} with sparsity \(1/2\) for intercept setting of \((0, 3, 5)\), with PRESTO\_L2 implemented as well, as described in Section \ref{ridge.supp}.}
\label{fig.sparse.11}
\end{center}
\end{figure}

\begin{table}[!htbp] \centering 
  \caption{Means and standard errors of empirical MSEs for each method in each of three intercept settings in the sparse differences synthetic experiment setting of Section \ref{sparse.sim} for intercept setting of \((0, 3, 5)\), with PRESTO\_L2 implemented as well, as described in Section \ref{ridge.supp}.} 
  \label{tab.presto.l2.mean.se} 
\begin{tabular}{@{\extracolsep{5pt}} cccccc} 
\\[-1.8ex]\hline 
\hline \\[-1.8ex] 
 Rare Prop. & Sparsity &  PRESTO & Logistic Regression & Proportional Odds & PRESTO\_L2 \\ 
\hline \\[-1.8ex] 
 1\% & \(1/3\) & 6.05e-05 (2.1e-06) & 9.38e-05 (2.5e-06) & 8.62e-05 (3.1e-06) & 6.25e-05 (2.4e-06) \\ 
 1.17\% & \(1/2\) & 9.87e-05 (2.9e-06) & 1.25e-04 (3.3e-06) & 1.66e-04 (5.5e-06) & 1.17e-04 (4.2e-06) \\ 
\hline \\[-1.8ex] 
\end{tabular} 
\end{table} 

\begin{table}[!htbp] \centering 
  \caption{Calculated \(p\)-values for one-tailed paired \(t\)-tests for sparse differences simulation setting of Section \ref{sparse.sim} for intercept setting of \((0, 3, 5)\), with PRESTO\_L2 implemented as well, as described in Section \ref{ridge.supp}. (Statistically significant \(p\)-values indicate better performance for PRESTO).} 
  \label{tab.presto.l2.p.values} 
\begin{tabular}{@{\extracolsep{5pt}} ccccc} 
\\[-1.8ex]\hline 
\hline \\[-1.8ex] 
 Rare Class Proportion & Sparsity &  Logistic Regression & Proportional Odds & PRESTO\_L2 \\ 
\hline \\[-1.8ex] 
 1\%  & \(1/3\) & 1.69e-33 & 6.42e-41 & 0.0358 \\ 
 1.17\% & \(1/2\) & 1.61e-15 & 2.78e-66 & 1.51e-14 \\ 
\hline \\[-1.8ex] 
\end{tabular} 
\end{table}

\section{Statement of Lemma \ref{asym.matrix} and Proofs of Theorems \ref{log.imb} and \ref{est.known.beta}}\label{main.proofs}

In Section \ref{asym.mat.sec} we state Lemma \ref{asym.matrix}, and we prove Theorems \ref{log.imb} and \ref{est.known.beta} in Section \ref{thms.1.2.proofs}.

\subsection{Statement of Lemma \ref{asym.matrix}}\label{asym.mat.sec}

Theorems \ref{log.imb}, \ref{est.known.beta}, and \ref{main.cov.thm.2} relate to the asymptotic covariance matrices of the maximum likelihood estimators of the parameters of the proportional odds and logistic regression models. Under mild regularity conditions, the asymptotic covariance matrix of any maximum likelihood estimator (when scaled by \(\sqrt{n}\)) is known to be the inverse of the Fisher information matrix
\[
-\E \left[ \pderiv{^2}{\boldsymbol{\theta} \boldsymbol{\theta}^\top} \mathcal{L}(\boldsymbol{\theta})  \right] ,
\]
where \(\boldsymbol{\theta}\) are the parameters estimated by the model and \( \mathcal{L}(\boldsymbol{\theta})  \) is the log likelihood \citep[Section 4.2.2]{serfling1980}. In the proof of Lemma \ref{asym.matrix}, we calculate these Fisher information matrices for the proportional odds and logistic regression models and verify the needed regularity conditions.

\begin{lemma}\label{asym.matrix} Assume that no class has probability 0 for any \(\boldsymbol{x} \in \mathcal{S}\) (equivalently, assume that all of the intercepts in the proportional odds model \eqref{prop_odds} are not equal, so \(\alpha_1 < \ldots < \alpha_{K-1}\)). Assume that \(dF(\boldsymbol{x})\) has bounded support.

\begin{enumerate}

\item The Fisher information matrix for the maximum likelihood estimator of the proportional odds model \eqref{prop_odds} is
\[
I^{\text{prop. odds}} (\boldsymbol{\alpha}, \boldsymbol{\beta})= \begin{pmatrix} I_{\alpha \alpha}^{\text{prop. odds}} &  \left( I_{\beta \alpha}^{\text{prop. odds}} \right)^\top 
\\ I_{\beta \alpha}^{\text{prop. odds}}& I_{\beta \beta}^{\text{prop. odds}}
\end{pmatrix} \in \mathbb{R}^{(K - 1 + p) \times (K - 1 + p)}
\]
where
\begin{align}
I_{\alpha \alpha}^{\text{prop. odds}}(\boldsymbol{\alpha}, \boldsymbol{\beta}) = ~ &  \begin{pmatrix}
M_1 & -\tilde{M}_2 & 0   & \cdots   & 0  & 0
\\ -\tilde{M}_2 & M_2 &- \tilde{M}_3   & \cdots   & 0  & 0
\\ 0 & -\tilde{M}_3 & M_3   & \cdots   & 0  & 0
\\ \vdots & \vdots & \vdots & \ddots  & \vdots & \vdots
\\ 0 & 0 & 0 &  \cdots & M_{K-2}  &- \tilde{M}_{K-1}
\\ 0 & 0 & 0 &  \cdots   & -\tilde{M}_{K-1}  & M_{K-1}
\end{pmatrix} ,\label{alpha.block}
\\I_{\beta \alpha}^{\text{prop. odds}}(\boldsymbol{\alpha}, \boldsymbol{\beta}) = ~ &  \begin{pmatrix}
J_1^{\boldsymbol{x}} + \tilde{J}_{2}^{\boldsymbol{x}}
\\   J_2^{\boldsymbol{x}} + \tilde{J}_{3}^{\boldsymbol{x}}
\\ \vdots
\\  J_{K-1}^{\boldsymbol{x}} + \tilde{J}_{K}^{\boldsymbol{x}}
\end{pmatrix}   , \qquad \text{and} \label{alpha.beta.block}
\\ I_{\beta \beta}^{\text{prop. odds}}(\boldsymbol{\alpha}, \boldsymbol{\beta})  = ~ & \sum_{k=1}^K \left( J_k^{\boldsymbol{x} \boldsymbol{x}^\top} +  \tilde{J}_k^{\boldsymbol{x} \boldsymbol{x}^\top} \right) ,\label{beta.block}
\end{align}
where
\begin{align}
M_k & := \int \left[ p_{k}(\boldsymbol{x}) (1 - p_{k} (\boldsymbol{x}) ) \right]^2   \left(     \frac{1 }{\pi_{k}(\boldsymbol{x}) }   +  \frac{1}{\pi_{k+1}(\boldsymbol{x}) }  \right) \ d F(\boldsymbol{x}) , \qquad k \in \{1, \ldots, K-1\}, \label{m.def}
\\ \tilde{M}_k & := \int  p_{k}(\boldsymbol{x}) (1 - p_{k}(\boldsymbol{x}) ) p_{k-1}(\boldsymbol{x}) (1 - p_{k-1}(\boldsymbol{x}) ) \cdot   \frac{1 }{\pi_{k}(\boldsymbol{x}) } \ d F(\boldsymbol{x}) , \qquad k \in \{2, \ldots, K-1\} \label{m.tilde.def}
\end{align}
and
\begin{align*}
J_k & := \int \pi_k(\boldsymbol{x}) p_k (\boldsymbol{x}) [ 1 - p_k (\boldsymbol{x}) ] \ d F(\boldsymbol{x})  \in \mathbb{R},
\\ J_k^{\boldsymbol{x}} & := \int \boldsymbol{x} \pi_k(\boldsymbol{x}) p_k (\boldsymbol{x}) [ 1 - p_k (\boldsymbol{x}) ] \ d F(\boldsymbol{x})  \in \mathbb{R}^p,
\\ J_k^{\boldsymbol{x} \boldsymbol{x}^\top} & := \int \boldsymbol{x} \boldsymbol{x}^\top \pi_k(\boldsymbol{x}) p_k (\boldsymbol{x}) [ 1 - p_k (\boldsymbol{x}) ] \ d F(\boldsymbol{x})  \in \mathbb{R}^{p \times p},
\\ \tilde{J}_k & := \int \pi_k(\boldsymbol{x}) p_{k-1} (\boldsymbol{x}) [ 1 - p_{k-1} (\boldsymbol{x}) ] \ d F(\boldsymbol{x})  \in \mathbb{R},
\\ \tilde{J}_k^{\boldsymbol{x}} & := \int \boldsymbol{x} \pi_k(\boldsymbol{x}) p_{k-1} (\boldsymbol{x}) [ 1 - p_{k-1} (\boldsymbol{x}) ] \ d F(\boldsymbol{x})  \in \mathbb{R}^p, \qquad \text{and}
\\ \tilde{J}_k^{\boldsymbol{x} \boldsymbol{x}^\top} & := \int \boldsymbol{x} \boldsymbol{x}^\top \pi_k(\boldsymbol{x}) p_{k-1} (\boldsymbol{x}) [ 1 - p_{k-1} (\boldsymbol{x}) ] \ d F(\boldsymbol{x})  \in \mathbb{R}^{p \times p}
\end{align*}
for all \(k \in [K]\).

\item The Fisher information matrix for the maximum likelihood estimator of the logistic regression model predicting whether or not each observation is in class 1 is
\begin{equation}\label{logistic.fisher}
I^{\text{logistic}}(\alpha_1, \boldsymbol{\beta})= \begin{pmatrix} I_{\alpha \alpha}^{\text{logistic}} & \left( I_{\beta \alpha}^{\text{logistic}} \right)^\top
\\ I_{\beta \alpha}^{\text{logistic}} & I_{\beta \beta}^{\text{logistic}}
\end{pmatrix} = \E \left[ \pi_1(\boldsymbol{X})[1 - \pi_1(\boldsymbol{X})] \boldsymbol{\tilde{X}} \boldsymbol{\tilde{X}}^\top \right] \in \mathbb{R}^{(p+1) \times (p+1)}
\end{equation}
where \(\boldsymbol{\tilde{X}} := \begin{pmatrix} \boldsymbol{1} & \boldsymbol{X} \end{pmatrix}\) (an \(n\)-vector of all ones followed by \(\boldsymbol{X}\)) and
\begin{align}
I_{\alpha \alpha}^{\text{logistic}}(\alpha_1, \boldsymbol{\beta}) =  ~ & M_1^{\text{logistic}}  ,\label{log.reg.alpha.info}
\\ I_{\beta \alpha}^{\text{logistic}}(\alpha_1, \boldsymbol{\beta})  = ~ & J_1^{\boldsymbol{x}\text{; logistic}}+\tilde{J}_2^{\boldsymbol{x}\text{; logistic}} ,  \qquad \text{and} \label{log.reg.alpha.beta.info}
\\ I_{\beta \beta}^{\text{logistic}}(\alpha_1, \boldsymbol{\beta})  = ~ &  J_1^{\boldsymbol{x} \boldsymbol{x}^\top\text{; logistic}}  + \tilde{J}_2^{\boldsymbol{x} \boldsymbol{x}^\top\text{; logistic}} \label{log.reg.beta.info} 
,
\end{align}
where we define
\begin{align}
M_1^{\text{logistic}} & := \int  \pi_{1}(\boldsymbol{x}) (1 - \pi_{1} (\boldsymbol{x}) ) \ d F(\boldsymbol{x})  \label{m.1.logit.def} 
\end{align}
and
\begin{align*}
J_1^{\boldsymbol{x}\text{; logistic}} & := \int \boldsymbol{x} \pi_1(\boldsymbol{x})^2  [ 1 - \pi_1 (\boldsymbol{x}) ] \ d F(\boldsymbol{x})  \in \mathbb{R}^p = J_1^{\boldsymbol{x}},
\\ J_1^{\boldsymbol{x} \boldsymbol{x}^\top\text{; logistic}} & := \int \boldsymbol{x} \boldsymbol{x}^\top \pi_1(\boldsymbol{x})^2 [ 1 - \pi_1 (\boldsymbol{x}) ] \ d F(\boldsymbol{x})  \in \mathbb{R}^{p \times p} = J_1^{\boldsymbol{x} \boldsymbol{x}^\top}  ,
\\ \tilde{J}_2^{\boldsymbol{x}\text{; logistic}} & := \int \boldsymbol{x} \pi_{1} (\boldsymbol{x}) [ 1 - \pi_{1} (\boldsymbol{x}) ]^2 \ d F(\boldsymbol{x})  \in \mathbb{R}^p, \qquad \text{and}
\\ \tilde{J}_2^{\boldsymbol{x} \boldsymbol{x}^\top\text{; logistic}} & := \int \boldsymbol{x} \boldsymbol{x}^\top \pi_{1} (\boldsymbol{x}) [ 1 - \pi_{1} (\boldsymbol{x}) ]^2 \ d F(\boldsymbol{x})  \in \mathbb{R}^{p \times p}.
\end{align*}

\item The information matrices \(I^{\text{prop. odds}} (\boldsymbol{\alpha}, \boldsymbol{\beta}) \) and \( I^{\text{logistic}}(\alpha_1, \boldsymbol{\beta})\) are finite and positive definite, and the following convergences hold:

\[
\sqrt{n} \cdot  \left( \boldsymbol{\hat{\theta}}_1^{\text{prop. odds}}  - \boldsymbol{\theta}_1 \right) \xrightarrow{d} \mathcal{N} \left(  \boldsymbol{0}, \left(  I^{\text{prop. odds}} (\boldsymbol{\alpha}, \boldsymbol{\beta}) \right)^{-1} \right)
\]
and

\[
\sqrt{n} \cdot \left( \boldsymbol{\hat{\theta}}_1^{\text{logistic}} - \boldsymbol{\theta}_1 \right) \xrightarrow{d} \mathcal{N} \left( \boldsymbol{0} ,  \left( I^{\text{logistic}}(\alpha_1, \boldsymbol{\beta})  \right)^{-1}\right).
\]
Further, because these information matrices are symmetric and positive definite, by Observation 7.1.2 in \citet{horn_johnson_2012} the principal submatrices \(I_{\alpha \alpha}^{\text{prop. odds}} := I_{\alpha \alpha}^{\text{prop. odds}}(\boldsymbol{\alpha}, \boldsymbol{\beta})\), \(I_{\beta \beta}^{\text{prop. odds}} := I_{\beta \beta}^{\text{prop. odds}}(\boldsymbol{\alpha}, \boldsymbol{\beta})\), and  \( I_{\beta \beta}^{\text{logistic}} := I_{\beta \beta}^{\text{logistic}}(\alpha_1, \boldsymbol{\beta})\) are all positive definite. Finally, we can characterize the finite-sample bias: for any \(\boldsymbol{v} \in \mathbb{R}^{p+1}\), \( [ (\hat{\alpha},  \boldsymbol{\hat{\beta}}^\top) - (\alpha, \boldsymbol{\beta})^\top] \boldsymbol{v} = \mathcal{O}(1/n)\).

\end{enumerate}

\end{lemma}
\begin{proof} Provided in Section \ref{lemmas.sec}.
\end{proof}

\begin{remark}\label{j.remark} 
Note that \(J_K = 0\) because \(1 - p_{K}\left( \boldsymbol{x} \right) = 1 - \mathbb{P}(y\left( \boldsymbol{x} \right) \leq K \mid \boldsymbol{x}) = 0\) for all \(\boldsymbol{x}\), and similarly for \(J_K^{\boldsymbol{x}}\) and \(J_K^{\boldsymbol{x} \boldsymbol{x}^\top}\). Likewise, \(\tilde{J}_1 = 0\) because \(p_{0}\left( \boldsymbol{x} \right) = \mathbb{P}(y\left( \boldsymbol{x} \right) \leq 0 \mid \boldsymbol{x}) = 0\) for all \(\boldsymbol{x}\), and similarly for \(\tilde{J}_1^{\boldsymbol{x}}\) and \(\tilde{J}_1^{\boldsymbol{x} \boldsymbol{x}^\top}\).

We also take a moment to briefly establish some identities we will use later. For any \(k \in \{1, \ldots, K-1\}\),
\begin{align}
J_k^{\boldsymbol{x}} + \tilde{J}_{k+1}^{\boldsymbol{x}}  = ~ &  \int  \boldsymbol{x} \pi_k(\boldsymbol{x}) p_k (\boldsymbol{x}) [ 1 - p_k (\boldsymbol{x}) ] \ d F(\boldsymbol{x})   \nonumber
+  \int  \boldsymbol{x}  \pi_{k+1}(\boldsymbol{x}) p_{k} (\boldsymbol{x}) [ 1 - p_{k} (\boldsymbol{x}) ] \ d F(\boldsymbol{x}) 
 \nonumber
\\ = ~ &  \int  \boldsymbol{x} \left[ \pi_k(\boldsymbol{x}) + \pi_{k+1}(\boldsymbol{x}) \right] p_k (\boldsymbol{x}) [ 1 - p_k (\boldsymbol{x}) ] \ d F(\boldsymbol{x})  \label{j.2.j.tilde.3.sum}
.
\end{align}
Similarly,
\begin{align}
J_k^{\boldsymbol{x} \boldsymbol{x}^\top}   + \tilde{J}_{k+1}^{\boldsymbol{x} \boldsymbol{x}^\top} = ~ &  \int \boldsymbol{x} \boldsymbol{x}^\top \left[ \pi_k(\boldsymbol{x}) + \pi_{k+1}(\boldsymbol{x}) \right] p_k (\boldsymbol{x}) [ 1 - p_k (\boldsymbol{x}) ] \ d F(\boldsymbol{x}) , \nonumber 
\end{align}
so from \eqref{beta.block} we have
\begin{align}
\ I_{\beta \beta}^{\text{prop. odds}}(\boldsymbol{\alpha}, \boldsymbol{\beta})  = ~ & \sum_{k=1}^K \left( J_k^{\boldsymbol{x} \boldsymbol{x}^\top} +  \tilde{J}_k^{\boldsymbol{x} \boldsymbol{x}^\top} \right) \nonumber
\\  = ~ &  \tilde{J}_{1}^{\boldsymbol{x} \boldsymbol{x}^\top}  +  \sum_{k=1}^{K-1}  \left( J_k^{\boldsymbol{x} \boldsymbol{x}^\top}   + \tilde{J}_{k+1}^{\boldsymbol{x} \boldsymbol{x}^\top} \right) + J_K^{\boldsymbol{x} \boldsymbol{x}^\top}   \nonumber
\\  = ~ & \int \boldsymbol{x} \boldsymbol{x}^\top \pi_1(\boldsymbol{x}) \underbrace{p_{0} (\boldsymbol{x})}_{= \mathbb{P}(y \leq 0 \mid \boldsymbol{x}) = 0} [ 1 - p_{0} (\boldsymbol{x}) ] \ d F(\boldsymbol{x}) +  \sum_{k=1}^{K-1}  \left( J_k^{\boldsymbol{x} \boldsymbol{x}^\top}   + \tilde{J}_{k+1}^{\boldsymbol{x} \boldsymbol{x}^\top} \right) \nonumber
\\ & + \int \boldsymbol{x} \boldsymbol{x}^\top \pi_K(\boldsymbol{x}) p_K (\boldsymbol{x}) \left[ 1 - \underbrace{p_K (\boldsymbol{x})}_{= \mathbb{P}(y \leq K \mid \boldsymbol{x}) = 1} \right] \ d F(\boldsymbol{x}) \nonumber
\\  = ~ &  \sum_{k=1}^{K-1}  \int \boldsymbol{x} \boldsymbol{x}^\top \left[ \pi_k(\boldsymbol{x}) + \pi_{k+1}(\boldsymbol{x}) \right] p_k (\boldsymbol{x}) [ 1 - p_k (\boldsymbol{x}) ] \ d F(\boldsymbol{x})  
.
\label{i.beta.beta.id}
\end{align}

\end{remark}

\subsection{Proofs of Theorems \ref{log.imb} and \ref{est.known.beta}}\label{thms.1.2.proofs}

Equipped with the results of Lemma \ref{asym.matrix}, we proceed to prove Theorems \ref{log.imb} and \ref{est.known.beta}. (Recall that for square matrices \(\boldsymbol{A}\) and \(\boldsymbol{B}\) of equal dimension \(p\), we say \(\boldsymbol{A} \preceq  \boldsymbol{B}\) if \(\boldsymbol{B} - \boldsymbol{A}\) is positive semidefinite.)

\begin{proof}[Proof of Theorem \ref{log.imb}] 

\begin{enumerate}

\item

Note that the assumptions of Lemma \ref{asym.matrix} are satisfied. Since the inverse of the asymptotic covariance matrix is
\[
 \E \left[ \pi(\boldsymbol{X}) [1 - \pi(\boldsymbol{X})] \boldsymbol{\tilde{X}} \boldsymbol{\tilde{X}}^\top \right]  \preceq    \pi_{\text{rare}} (1 -   \pi_{\text{rare}} ) \E \left[ \boldsymbol{\tilde{X}} \boldsymbol{\tilde{X}}^\top \right]  
\]
(where the second step is valid because \(t \mapsto t(1-t)\) is monotone increasing in \(t\) for \(t \in [0, 1/2]\)), by Corollary 7.7.4 in \citet{horn_johnson_2012} the largest eigenvalue of the inverse of the asymptotic covariance matrix is no larger than \(\pi_{\text{rare}} (1 -   \pi_{\text{rare}} ) \lambda_{\text{max}}\). Therefore
\begin{align*}
\mathrm{Asym. } \Cov \left( (\sqrt{n} \cdot \hat{\alpha}, \sqrt{n} \cdot \boldsymbol{\hat{\beta}} ) \right) & =  \left( \E \left[ \pi(\boldsymbol{X}) [1 - \pi(\boldsymbol{X})] \boldsymbol{\tilde{X}} \boldsymbol{\tilde{X}}^\top \right] \right)^{-1}
 \succeq  \left(  \pi_{\text{rare}} (1 -   \pi_{\text{rare}} ) \E \left[ \boldsymbol{\tilde{X}} \boldsymbol{\tilde{X}}^\top \right] \right)^{-1}
\end{align*}
has smallest eigenvalue at least \(1/[\lambda_{\text{max}}  \pi_{\text{rare}} (1 -   \pi_{\text{rare}} )  ]\), which is larger than \(1/[\lambda_{\text{max}}  \pi_{\text{rare}} ]\), again using Corollary 7.7.4 in \citet{horn_johnson_2012}. So for any \(\boldsymbol{v} \in \mathbb{R}^{p+1}\), by Theorem 5.1.8 in \citet{lehmann1999elements}
\begin{align*}
\mathrm{Asym. } \Var \left(   (\hat{\alpha}, \boldsymbol{\hat{\beta}}^\top)  \boldsymbol{v} \right) =  \boldsymbol{v}^\top  \mathrm{Asym. } \Cov \left( (\sqrt{n} \cdot \hat{\alpha}, \sqrt{n} \cdot \boldsymbol{\hat{\beta}} ) \right) \boldsymbol{v}
\geq  \frac{\boldsymbol{v}^\top \boldsymbol{v}}{\lambda_{\text{max}}  \pi_{\text{rare}} }
.
\end{align*}
Finally, since we have already shown that \(  (\hat{\alpha}, \boldsymbol{\hat{\beta}}^\top)   \) is asymptotically unbiased, the asymptotic MSE is equal to this asymptotic variance:
\begin{align*}
\mathrm{Asym. MSE}((\hat{\alpha}, \boldsymbol{\hat{\beta}}^\top)  \boldsymbol{v})) =  ~ & \E \left[ \lim_{n \to \infty} \left( \sqrt{n} \cdot \left[ (\hat{\alpha}, \boldsymbol{\hat{\beta}}^\top)  \boldsymbol{v}  - (\alpha, \boldsymbol{\beta}^\top)  \boldsymbol{v}  \right] \right)^2 \right]
\\  =  ~ & \E \bigg[ \lim_{n \to \infty} \left( \sqrt{n} \cdot \left[ \E \left[ (\hat{\alpha}, \boldsymbol{\hat{\beta}}^\top)  \boldsymbol{v} \right] -  (\hat{\alpha}, \boldsymbol{\hat{\beta}}^\top)  \boldsymbol{v}  \right] \right)^2 
\\ & +  \lim_{n \to \infty} \left( \sqrt{n} \cdot \left[   (\alpha, \boldsymbol{\beta}^\top)  \boldsymbol{v}   - \E [ (\hat{\alpha}, \boldsymbol{\hat{\beta}}^\top)  \boldsymbol{v} \right]\right)^2 \bigg] 
\\  =  ~ &\mathrm{Asym. } \Var \left( \sqrt{n}\left[  (\alpha, \boldsymbol{\beta}^\top)  \boldsymbol{v}  -  (\hat{\alpha}, \boldsymbol{\hat{\beta}}^\top)  \boldsymbol{v}  \right]  \right)  + 0
\\ \geq ~ &  \frac{\boldsymbol{v}^\top \boldsymbol{v}}{\lambda_{\text{max}}  \pi_{\text{rare}} }
,
\end{align*}
where in the second-to-last step we used \([ (\hat{\alpha},  \boldsymbol{\hat{\beta}}^\top) - (\alpha, \boldsymbol{\beta})^\top] \boldsymbol{v} = \mathcal{O}(1/n)\) from Lemma \ref{asym.matrix}.

\item Because \((\hat{\alpha}, \boldsymbol{\hat{\beta}})  \mapsto \hat{\pi}(\boldsymbol{z})\) is differentiable for all \(\boldsymbol{z} \in \mathbb{R}^p\), by the delta method (Theorem 3.1 in \citealt{van2000asymptotic})
\[
\sqrt{n} \cdot [ \hat{\pi}(\boldsymbol{z}) - \pi(\boldsymbol{z})  ] \xrightarrow{d}   \mathcal{N} \left( 0,  \pi (\boldsymbol{z})^2  \left[1 - \pi (\boldsymbol{z}) \right]^2  \left(1, \boldsymbol{z}^\top \right) \left( I^{\text{logistic}} (\boldsymbol{\alpha}, \boldsymbol{\beta}) \right)^{-1}   \left(1, \boldsymbol{z}^\top \right)^\top \right)
\]
for any \(\boldsymbol{z} \in \mathbb{R}^p\). Therefore
\begin{align*}
\mathrm{Asym. } \Var \left(  \sqrt{n} \cdot \hat{\pi}(\boldsymbol{z})  \right) & =  \pi (\boldsymbol{z})^2  \left[1 - \pi (\boldsymbol{z}) \right]^2  \left(1, \boldsymbol{z}^\top \right) \left( I^{\text{logistic}} (\boldsymbol{\alpha}, \boldsymbol{\beta}) \right)^{-1}   \left(1, \boldsymbol{z}^\top \right)^\top
\\ & \geq   \pi (\boldsymbol{z})^2  \left[1 - \pi (\boldsymbol{z}) \right]^2  \left \lVert \left(1, \boldsymbol{z}^\top \right) \right \rVert_2^2 \lambda_{\text{min}} \left( \left( I^{\text{logistic}} (\boldsymbol{\alpha}, \boldsymbol{\beta}) \right)^{-1}  \right)
\\ & =  \frac{\pi (\boldsymbol{z})^2  \left[1 - \pi (\boldsymbol{z}) \right]^2  \left \lVert \left(1, \boldsymbol{z}^\top \right) \right \rVert_2^2 }{\left \lVert I^{\text{logistic}} (\boldsymbol{\alpha}, \boldsymbol{\beta}) \right \rVert_{\text{op}}}
\\ & \geq  \frac{\pi (\boldsymbol{z})^2  \left[1 - \pi (\boldsymbol{z}) \right]^2  \left \lVert \left(1, \boldsymbol{z}^\top \right) \right \rVert_2^2 }{ \pi_{\text{rare}} (1 -   \pi_{\text{rare}} )  \left \lVert \E \left[ \boldsymbol{\tilde{X}} \boldsymbol{\tilde{X}}^\top \right] \right \rVert_{\text{op}}}
\\ & \stackrel{(*)}{\geq}  \frac{\pi (\boldsymbol{z})^2  \left[1 - \pi_{\text{rare}} \right]^2   }{ \pi_{\text{rare}} (1 -   \pi_{\text{rare}} ) \lambda_{\text{max}}}
\\ & = \frac{\pi (\boldsymbol{z})^2  \left[1 - \pi_{\text{rare}} \right]   }{ \pi_{\text{rare}} \lambda_{\text{max}}}
,
\end{align*}
where \(\lambda_{\text{min}}(\cdot)\) denotes the minimum eigenvalue of \(\cdot\) and \((*)\) uses \(\left \lVert \left(1, \boldsymbol{z}^\top \right) \right \rVert_2^2 \geq 1\) and \(\pi (\boldsymbol{z})  \leq \pi_{\text{rare}}\) for all \(\boldsymbol{z} \in \mathcal{S}\). This yields
\begin{align*}
\mathrm{Asym. } \Var \left( \sqrt{n}\frac{\pi(\boldsymbol{z}) - \hat{\pi}_n(\boldsymbol{z})}{\pi(\boldsymbol{z})}  \right) =  \frac{1}{\pi(\boldsymbol{z})^2}  \mathrm{Asym. } \Var \left(  \sqrt{n} \cdot \hat{\pi}(\boldsymbol{z})  \right)  
 \geq   \frac{1 - \pi_{\text{rare}} }{ \pi_{\text{rare}} \lambda_{\text{max}}}
.
\end{align*}
Similarly to the previous result, \(\hat{\pi}(\boldsymbol{z})  \) is asymptotically unbiased, its finite sample bias \(\pi(\boldsymbol{z}) - \E [ \hat{\pi}_n(\boldsymbol{z})]\) is \(\mathcal{O}(1/n)\) by standard maximum likelihood theory \citep{Cordeiro1991} since it is a maximum likelihood estimator by the functional equivariance of maximum likelihood esimators, and its asymptotic MSE is equal to its asymptotic variance:
\begin{align*}
\mathrm{Asym. MSE}(\hat{\pi}(\boldsymbol{z})) =  ~ & \E \left[ \lim_{n \to \infty} \left( \sqrt{n} \cdot \frac{\pi(\boldsymbol{z}) - \hat{\pi}_n(\boldsymbol{z})}{\pi(\boldsymbol{z})} \right)^2 \right]
\\  =  ~ & \E \left[ \lim_{n \to \infty} \left( \sqrt{n} \cdot \frac{ \E [ \hat{\pi}_n(\boldsymbol{z})] - \hat{\pi}_n(\boldsymbol{z})}{\pi(\boldsymbol{z})} \right)^2 +  \lim_{n \to \infty} \left( \sqrt{n} \cdot \frac{\pi(\boldsymbol{z}) - \E [ \hat{\pi}_n(\boldsymbol{z})]}{\pi(\boldsymbol{z})} \right)^2 \right] 
\\  =  ~ &\mathrm{Asym. } \Var \left( \sqrt{n}\frac{\pi(\boldsymbol{z}) - \hat{\pi}_n(\boldsymbol{z})}{\pi(\boldsymbol{z})}  \right)  + 0
\\ \geq ~ & \frac{1 - \pi_{\text{rare}} }{ \pi_{\text{rare}}\lambda_{\text{max}}}
.
\end{align*}

\end{enumerate}
\end{proof}

\begin{proof}[Proof of Theorem \ref{est.known.beta}] \begin{enumerate} \item Again, the assumptions of Lemma \ref{asym.matrix} are satisfied. Lemma \ref{asym.matrix} shows that the asymptotic covariance matrix of the scaled maximum likelihood estimates of the parameters of logistic regression (a special case of the proportional odds model with \(K = 2\) categories) is
\begin{align*}
\mathrm{Asym. } \Cov \left(  \sqrt{n} \cdot \left( \hat{\alpha}, \boldsymbol{\hat{\beta}}  \right)^\top \right)  & = \left( I^{\text{logistic}} (\boldsymbol{\alpha}, \boldsymbol{\beta}) \right)^{-1}
 = \begin{pmatrix} I_{\alpha \alpha}^{\text{logistic}} & \left( I_{\beta \alpha}^{\text{logistic}} \right)^\top
\\ I_{\beta \alpha}^{\text{logistic}} & I_{\beta \beta}^{\text{logistic}}
\end{pmatrix}^{-1} 
,
\end{align*}
so in the case that \(\boldsymbol{\beta}\) is known, we have
\[
\mathrm{Asym. } \Var \left(  \sqrt{n} \cdot  \hat{\alpha}_q   \right)  = \left(  I_{\alpha \alpha}^{\text{logistic}} \right)^{-1} =  \frac{1}{I_{\alpha \alpha}^{\text{logistic}}}
.
\]
If \(\boldsymbol{\beta}\) is not known, then if \( I_{\beta \beta}^{\text{logistic}}\) is positive definite (and therefore invertible) the formula for block matrix inversion yields
\begin{align}
\mathrm{Asym. } \Var \left(  \sqrt{n} \cdot  \hat{\alpha}   \right) & = \frac{1}{I_{\alpha \alpha}^{\text{logistic}} -  \left( I_{\beta \alpha}^{\text{logistic}} \right)^\top \left(I_{\beta \beta}^{\text{logistic}}\right)^{-1}   I_{\beta \alpha}^{\text{logistic}}} \label{block.mat.inv.lem.form}
.
\end{align}
We know that \( I_{\beta \beta}^{\text{logistic}}\) is positive definite because \(I^{\text{logistic}} (\boldsymbol{\alpha}, \boldsymbol{\beta})\) is finite and positive definite from Lemma \ref{asym.matrix}, so the principal submatrix \( I_{\beta \beta}^{\text{logistic}}\) is positive definite by Observation 7.1.2 in \citet{horn_johnson_2012}. 
%
%
Further, since we know from Lemma \ref{asym.matrix} that the covariance matrix of \((\hat{\alpha}, \boldsymbol{\hat{\beta}})\) is finite and positive definite under our conditions, this also implies that 
\begin{equation}\label{Q.lb.lem}
0 < I_{\alpha \alpha}^{\text{logistic}} -  \left( I_{\beta \alpha}^{\text{logistic}} \right)^\top \left(I_{\beta \beta}^{\text{logistic}}\right)^{-1}   I_{\beta \alpha}^{\text{logistic}} < \infty.
\end{equation}

Now we seek a lower bound for \(\mathrm{Asym. } \Var \left(  \sqrt{n} \cdot  \hat{\alpha}   \right)\). We see from \eqref{block.mat.inv.lem.form} that we can get such a bound by lower-bounding \(\left( I_{\beta \alpha}^{\text{logistic}} \right)^\top \left(I_{\beta \beta}^{\text{logistic}}\right)^{-1}   I_{\beta \alpha}^{\text{logistic}}\). Because \(t \mapsto t(1 -t)\) is upper-bounded by \(1/4\) for all \(t \in [0,1]\),
\begin{align*}
 I_{\beta \beta}^{\text{logistic}} & = \int \boldsymbol{x} \boldsymbol{x}^\top \pi_2(\boldsymbol{x})[1 - \pi_2(\boldsymbol{x})] \ d F(\boldsymbol{x})
 \preceq   \int \boldsymbol{x} \boldsymbol{x}^\top \cdot \frac{1}{4} \ d F(\boldsymbol{x})
=  \frac{1}{4} \E\left[ \boldsymbol{X} \boldsymbol{X}^\top \right].
\end{align*}
Then
\begin{align}
\left( I_{\beta \alpha}^{\text{logistic}} \right)^\top \left(I_{\beta \beta}^{\text{logistic}}\right)^{-1}  I_{\beta \alpha}^{\text{logistic}} \nonumber
 & \geq  \left( I_{\beta \alpha}^{\text{logistic}} \right)^\top \left( \frac{1}{4} \E\left[ \boldsymbol{X} \boldsymbol{X}^\top \right] \right)^{-1}   I_{\beta \alpha}^{\text{logistic}}   \nonumber
\\ & \geq 4 \lambda_{\text{min}}  \left( \left( \E\left[ \boldsymbol{X} \boldsymbol{X}^\top \right]   \right)^{-1} \right) \left \lVert I_{\beta \alpha}^{\text{logistic}} \right \rVert_2^2    \nonumber
\\ & = \frac{4  \left \lVert I_{\beta \alpha}^{\text{logistic}} \right \rVert_2^2 }{ \left \lVert \E\left[ \boldsymbol{X} \boldsymbol{X}^\top \right]   \right \rVert_{\text{op}}}  \nonumber
\\ & \geq \frac{4  \pi_{\text{min}}^2(1 - \pi_{\text{min}})^2 \left \lVert \E \left[ \boldsymbol{X} \right] \right \rVert_2^2} { \lambda_{\text{max}}} = : \Delta , \label{asym.var.ineq.1}
\end{align}
where \(\lambda_{\text{min}}(\cdot)\) denotes the minimum eigenvalue of \(\cdot\) and the last step follows because
\begin{align*}
 \left \lVert I_{\beta \alpha}^{\text{logistic}} \right \rVert_2    =  ~ &  \left \lVert \int \boldsymbol{x} \pi_2(\boldsymbol{x}) [1 - \pi_2(\boldsymbol{x})]  \ d F(\boldsymbol{x}) \right \rVert_2
\\ \geq  ~ &  \left \lVert   \int   \boldsymbol{x}   \pi_{\text{min}}(1 - \pi_{\text{min}})  \ d F(\boldsymbol{x})  \right \rVert_2 
\\ = ~ &  \pi_{\text{min}}(1 - \pi_{\text{min}})  \left \lVert \E \left[ \boldsymbol{X} \right] \right \rVert_2
\end{align*}
(where we used the fact that \(\boldsymbol{X}\) has support only over nonnegative numbers).
Therefore \eqref{block.mat.inv.lem.form} and \eqref{asym.var.ineq.1} yield 
\begin{equation}\label{thm.2.var.result}
 \mathrm{Asym. } \Var \left(  \sqrt{n} \cdot  \hat{\alpha}   \right) \geq   \left(  \frac{1}{\mathrm{Asym. } \Var \left(  \sqrt{n} \cdot  \hat{\alpha}_q   \right)} -  \Delta \right)^{-1}
 .
\end{equation}
The remainder of the argument is similar to the end of the proof of Theorem \ref{log.imb}: the asymptotic unbiasedness and \(\mathcal{O}(1/n)\) finite-sample bias of these estimators yields
\[
\mathrm{Asym. MSE}(\hat{\alpha}) =  \mathrm{Asym. } \Var \left(  \sqrt{n} \cdot  [ \alpha - \hat{\alpha}]  \right) 
\]
and
\[
\mathrm{Asym. MSE}(\hat{\alpha}_q)  = \mathrm{Asym. } \Var \left(  \sqrt{n} \cdot [\alpha -  \hat{\alpha}_q]  \right) 
.
\]
Then from \eqref{Q.lb.lem} we know that
\begin{equation}\label{needed.Q.cond.lem.var}
I_{\alpha \alpha}^{\text{logistic}} >  \left( I_{\beta \alpha}^{\text{logistic}} \right)^\top \left(I_{\beta \beta}^{\text{logistic}}\right)^{-1}   I_{\beta \alpha}^{\text{logistic}}  \geq \frac{4  \pi_{\text{min}}^2(1 - \pi_{\text{min}})^2 \left \lVert \E \left[ \boldsymbol{X} \right] \right \rVert_2^2} { \lambda_{\text{max}}} .
\end{equation}
Making the appropriate substitutions into \eqref{needed.Q.cond.lem.var} yields
\[
\frac{1}{\mathrm{Asym. MSE}(\hat{\alpha}_q) } -  \Delta > 0,
\]
and then substituting into \eqref{thm.2.var.result} yields
\begin{align*}
\mathrm{Asym. MSE}(\hat{\alpha})   \geq  ~ &    \left(  \frac{1}{\mathrm{Asym. MSE}(\hat{\alpha}_q)  } -  \Delta \right)^{-1}%
\\ = ~ & \frac{\mathrm{Asym. MSE}(\hat{\alpha}_q) }{  1  -  \Delta \cdot \mathrm{Asym. MSE}(\hat{\alpha}_q)  }
\\ \stackrel{(*)}{\geq} ~ &  \mathrm{Asym. MSE}(\hat{\alpha}_q) \cdot \left(1 + \Delta \cdot \mathrm{Asym. MSE}(\hat{\alpha}_q) \right)
\\ \iff \qquad \frac{\mathrm{Asym. MSE}(\hat{\alpha})    -  \mathrm{Asym. MSE}(\hat{\alpha}_q)  }{ \left[ \mathrm{Asym. MSE}(\hat{\alpha}_q) \right]^2 } \geq  ~ & \Delta,
\end{align*}
where in \((*)\) we used the inequality \(c/(1-ct) \leq c(1 + c t)\) for any \(c > 0\), \(t < \frac{1}{c}\).

\item 

Because \((\hat{\alpha}, \boldsymbol{\hat{\beta}})  \mapsto \hat{\pi}(\boldsymbol{z})\) is differentiable for all \(\boldsymbol{z} \in \mathbb{R}^p\), by the delta method (Theorem 3.1 in \citealt{van2000asymptotic})
\[
\sqrt{n} \cdot [ \hat{\pi}(\boldsymbol{z}) - \pi(\boldsymbol{z})  ] \xrightarrow{d}   \mathcal{N} \left( 0,  \pi (\boldsymbol{z})^2  \left[1 - \pi (\boldsymbol{z}) \right]^2  \left(1, \boldsymbol{z}^\top \right) \left( I^{\text{logistic}} (\boldsymbol{\alpha}, \boldsymbol{\beta}) \right)^{-1}   \left(1, \boldsymbol{z}^\top \right)^\top \right)
\]
for any \(\boldsymbol{z} \in \mathbb{R}^p\), and similarly
\[
\sqrt{n} \cdot [ \hat{\pi}_q(\boldsymbol{z}) - \pi(\boldsymbol{z})  ] \xrightarrow{d}   \mathcal{N} \left( 0,  \frac{\pi (\boldsymbol{z})^2  \left[1 - \pi (\boldsymbol{z}) \right]^2  }{I_{\alpha \alpha}^{\text{logistic}}} \right)
.
\]
We can find \(\left( I^{\text{logistic}} (\boldsymbol{\alpha}, \boldsymbol{\beta}) \right)^{-1} \) using the formula for block matrix inversion if
\[
\boldsymbol{D} := I_{\beta \beta}^{\text{logistic}} - \frac{I_{\beta \alpha}^{\text{logistic}}\left( I_{\beta \alpha}^{\text{logistic}} \right)^\top}{ I_{\alpha \alpha}^{\text{logistic}} }
\]
is positive definite (and therefore invertible; note that \(\boldsymbol{D}\) is symmetric) and \(I_{\alpha \alpha}^{\text{logistic}} > 0\). We have
\[
I_{\alpha \alpha}^{\text{logistic}} = \int \pi_2(\boldsymbol{x})[1 - \pi_2(\boldsymbol{x})] \ d F(\boldsymbol{x}) \geq    \int \pi_{\text{min}}[1  -\pi_{\text{min}} ]  \ d F(\boldsymbol{x}) = \pi_{\text{min}}[1  -\pi_{\text{min}} ]  > 0,
\] 
and by Theorem 1.12 in \citet{zhang2005schur}, we then know \(\boldsymbol{D}\) is positive definite (and invertible) since \(I^{\text{logistic}} (\boldsymbol{\alpha}, \boldsymbol{\beta})\) is by Lemma \ref{asym.matrix} and \( I_{\alpha \alpha}^{\text{logistic}} > 0\). Let \(\lambda_{\text{max}}^{\boldsymbol{D}} := \lVert \boldsymbol{D} \rVert_{\text{op}}\) be the largest eigenvalue of \(\boldsymbol{D}\); then \(1/\lambda_{\text{max}}^{\boldsymbol{D}} \) is the smallest eigenvalue of \(\boldsymbol{D}^{-1}\). Then for any \(\boldsymbol{z} \in \mathbb{R}^p\), we have
\begin{align}
& \begin{pmatrix} 1 &  \boldsymbol{z}^\top \end{pmatrix} \left( I^{\text{logistic}} (\boldsymbol{\alpha}, \boldsymbol{\beta}) \right)^{-1}   \begin{pmatrix} 1 \\  \boldsymbol{z} \end{pmatrix}   \nonumber
\\ = ~ &  \begin{pmatrix} 1 &  \boldsymbol{z}^\top \end{pmatrix}    \begin{pmatrix}
\frac{1}{I_{\alpha \alpha}^{\text{logistic}}} + \frac{1}{\left(I_{\alpha \alpha}^{\text{logistic}} \right)^2} \left( I_{\beta \alpha}^{\text{logistic}} \right)^\top  \boldsymbol{D}^{-1}   I_{\beta \alpha}^{\text{logistic}}  & - \frac{1}{I_{\alpha \alpha}^{\text{logistic}}} \left( I_{\beta \alpha}^{\text{logistic}} \right)^\top \boldsymbol{D}^{-1}
\\    - \frac{1}{I_{\alpha \alpha}^{\text{logistic}}}   \boldsymbol{D}^{-1} I_{\beta \alpha}^{\text{logistic}}  &  \boldsymbol{D}^{-1}
\end{pmatrix} \begin{pmatrix} 1 \\  \boldsymbol{z} \end{pmatrix}    \nonumber
\\ = ~ &   \frac{1}{I_{\alpha \alpha}^{\text{logistic}}} + \frac{1}{\left(I_{\alpha \alpha}^{\text{logistic}} \right)^2} \left( I_{\beta \alpha}^{\text{logistic}} \right)^\top  \boldsymbol{D}^{-1}   I_{\beta \alpha}^{\text{logistic}}   + \boldsymbol{z}^\top \boldsymbol{D}^{-1} \boldsymbol{z}  - 2 \frac{1}{I_{\alpha \alpha}^{\text{logistic}}}    I_{\beta \alpha}^{\text{logistic}} \boldsymbol{D}^{-1} \boldsymbol{z}  \nonumber
\\ \stackrel{(a)}{=} ~ &  \frac{1}{I_{\alpha \alpha}^{\text{logistic}}}   + \frac{1}{\left(I_{\alpha \alpha}^{\text{logistic}} \right)^2}
  \left \lVert \boldsymbol{D}^{-1/2}   \left(  I_{\beta \alpha}^{\text{logistic}}  - I_{\alpha \alpha}^{\text{logistic}} \boldsymbol{z} \right) 
\right  \rVert_2^2  \nonumber
\\ \stackrel{(b)}{\geq} ~ &  \frac{1}{I_{\alpha \alpha}^{\text{logistic}}}   + \frac{1}{\left(I_{\alpha \alpha}^{\text{logistic}} \right)^2 \lambda_{\text{max}}^{\boldsymbol{D}} }
  \left \lVert   I_{\beta \alpha}^{\text{logistic}}  - I_{\alpha \alpha}^{\text{logistic}} \boldsymbol{z} 
\right  \rVert_2^2 \label{thm.2.result.intmd}
\\ \geq ~ &  \frac{1}{I_{\alpha \alpha}^{\text{logistic}}}   \nonumber
,
\end{align}
where \((a)\) follows from
\begin{align*}
& \frac{1}{\left(I_{\alpha \alpha}^{\text{logistic}} \right)^2}  \left \lVert \boldsymbol{D}^{-1/2}   \left(  I_{\beta \alpha}^{\text{logistic}}  - I_{\alpha \alpha}^{\text{logistic}} \boldsymbol{z} \right) 
\right  \rVert_2^2 
\\ = ~ &  \frac{1}{\left(I_{\alpha \alpha}^{\text{logistic}} \right)^2}  \left(  \left( I_{\beta \alpha}^{\text{logistic}} \right)^\top  \boldsymbol{D}^{-1}   I_{\beta \alpha}^{\text{logistic}}   +  \left(I_{\alpha \alpha}^{\text{logistic}} \right)^2 \boldsymbol{z}^\top \boldsymbol{D}^{-1} \boldsymbol{z}  - 2 I_{\alpha \alpha}^{\text{logistic}}   I_{\beta \alpha}^{\text{logistic}}  \boldsymbol{D}^{-1} \boldsymbol{z}  \right)
\end{align*}
and \((b)\) uses the fact that \(1/\sqrt{ \lambda_{\text{max}}^{\boldsymbol{D}} }\) is the smallest eigenvalue of \(\boldsymbol{D}^{-1/2}\). If we can show that \(  \left \lVert   I_{\beta \alpha}^{\text{logistic}}  - I_{\alpha \alpha}^{\text{logistic}} \boldsymbol{z}  
\right  \rVert_2 \neq 0\), then \eqref{thm.2.result.intmd} is enough to establish the strict inequality in the result. Using \eqref{logistic.fisher}, note that
\begin{align*}
0 = ~ & I_{\beta \alpha}^{\text{logistic}}  - I_{\alpha \alpha}^{\text{logistic}} \boldsymbol{z}  
%
%
\\ = ~ & \int \boldsymbol{x} \pi(\boldsymbol{x}) [1 - \pi(\boldsymbol{x})] \ d F(\boldsymbol{x})  - \boldsymbol{z} \int \pi(\boldsymbol{x}) [1 - \pi(\boldsymbol{x})] \ d F(\boldsymbol{x}) 
\\ \iff \qquad \boldsymbol{z}  = ~ & \frac{\int \boldsymbol{x} \pi(\boldsymbol{x}) [1 - \pi(\boldsymbol{x})] \ d F(\boldsymbol{x})}{\int \pi(\boldsymbol{x}) [1 - \pi(\boldsymbol{x})] \ d F(\boldsymbol{x}) }  
,
\end{align*}
so for all \(\boldsymbol{z} \neq \E \left[ \boldsymbol{X} \pi(\boldsymbol{X})[1 - \pi(\boldsymbol{X})] \right] / \E \left[ \pi(\boldsymbol{X})[1 - \pi(\boldsymbol{X})] \right] \), we have \( \left \lVert    I_{\beta \alpha}^{\text{logistic}}  - I_{\alpha \alpha}^{\text{logistic}} \boldsymbol{z} 
\right  \rVert_2^2 > 0\). So for any \(\boldsymbol{z} \in \mathbb{R}^p \setminus \E \left[ \boldsymbol{X} \pi(\boldsymbol{X})[1 - \pi(\boldsymbol{X})] \right] / \E \left[ \pi(\boldsymbol{X})[1 - \pi(\boldsymbol{X})] \right]\), we have
\begin{align*}
\pi (\boldsymbol{z})^2  \left[1 - \pi (\boldsymbol{z}) \right]^2  \begin{pmatrix} 1 & \boldsymbol{z}^\top \end{pmatrix}  \begin{pmatrix} I_{\alpha \alpha}^{\text{logistic}} & \left( I_{\beta \alpha}^{\text{logistic}} \right)^\top
\\ I_{\beta \alpha}^{\text{logistic}} & I_{\beta \beta}^{\text{logistic}}
\end{pmatrix}^{-1}    \begin{pmatrix} 1 \\  \boldsymbol{z} \end{pmatrix} 
> ~ & \frac{\pi (\boldsymbol{z})^2  \left[1 - \pi (\boldsymbol{z}) \right]^2 }{I_{\alpha \alpha}^{\text{logistic}}} 
\\ \iff \qquad \mathrm{Asym. } \Var \left(  \sqrt{n} \cdot \hat{\pi}_q(\boldsymbol{z})  \right) < ~ &   \mathrm{Asym. } \Var \left(  \sqrt{n} \cdot  (\hat{\pi}(\boldsymbol{z}) \right) 
.
\end{align*}

\end{enumerate}
\end{proof}

\section{Investigation of Plausibility of Assumption \eqref{delta.small.assum.fin} and Proof of Theorem \ref{main.cov.thm.2}}\label{main.thm.sec}

Before we prove Theorem \ref{main.cov.thm.2}, we begin by investigating the plausibility of Assumption \eqref{delta.small.assum.fin}. We start by investigating whether \( \lambda_{\text{min}} \left( I_{\beta \beta} -   2 \frac{I_{\beta \alpha_1} I_{\beta \alpha_1}^\top}{I_{\alpha_1 \alpha_1}} \right) \) is bounded away from 0 in Section \ref{min.eigen.bound}. Then in Section \ref{min.eigen.cond.sim} we directly investigate whether Assumption \eqref{delta.small.assum.fin} seems to hold in a variety of contexts. We prove Theorem \ref{main.cov.thm.2} in Section \ref{proof.thm.cov.sec}, and we prove the supporting result Lemma \ref{lem.ineq.final.thm} in Section \ref{sec.lem.ineq.final.thm}.


\subsection{Investigating Whether \( \lambda_{\text{min}} \left( I_{\beta \beta} -   2 \frac{I_{\beta \alpha_1} I_{\beta \alpha_1}^\top}{I_{\alpha_1 \alpha_1}} \right) \) is Bounded Away From Zero}\label{min.eigen.bound}


To investigate the plausibility of the assumption that \( I_{\beta \beta} -   2 \frac{I_{\beta \alpha_1} I_{\beta \alpha_1}^\top}{I_{\alpha_1 \alpha_1}} \) is positive definite (and that the upper bound for \(\pi_{\text{rare}}\) in Equation \ref{delta.small.assum.fin} can hold) empirically, we perform two simulation studies, with setups similar to those of our synthetic data experiments in Section \ref{sec.sim} of the paper. We repeat Simulation Study A 25 times. Using \(n = 10^6\), \(p = 10\), and \(K = 3\), we generate \(\boldsymbol{X} \in \mathbb{R}^{n \times p}\) with \(\boldsymbol{X}_{ij}\sim \operatorname{Uniform}(-1, 1)\) iid for all \(i \in \{1, \ldots, n\}\) and \(j \in \{1, \ldots, p\}\). We then generate \(y_i \in \{1, 2, 3\}\) from \(\boldsymbol{x}_i\) for each \(i \in \{1, \ldots, n\}\) according to the proportional odds model \eqref{prop_odds}, using \(\boldsymbol{\beta} = (1, \ldots, 1)^\top\) and \(\boldsymbol{\alpha} = (0, 20)\) (so that class 3 is very rare, with \(\pi_{\text{rare}} \approx 4.54 \cdot 10^{-5}\)). We estimate \( I_{\beta \beta}\), \( I_{\beta \alpha_1}\), and \(I_{\alpha_1 \alpha_1}\) using empirical (``plug-in") estimates of the expressions in \eqref{alpha.block}, \eqref{alpha.beta.block}, and \eqref{beta.block}; for example, using \eqref{i.beta.beta.id} we estimated \( I_{\beta \beta}\) by
\[
\frac{1}{n} \sum_{i=1}^n \sum_{k=1}^{K-1}   \boldsymbol{x}_i \boldsymbol{x}_i^\top \left[ \pi_k(\boldsymbol{x}_i) + \pi_{k+1}(\boldsymbol{x}_i) \right] p_k (\boldsymbol{x}_i) [ 1 - p_k (\boldsymbol{x}_i) ] 
.
\]
(Note that we used the true \(\pi_k(\cdot)\) and \(p_k(\cdot)\) functions in these estimates.) Finally, we use these estimated quantities to estimate \(I_{\beta \beta} -   2 \frac{I_{\beta \alpha_1} I_{\beta \alpha_1}^\top}{I_{\alpha_1 \alpha_1}} \), and we calculate the minimum eigenvalue of this estimated matrix. Across all 25 simulations, the sample mean of this minimum eigenvalue is 0.02361, and the minimum is 0.02359. The standard error is \(2.94 \cdot 10^{-6}\), and the \(95\%\) confidence interval for the mean of the minimum eigenvalue is \((0.02360, 0.02362)\). See Figure \ref{min_lambda_plot} for a boxplot of the 25 estimated minimum eigenvalues. These results seem to suggest that the assumption that \( \lambda_{\text{min}} \left( I_{\beta \beta} -   2 \frac{I_{\beta \alpha_1} I_{\beta \alpha_1}^\top}{I_{\alpha_1 \alpha_1}} \right) > 0 \) is reasonable under the assumptions of Theorem \ref{main.cov.thm.2} for \(\boldsymbol{X}\) with this iid uniform distribution. 

\begin{figure}[htbp]
\begin{center}
\includegraphics{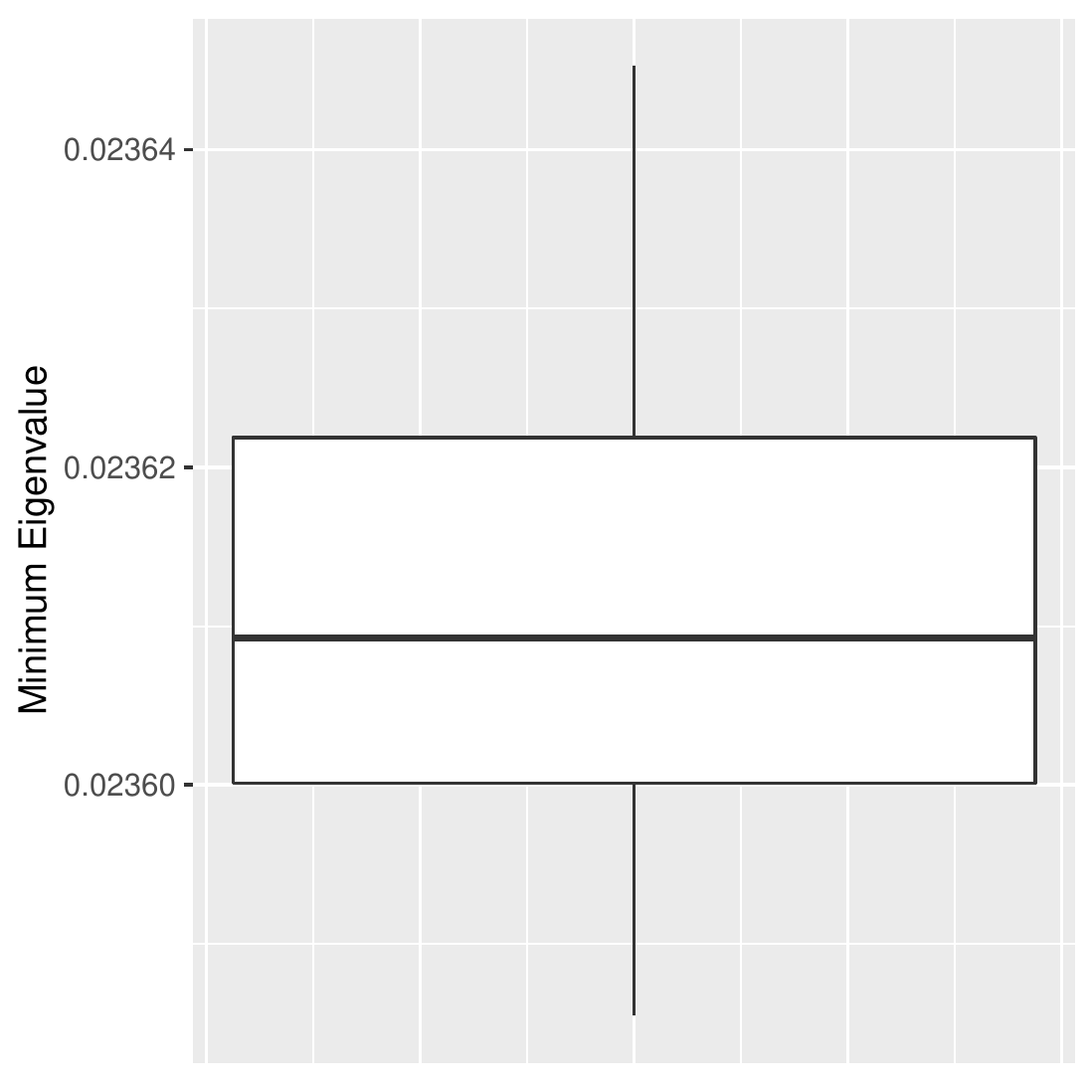}
\caption{Boxplot of the estimated minimum eigenvalues of \(I_{\beta \beta} -   2 \frac{I_{\beta \alpha_1} I_{\beta \alpha_1}^\top}{I_{\alpha_1 \alpha_1}} \) in Simulation Study A as described in Section \ref{min.eigen.bound}.}
\label{min_lambda_plot}
\end{center}
\end{figure}

Next, we conduct another simulation study to investigate whether \( \lambda_{\text{min}} \left( I_{\beta \beta} -   2 \frac{I_{\beta \alpha_1} I_{\beta \alpha_1}^\top}{I_{\alpha_1 \alpha_1}} \right)  \) is bounded away from 0 if the features in \(\boldsymbol{X}\) are correlated, and over a wider range of \(\pi_{\text{rare}}\). In Simulation Study B, we generate matrices of standard multivariate Gaussian data truncated entrywise between \(-3\) and \(3\), with \(n = 10^6\) and \(p = 10\). We generate matrices with covariance matrices
\[
\begin{pmatrix}
1 & \rho & \cdots & \rho
\\ \rho & 1 & \cdots & \rho
\\ \vdots & \vdots & \ddots & \vdots
\\ \rho & \rho & \cdots & 1
\end{pmatrix}
\]
for \(\rho \in \{0, 0.25, 0.5, 0.75\}\). We then generate ordinal responses from the proportional odds model with \(K = 3\) and \(\boldsymbol{\beta} = (1, 1, \ldots, 1)^\top\). The intercept for the first separating hyperplane is 0 (so that the expected class probabilities for classes 1 and 2 are both close to \(1/2\)) and we choose three different values for the second intercept by analytically solving for intercepts such that \(\pi_{\text{rare}} = \sup_{x \in [-3,3]^p}\{\pi_3(\boldsymbol{x})\}\) equals \(10^{-5}, 10^{-6}\), or \(10^{-7}\). To see the behavior as \(\pi_{\text{rare}}\) vanishes, we also investigate a fourth setting where all of the \(\pi_3(\boldsymbol{x})\) terms in the Fisher information are simply set to 0 (and \(\pi_2(\boldsymbol{x})\) is coerced to equal \(1 - \pi_1(\boldsymbol{x})\) pointwise), as if the intercept for the second separating hyperplane were infinity. In each setting, we generate 7 random matrices and estimate \(\lambda_{\text{min}} \left( I_{\beta \beta} -   2 \frac{I_{\beta \alpha_1} I_{\beta \alpha_1}^\top}{I_{\alpha_1 \alpha_1}} \right)\) in the same way as in Simulation Study A.

The results are displayed in Figures \ref{min_eigen_corr_0}, \ref{min_eigen_corr_025}, \ref{min_eigen_corr_05}, and \ref{min_eigen_corr_075} for correlations of 0, 0.25, 0.5, and 0.75, respectively. In each plot the rare class probability is displayed on the horizontal axis on a log scale (with a line break after the far left result for the asymptotic case with \(\pi_{\text{rare}} = 0\)), and the average minimum eigenvalue for each of the seven random matrices is displayed on the vertical axis. We see that the means of the minimum eigenvalues are well above 0; we also note that in every generated matrix in every setting, the minimum eigenvalue was strictly positive.

\begin{figure}[htbp]
\begin{center}
\includegraphics[scale=0.7]{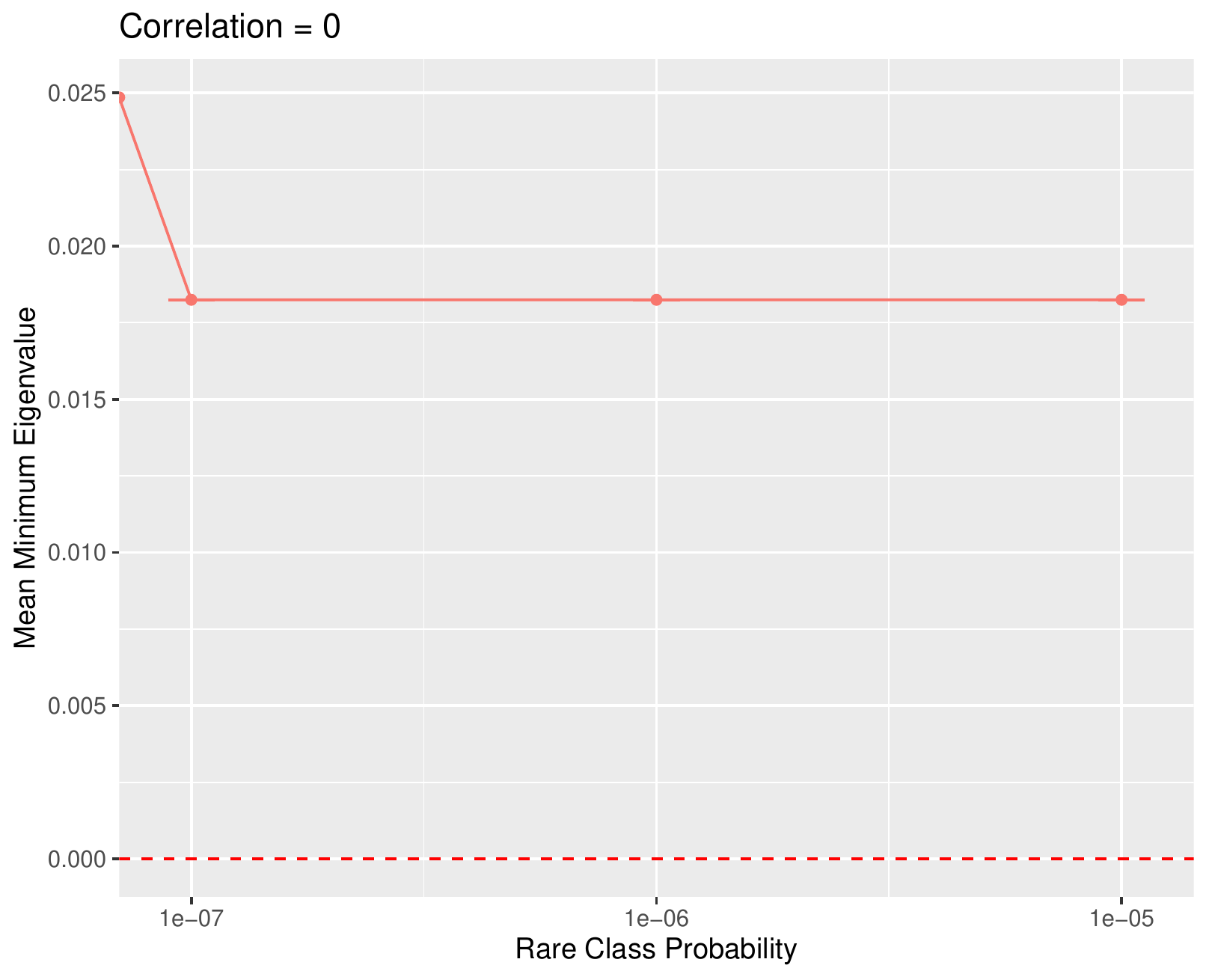}
\caption{Plot of the average estimated minimum eigenvalues of \(I_{\beta \beta} -   2 \frac{I_{\beta \alpha_1} I_{\beta \alpha_1}^\top}{I_{\alpha_1 \alpha_1}} \) in Simulation Study B as described in Section \ref{min.eigen.bound} with each column of \(\boldsymbol{X}\) having correlation 0 with all of the others.}
\label{min_eigen_corr_0}
\end{center}
\end{figure}

\begin{figure}[htbp]
\begin{center}
\includegraphics[scale=0.7]{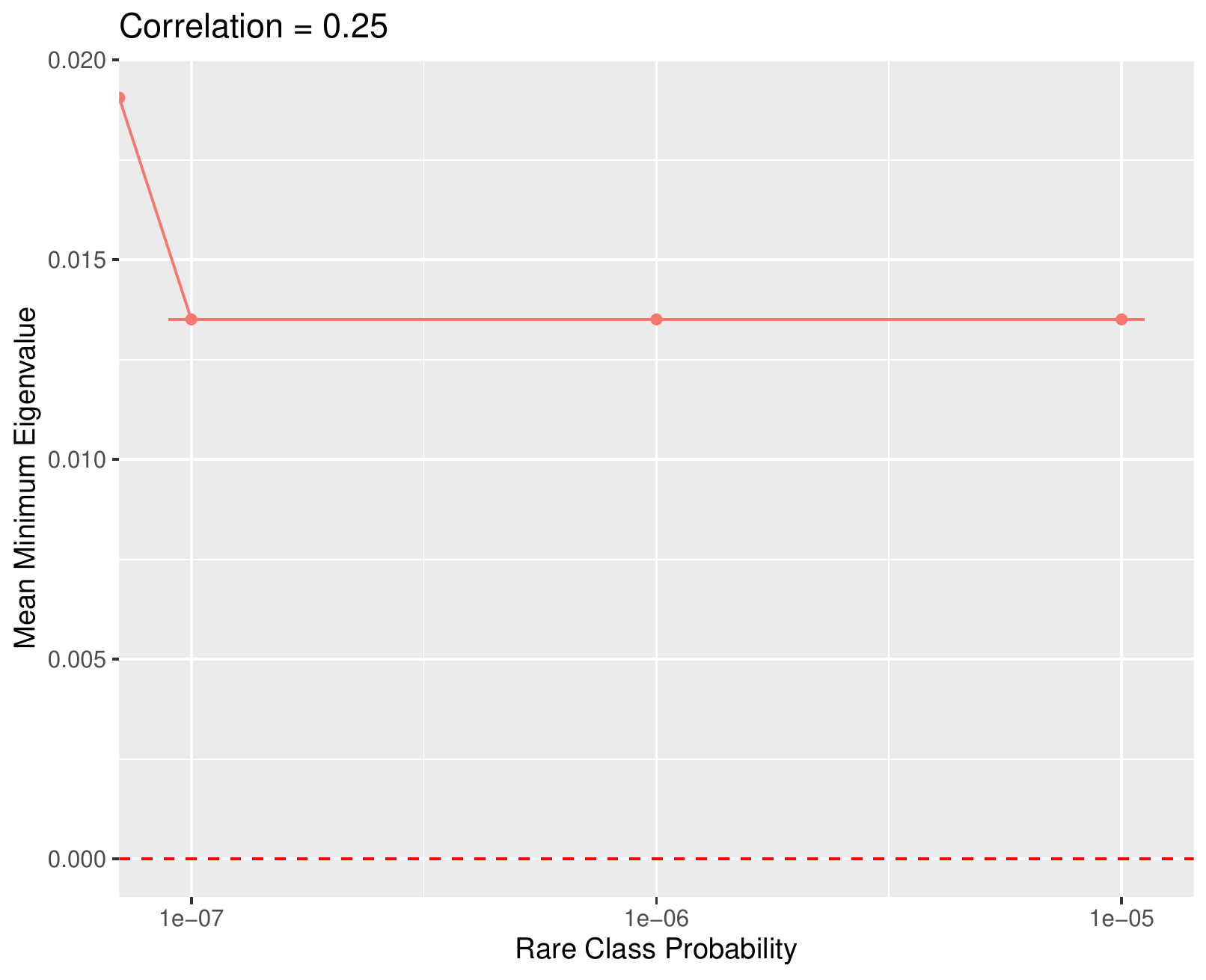}
\caption{Plot of the average estimated minimum eigenvalues of \(I_{\beta \beta} -   2 \frac{I_{\beta \alpha_1} I_{\beta \alpha_1}^\top}{I_{\alpha_1 \alpha_1}} \) in Simulation Study B as described in Section \ref{min.eigen.bound} with each column of \(\boldsymbol{X}\) having correlation 0.25 with all of the others.}
\label{min_eigen_corr_025}
\end{center}
\end{figure}

\begin{figure}[htbp]
\begin{center}
\includegraphics[scale=0.7]{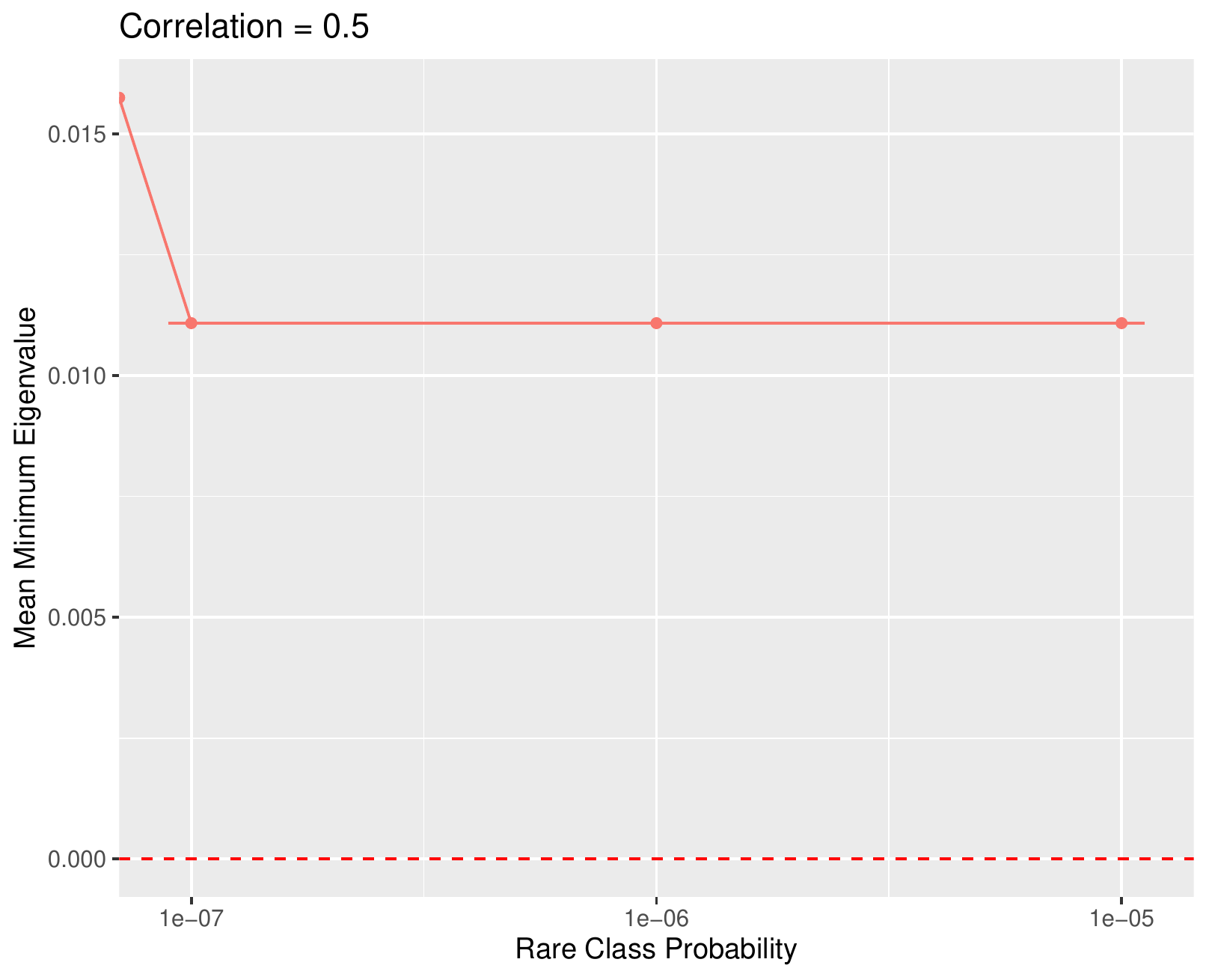}
\caption{Plot of the average estimated minimum eigenvalues of \(I_{\beta \beta} -   2 \frac{I_{\beta \alpha_1} I_{\beta \alpha_1}^\top}{I_{\alpha_1 \alpha_1}} \) in Simulation Study B as described in Section \ref{min.eigen.bound} with each column of \(\boldsymbol{X}\) having correlation 0.5 with all of the others.}
\label{min_eigen_corr_05}
\end{center}
\end{figure}

\begin{figure}[htbp]
\begin{center}
\includegraphics[scale=0.7]{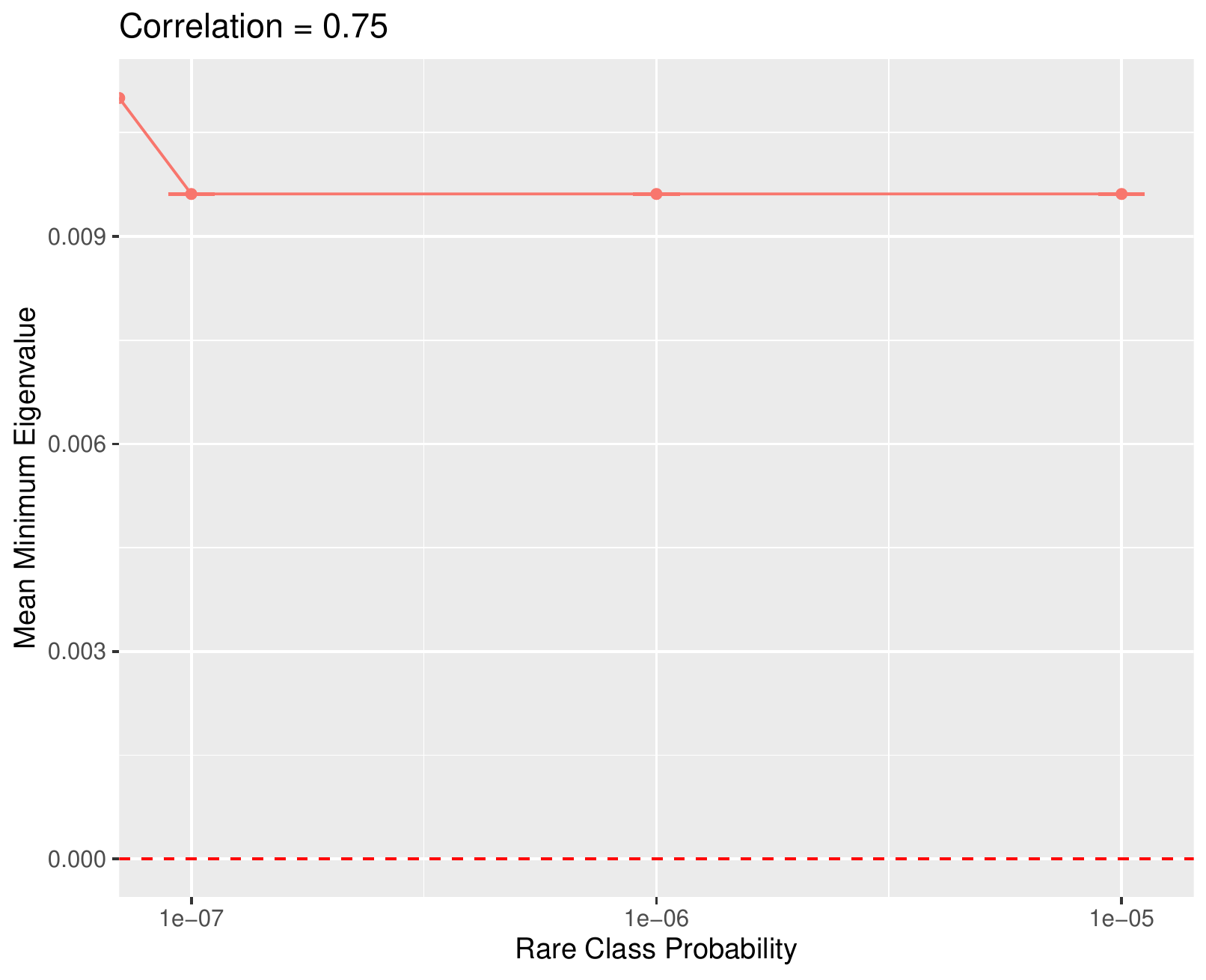}
\caption{Plot of the average estimated minimum eigenvalues of \(I_{\beta \beta} -   2 \frac{I_{\beta \alpha_1} I_{\beta \alpha_1}^\top}{I_{\alpha_1 \alpha_1}} \) in Simulation Study B as described in Section \ref{min.eigen.bound} with each column of \(\boldsymbol{X}\) having correlation 0.75 with all of the others.}
\label{min_eigen_corr_075}
\end{center}
\end{figure}

Finally, we also investigate analytically whether \( \lambda_{\text{min}} \left( I_{\beta \beta} -   2 \frac{I_{\beta \alpha_1} I_{\beta \alpha_1}^\top}{I_{\alpha_1 \alpha_1}} \right)  \) is bounded away from 0 as \(\pi_{\text{rare}}\) becomes arbitrarily small. To see that this seems reasonable, note that from \eqref{i.beta.beta.id} we have 
\begin{align*}
I_{\beta \beta} = ~ &  \sum_{k=1}^{2}  \int \boldsymbol{x} \boldsymbol{x}^\top \left[ \pi_k(\boldsymbol{x}) + \pi_{k+1}(\boldsymbol{x}) \right] p_k (\boldsymbol{x}) [ 1 - p_k (\boldsymbol{x}) ] \ d F(\boldsymbol{x})
\\ = ~ &   \int \boldsymbol{x} \boldsymbol{x}^\top  \left( \left[ \pi_1(\boldsymbol{x}) + \pi_{2}(\boldsymbol{x}) \right] \pi_1 (\boldsymbol{x}) [ 1 - \pi_1 (\boldsymbol{x})] + \left[ \pi_2(\boldsymbol{x}) + \pi_{3}(\boldsymbol{x}) \right] \pi_3 (\boldsymbol{x}) [ 1 - \pi_3 (\boldsymbol{x}) ] \right) \ d F(\boldsymbol{x})
\\ = ~ &   \int \boldsymbol{x} \boldsymbol{x}^\top  \left( \left[ 1 -  \pi_{3}(\boldsymbol{x}) \right] \pi_1 (\boldsymbol{x}) [ 1 - \pi_1 (\boldsymbol{x})] + \left[ 1 -  \pi_{1}(\boldsymbol{x}) \right] \pi_3 (\boldsymbol{x}) [ 1 - \pi_3 (\boldsymbol{x}) ] \right) \ d F(\boldsymbol{x})
\\ = ~ &   \int \boldsymbol{x} \boldsymbol{x}^\top   [ 1 - \pi_1 (\boldsymbol{x})]   [ 1 - \pi_3 (\boldsymbol{x}) ]  \left( \pi_1 (\boldsymbol{x})+ \pi_3 (\boldsymbol{x})\right) \ d F(\boldsymbol{x})
\\ = ~ &   \int \boldsymbol{x} \boldsymbol{x}^\top   [  \pi_2 (\boldsymbol{x}) +  \pi_3 (\boldsymbol{x})]   [1 -  \pi_3 (\boldsymbol{x}) ]  [ \pi_1 (\boldsymbol{x})+ \pi_3 (\boldsymbol{x}) ] \ d F(\boldsymbol{x})
%
%
,
\end{align*}
which is non-vanishing in \(\pi_{\text{rare}}\). (In particular, there seems to be no reason to suspect that the eigenvalues of \(I_{\beta \beta}\) change drastically for, say \( \pi_3 (\boldsymbol{x}) \leq  10^{-7}\) for all \(\boldsymbol{x} \in \mathcal{S}\), as in Simulation Study B, versus \( \pi_3 (\boldsymbol{x}) \leq 10^{-20}\) for all \(\boldsymbol{x} \in \mathcal{S}\).) Further, \eqref{h} below shows that
\[
 \left \lVert \frac{I_{\beta \alpha_1} I_{\beta \alpha_1}^\top}{I_{\alpha_1 \alpha_1}} \right\rVert_{\text{op}} =   \frac{ \lVert I_{\beta \alpha_1} \rVert_2^2}{ I_{\alpha_1 \alpha_1}}
\]
(where \(   \left \lVert  \cdot \right\rVert_{\text{op}}  \) is the operator norm) is bounded from above by a constant not depending on \(\pi_{\text{rare}}\). This also suggests that \(\lambda_{\text{min}} \left( I_{\beta \beta} -   2 \frac{I_{\beta \alpha_1} I_{\beta \alpha_1}^\top}{I_{\alpha_1 \alpha_1}} \right)\) should be insensitive to \(\pi_{\text{rare}}\) becoming vanishingly small. Taken together, this suggests that Assumption \eqref{ub.new} does not become more implausible as \(\pi_{\text{rare}}\) becomes arbitrarily small.

\subsection{Investigating Whether Assumption \eqref{delta.small.assum.fin} is Plausible}\label{min.eigen.cond.sim}

Having developed evidence that \( \lambda_{\text{min}} \left( I_{\beta \beta} -   2 \frac{I_{\beta \alpha_1} I_{\beta \alpha_1}^\top}{I_{\alpha_1 \alpha_1}} \right) \) is bounded away from 0, we now directly investigate the plausibility of the assumption
\begin{equation}\label{ub.new}
 \pi_{\text{rare}} \leq 
  \frac{  \lambda_{\text{min}} \left(I_{\beta \beta} -   2 \frac{I_{\beta \alpha_1} I_{\beta \alpha_1}^\top}{I_{\alpha_1 \alpha_1}}  \right)  }{3 M^2   \left(   2   
 +  M \right) }
 .
\end{equation}
First we examine this in the context of Simulation Study A. Note that \(M = \lVert (1, \ldots, 1) \rVert_2 = \sqrt{p}\), \(\mathcal{S} = [-1, 1]^p\) and
\[
\sup_{\boldsymbol{x} \in \mathcal{S}} \{\pi_3(\boldsymbol{x}) \} =  \pi_3(\boldsymbol{x}^*)  =  1 - \frac{1}{1+ \exp\{-(20 + \sum_{j=1}^p  x_j^* )\}},
\]
for \(\boldsymbol{x}^* = (-1, \ldots, -1)\). So all of the quantities in \eqref{ub.new} are known except the minimum eigenvalue, which we were able to estimate with seemingly high precision. It turns out that \eqref{ub.new} is satisfied even when we use the minimum value of \(\lambda_{\text{min}}\) across all 25 simulations as our estimate; in this case, the left side of \eqref{ub.new} is \(4.54 \cdot 10^{-5}\) and the right side is \(1.52 \cdot 10^{-4}\).

Next we investigate \eqref{ub.new} in the context of Simulation Study B. In each of the four correlation settings, we calculate the proportion of generated random matrices in which \eqref{ub.new} is satisfied at each rarity level (according to our minimum eigenvalue estimates). The results are displayed in Figures \ref{min_eigen_cond_corr_0}, \ref{min_eigen_cond_corr_025}, \ref{min_eigen_cond_corr_05}, and \ref{min_eigen_cond_corr_075}. We see that \eqref{ub.new} is not satisfied for any of the random matrices where \(\pi_{\text{rare}} = 10^{-5}\), but it is satisfied for all of the matrices with all of the remaining values of \(\pi_{\text{rare}}\) regardless of correlation of the features. This suggests that assumption \eqref{delta.small.assum.fin} is indeed reasonable for \(\pi_{\text{rare}}\) small enough.

\begin{figure}[htbp]
\begin{center}
\includegraphics[scale=0.7]{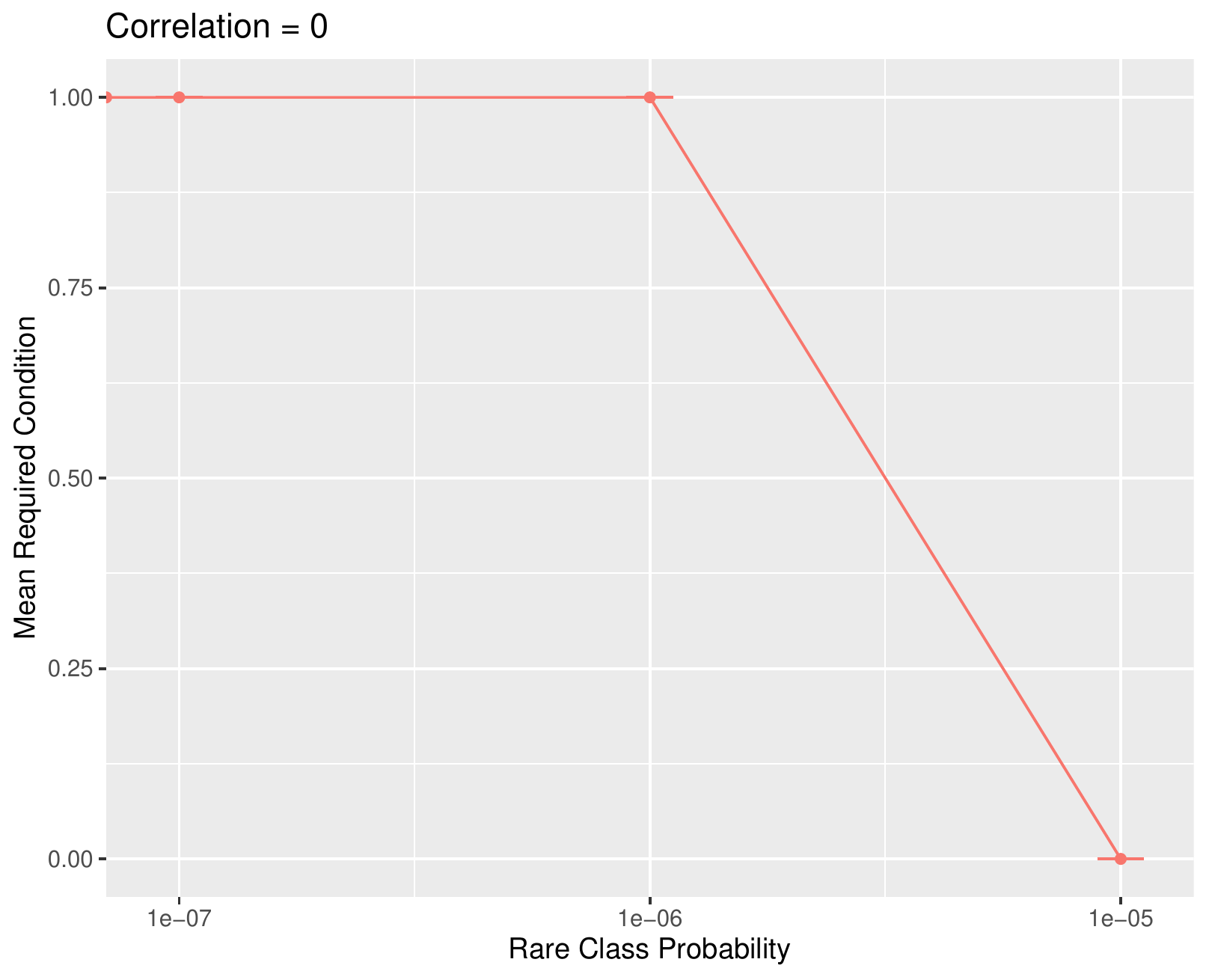}
\caption{Plot of the proportion of random matrices in Simulation Study B (as described in Section \ref{min.eigen.bound}) in which \eqref{ub.new} is satisfied in the setting where each column of \(\boldsymbol{X}\) has correlation 0 with all of the others.}
\label{min_eigen_cond_corr_0}
\end{center}
\end{figure}

\begin{figure}[htbp]
\begin{center}
\includegraphics[scale=0.7]{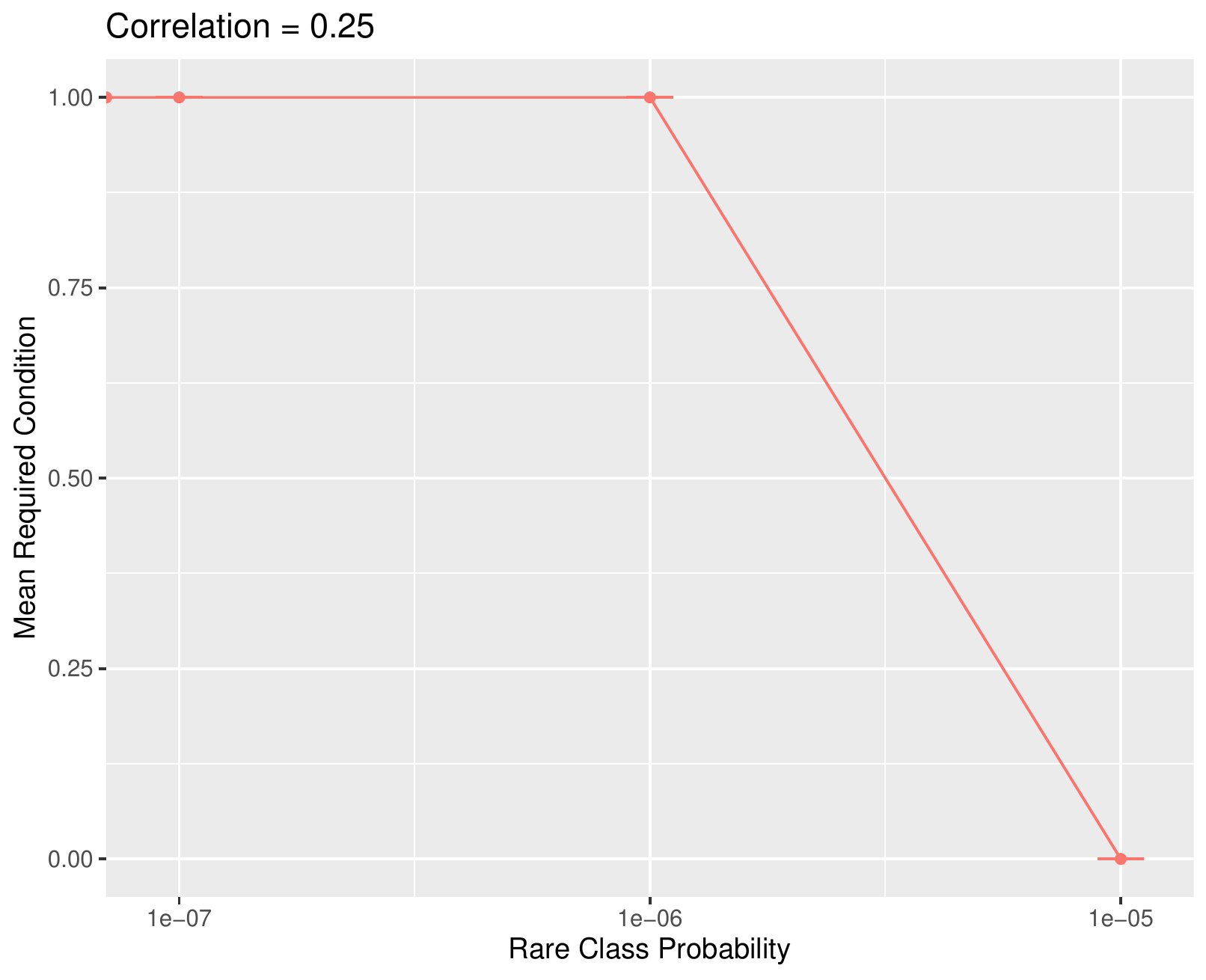}
\caption{Plot of the proportion of random matrices in Simulation Study B (as described in Section \ref{min.eigen.bound}) in which \eqref{ub.new} is satisfied in the setting where each column of \(\boldsymbol{X}\) has correlation 0.25 with all of the others.}
\label{min_eigen_cond_corr_025}
\end{center}
\end{figure}

\begin{figure}[htbp]
\begin{center}
\includegraphics[scale=0.7]{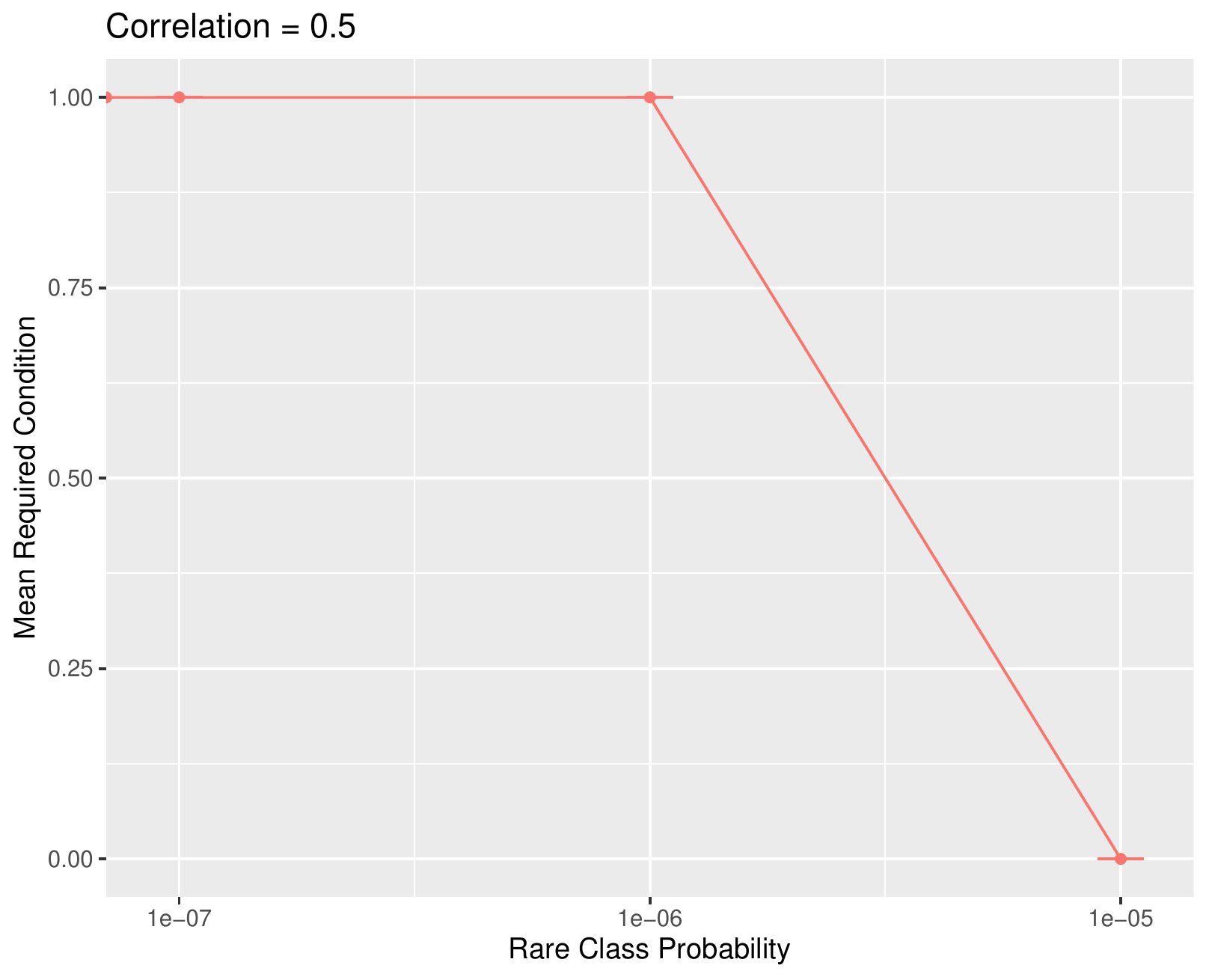}
\caption{Plot of the proportion of random matrices in Simulation Study B (as described in Section \ref{min.eigen.bound}) in which \eqref{ub.new} is satisfied in the setting where each column of \(\boldsymbol{X}\) has correlation 0.5 with all of the others.}
\label{min_eigen_cond_corr_05}
\end{center}
\end{figure}

\begin{figure}[htbp]
\begin{center}
\includegraphics[scale=0.7]{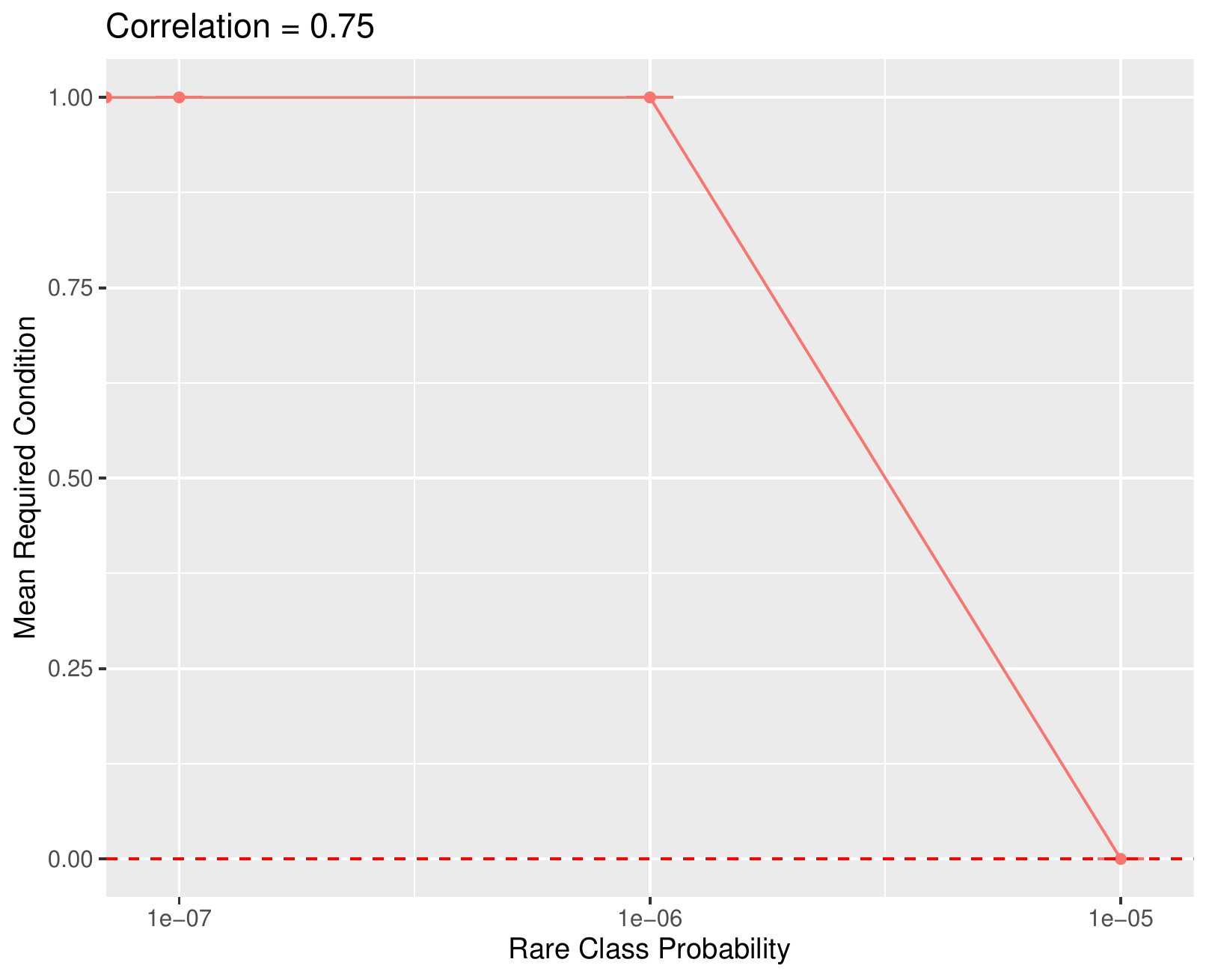}
\caption{Plot of the proportion of random matrices in Simulation Study B (as described in Section \ref{min.eigen.bound}) in which \eqref{ub.new} is satisfied in the setting where each column of \(\boldsymbol{X}\) has correlation 0.75 with all of the others.}
\label{min_eigen_cond_corr_075}
\end{center}
\end{figure}

\subsection{Proof of Theorem \ref{main.cov.thm.2}}\label{proof.thm.cov.sec}

Before we proceed with the proof of Theorem \ref{main.cov.thm.2}, we state a lemma with inequalities we will use.

\begin{lemma}\label{lem.ineq.final.thm} The following inequalities hold under the assumptions of Theorem \ref{main.cov.thm.2}:
\begin{align}
I_{\alpha_2 \alpha_2} \leq ~ & \pi_{\text{rare}} I_{\alpha_1 \alpha_1} \cdot \frac{1}{1/4 - \Delta^2}, \label{a}
\\ I_{\alpha_2 \alpha_2} \leq ~ & \pi_{\text{rare}} \left(1 + \frac{\pi_{\text{rare}}}{1/2 - \Delta} \right), \label{b}
\\ \left| I_{\alpha_1 \alpha_2}  \right| \leq ~ & I_{\alpha_2 \alpha_2},  \label{e}
\\ \left \lVert I_{\beta \alpha_2} \right \rVert_2 \leq ~ & M \cdot I_{\alpha_2 \alpha_2},   \label{f}
\\  \frac{\left \lVert I_{\beta \alpha_2} \right \rVert_2^2}{I_{\alpha_2 \alpha_2}} \leq ~ & M^2 \pi_{\text{rare}} \left(1 + \frac{\pi_{\text{rare}}}{1/2 - \Delta} \right), \qquad \text{and} \label{g}
\\ \frac{ \lVert I_{\beta \alpha_1} \rVert_2^2}{ I_{\alpha_1 \alpha_1}} \leq ~ &  \frac{M^2}{4} \label{h}
.
\end{align}

\end{lemma}
\begin{proof}
Provided immediately after the proof of Theorem \ref{main.cov.thm.2}, in Section \ref{sec.lem.ineq.final.thm}.
\end{proof}

By an argument analogous to the one used at the end of the proof of Theorem \ref{log.imb}, it is enough to upper-bound 
\[
 \left \lVert \Cov \left( \lim_{n \to \infty}  \sqrt{n} \left[ \boldsymbol{\hat{\beta}}^{\text{prop. odds}} - \boldsymbol{\beta}  \right] \right) \right\rVert_{\text{op}} = \left \lVert \mathrm{Asym. } \Cov \left( \sqrt{n} \cdot \boldsymbol{\hat{\beta}}^{\text{prop. odds}} \right) \right\rVert_{\text{op}}.
 \]
Using Lemma \ref{asym.matrix} and the block matrix inversion formula,
\begin{align}
\mathrm{Asym. } \Cov \left( \sqrt{n} \cdot \boldsymbol{\hat{\beta}}^{\text{prop. odds}} \right) 
= ~ & \left( I_{\beta \beta} - I_{\alpha \beta}^\top I_{\alpha \alpha}^{-1} I_{\alpha \beta} \right)^{-1} \nonumber
\\ \implies \qquad \left \lVert \mathrm{Asym. } \Cov \left( \sqrt{n} \cdot \boldsymbol{\hat{\beta}}^{\text{prop. odds}} \right) \right\rVert_{\text{op}}
= ~ & \frac{1}{\lambda_{\text{min}} \left( I_{\beta \beta} - I_{\alpha \beta}^\top I_{\alpha \alpha}^{-1} I_{\alpha \beta} \right)} \label{final.cov.res}
,
\end{align}
where \(\lambda_{\text{min}} (\cdot)\) is the minimum eigenvalue of \(\cdot\) and \(\left \lVert \cdot \right \rVert_{\text{op}} \) is the operator norm. \(I_{\alpha \alpha}\) is a \(2 \times 2\) matrix, so
\[
I_{\alpha \alpha}^{-1} = \frac{1}{\operatorname{det} \left( I_{\alpha \alpha} \right)}  \begin{pmatrix} I_{\alpha_2 \alpha_2} & - I_{\alpha_1 \alpha_2} \\
-I_{\alpha_2 \alpha_1} & I_{\alpha_1 \alpha_1}
\end{pmatrix},
\]
and
\begin{align}
I_{\alpha \beta}^\top I_{\alpha \alpha}^{-1} I_{\alpha \beta}  = ~ &   \frac{1}{\operatorname{det} \left( I_{\alpha \alpha} \right)} \begin{pmatrix}I_{\beta \alpha_1} & I_{\beta \alpha_2} \end{pmatrix}  \begin{pmatrix} I_{\alpha_2 \alpha_2} & - I_{\alpha_1 \alpha_2} \\
-I_{\alpha_2 \alpha_1} & I_{\alpha_1 \alpha_1}
\end{pmatrix}  \begin{pmatrix}I_{\alpha_1 \beta} \\ I_{\alpha_2 \beta} \end{pmatrix}  \nonumber
\\ = ~ &   \frac{1}{\operatorname{det} \left( I_{\alpha \alpha} \right)} \left(  I_{\alpha_2 \alpha_2}  I_{\beta \alpha_1} I_{\beta\alpha_1 }^\top +  I_{\alpha_1 \alpha_1} I_{\beta \alpha_2} I_{\beta \alpha_2}^\top - I_{\alpha_1 \alpha_2} I_{\beta \alpha_1}  I_{ \beta \alpha_2}^\top -   I_{\alpha_1 \alpha_2}I_{\beta \alpha_2} I_{\beta \alpha_1}^\top \right) \nonumber
\\ = ~ &   \frac{1}{\operatorname{det} \left( I_{\alpha \alpha} \right)} \left(  I_{\alpha_2 \alpha_2}  I_{\beta \alpha_1} I_{\beta\alpha_1 }^\top +  I_{\alpha_1 \alpha_1} I_{\beta \alpha_2} I_{\beta \alpha_2}^\top + | I_{\alpha_1 \alpha_2} | I_{\beta \alpha_1}  I_{ \beta \alpha_2}^\top +  | I_{\alpha_1 \alpha_2} | I_{\beta \alpha_2} I_{\beta \alpha_1}^\top \right) \label{q.def}
,
\end{align}
where in the last step we used that \( I_{\alpha_1 \alpha_2} = -\tilde{M}_2 < 0\), which is clear from \eqref{alpha.beta.block} and \eqref{m.tilde.def}. Because we know from Lemma \ref{asym.matrix} that \(I_{\alpha \alpha}\) (and therefore also \(I_{\alpha \alpha}^{-1}\)) is positive definite, by Observation 7.1.6 in \citet{horn_johnson_2012} \(I_{\alpha \beta}^\top I_{\alpha \alpha}^{-1} I_{\alpha \beta} \) is positive definite as well. Therefore we can use an upper bound on \(1/\operatorname{det} \left( I_{\alpha \alpha} \right) \) to upper bound \(I_{\alpha \beta}^\top I_{\alpha \alpha}^{-1} I_{\alpha \beta}  \). We have
\begin{align}
\operatorname{det} \left( I_{\alpha \alpha} \right) = ~ & I_{\alpha_1 \alpha_1} I_{\alpha_2 \alpha_2} - I_{\alpha_1 \alpha_2}^2 \nonumber
\\ \stackrel{(a)}{\geq} ~ & I_{\alpha_1 \alpha_1} I_{\alpha_2 \alpha_2} - I_{\alpha_2 \alpha_2}^2  \nonumber
\\ = ~ &  I_{\alpha_2 \alpha_2} \left( I_{\alpha_1 \alpha_1}  - I_{\alpha_2 \alpha_2} \right)  \nonumber
\\ \stackrel{(b)}{\geq} ~ &  I_{\alpha_2 \alpha_2} \left( I_{\alpha_1 \alpha_1}  -  \pi_{\text{rare}} I_{\alpha_1 \alpha_1} \cdot \frac{1}{1/4 - \Delta^2} \right)  \nonumber
\\ = ~ & \left( 1 - \frac{\pi_{\text{rare}}}{1/4 - \Delta^2} \right) I_{\alpha_1 \alpha_1} I_{\alpha_2 \alpha_2}  \nonumber
\\ \stackrel{(c)}{\geq} ~ & \frac{1}{2} I_{\alpha_1 \alpha_1} I_{\alpha_2 \alpha_2}  \nonumber
,
\end{align}
where in \((a)\) we used \eqref{e}, in \((b)\) we used \eqref{a}, and \((c)\) uses that from the upper bound \eqref{delta.small.assum.fin} for \(\pi_{\text{rare}}\) we have
\begin{align*}
\pi_{\text{rare}} \leq ~ & \frac{1}{2}\left(\frac{1}{2} - \Delta\right)\left(\frac{1}{2} + \Delta\right) 
 =  \frac{1}{2} \left(\frac{1}{4} - \Delta^2\right)
\\ \iff \qquad  \frac{\pi_{\text{rare}}}{1/4 - \Delta^2}  \leq ~ & \frac{1}{2} 
\\ \iff \qquad  1 - \frac{\pi_{\text{rare}}}{1/4 - \Delta^2}  \geq ~ & \frac{1}{2} 
.
\end{align*}
Now we can bound \(I_{\alpha \beta}^\top I_{\alpha \alpha}^{-1} I_{\alpha \beta} \) using \eqref{q.def}:
\begin{align*}
I_{\alpha \beta}^\top I_{\alpha \alpha}^{-1} I_{\alpha \beta} \preceq ~ &   \frac{2}{I_{\alpha_1 \alpha_1} I_{\alpha_2 \alpha_2}}   \left( I_{\alpha_2 \alpha_2}  I_{\beta \alpha_1} I_{\beta\alpha_1 }^\top +  I_{\alpha_1 \alpha_1} I_{\beta \alpha_2} I_{\beta \alpha_2}^\top   + | I_{\alpha_1 \alpha_2} | I_{\beta \alpha_1}  I_{ \beta \alpha_2}^\top +  | I_{\alpha_1 \alpha_2} | I_{\beta \alpha_2} I_{\beta \alpha_1}^\top  \right) 
\\ = ~ &  2   \left(
 \frac{  I_{\beta \alpha_1} I_{\beta\alpha_1 }^\top }{I_{\alpha_1 \alpha_1} }
+ \frac{I_{\beta \alpha_2} I_{\beta \alpha_2}^\top }{ I_{\alpha_2 \alpha_2}} 
+ \frac{ | I_{\alpha_1 \alpha_2} |  }{I_{\alpha_1 \alpha_1} I_{\alpha_2 \alpha_2}}  \left( I_{\beta \alpha_1}  I_{ \beta \alpha_2}^\top + I_{\beta \alpha_2} I_{\beta \alpha_1}^\top \right)  \right)
,
\end{align*}
and
\begin{align*}
I_{\beta \beta} - I_{\alpha \beta}^\top I_{\alpha \alpha}^{-1} I_{\alpha \beta} \succeq ~ &  I_{\beta \beta}  -  2 \frac{  I_{\beta \alpha_1} I_{\beta\alpha_1 }^\top }{I_{\alpha_1 \alpha_1} } 
 -    2  \left( \frac{I_{\beta \alpha_2} I_{\beta \alpha_2}^\top }{ I_{\alpha_2 \alpha_2}} 
+ \frac{ | I_{\alpha_1 \alpha_2} |  }{I_{\alpha_1 \alpha_1} I_{\alpha_2 \alpha_2}}  \left( I_{\beta \alpha_1}  I_{ \beta \alpha_2}^\top + I_{\beta \alpha_2} I_{\beta \alpha_1}^\top \right)  \right)
\end{align*}
Note that \( I_{\beta \beta}\), \(I_{\beta \alpha_1} I_{\beta \alpha_1}^\top \), \(I_{\beta \alpha_2} I_{\beta \alpha_2}^\top \), and \( I_{\beta \alpha_1}  I_{ \beta \alpha_2}^\top + I_{\beta \alpha_2} I_{\beta \alpha_1}^\top \) are all symmetric. By Weyl's theorem, it holds that for symmetric matrices with matching dimensions \(\boldsymbol{A}\) and \(\boldsymbol{B}\),
\[
\lambda_{\text{min}} (\boldsymbol{A} - \boldsymbol{B}) \geq \lambda_{\text{min}} (\boldsymbol{A}) - \lVert \boldsymbol{B} \rVert_{\text{op}}
,
\]
because for any \(\boldsymbol{v}\)
\begin{align*}
(\boldsymbol{A} - \boldsymbol{B}) \boldsymbol{v} = \boldsymbol{A}\boldsymbol{v} - \boldsymbol{B}\boldsymbol{v} \geq  \lambda_{\text{min}} (\boldsymbol{A}) \boldsymbol{v} -  \lVert \boldsymbol{B} \rVert_{\text{op}} \boldsymbol{v} =   \left( \lambda_{\text{min}} (\boldsymbol{A})  -  \lVert \boldsymbol{B} \rVert_{\text{op}} \right) \boldsymbol{v}
.
\end{align*}
So
\begin{align*}
& \lambda_{\text{min}}  \left( I_{\beta \beta} - I_{\alpha \beta}^\top I_{\alpha \alpha}^{-1} I_{\alpha \beta} \right)
\\ \geq ~ & \lambda_{\text{min}} \left(I_{\beta \beta} -  2 \frac{I_{\beta \alpha_1} I_{\beta \alpha_1}^\top}{I_{\alpha_1 \alpha_1}}  \right)   -   \frac{2}{ I_{\alpha_2 \alpha_2} }  \left\lVert   I_{\beta \alpha_2} I_{\beta \alpha_2}^\top   \right\rVert_{\text{op}}
  -    2 \frac{ | I_{\alpha_1 \alpha_2} |  }{I_{\alpha_1 \alpha_1} I_{\alpha_2 \alpha_2}}   \left\lVert I_{\beta \alpha_1}  I_{ \beta \alpha_2}^\top + I_{\beta \alpha_2} I_{\beta \alpha_1}^\top   \right\rVert_{\text{op}}
\\ \stackrel{(a)}{\geq} ~ &   \lambda_{\text{min}} \left(I_{\beta \beta} - 2 \frac{I_{\beta \alpha_1} I_{\beta \alpha_1}^\top}{I_{\alpha_1 \alpha_1}}  \right)  
-   \frac{2}{ I_{\alpha_2 \alpha_2} }  \left\lVert   I_{\beta \alpha_2} I_{\beta \alpha_2}^\top   \right\rVert_{\text{op}}
  -    2 \frac{ | I_{\alpha_1 \alpha_2} |  }{I_{\alpha_1 \alpha_1} I_{\alpha_2 \alpha_2}}  \left( \left\lVert I_{\beta \alpha_1}  I_{ \beta \alpha_2}^\top   \right\rVert_{\text{op}}  +  \left\lVert  I_{\beta \alpha_2} I_{\beta \alpha_1}^\top   \right\rVert_{\text{op}} \right)
\\ \stackrel{(b)}{\geq} ~ &   \lambda_{\text{min}} \left(I_{\beta \beta} -  2 \frac{I_{\beta \alpha_1} I_{\beta \alpha_1}^\top}{I_{\alpha_1 \alpha_1}}  \right)   
-   2\frac{\lVert   I_{\beta \alpha_2} \rVert_2^2}{ I_{\alpha_2 \alpha_2}}   
-    2 \frac{ | I_{\alpha_1 \alpha_2} |  }{I_{\alpha_1 \alpha_1} I_{\alpha_2 \alpha_2}}  \left(2 \left\lVert I_{\beta \alpha_1} \right \rVert_2 \left \lVert I_{ \beta \alpha_2}   \right\rVert_2  \right)
\\ \stackrel{(c)}{\geq} ~ &   \lambda_{\text{min}} \left(I_{\beta \beta} - 2 \frac{I_{\beta \alpha_1} I_{\beta \alpha_1}^\top}{I_{\alpha_1 \alpha_1}}  \right)   
 -   2 M^2  \pi_{\text{rare}} \left(1 + \frac{\pi_{\text{rare}}}{1/2 - \Delta} \right)
-    2 \frac{  I_{\alpha_2 \alpha_2}  }{I_{\alpha_1 \alpha_1} I_{\alpha_2 \alpha_2}}  \left(2 \left\lVert I_{\beta \alpha_1} \right \rVert_2 \cdot M \cdot I_{\alpha_2 \alpha_2}  \right)
\\ = ~ &   \lambda_{\text{min}} \left(I_{\beta \beta} - 2  \frac{I_{\beta \alpha_1} I_{\beta \alpha_1}^\top}{I_{\alpha_1 \alpha_1}}  \right)    
 -  2M \left(   M  \pi_{\text{rare}} \left(1 + \frac{\pi_{\text{rare}}}{1/2 - \Delta} \right) 
 +   2\frac{  I_{\alpha_2 \alpha_2}  }{I_{\alpha_1 \alpha_1}}   \left\lVert I_{\beta \alpha_1} \right \rVert_2  \right)
\\ \stackrel{(d)}{\geq} ~ &   \lambda_{\text{min}} \left(I_{\beta \beta} - 2 \frac{I_{\beta \alpha_1} I_{\beta \alpha_1}^\top}{I_{\alpha_1 \alpha_1}}  \right)    
 -  2M \left(  M   \pi_{\text{rare}} \left(1 + \frac{\pi_{\text{rare}}}{1/2 - \Delta} \right)
 +   2\frac{1  }{I_{\alpha_1 \alpha_1}}   \left\lVert I_{\beta \alpha_1} \right \rVert_2  \cdot  \pi_{\text{rare}} \left(1 + \frac{\pi_{\text{rare}}}{1/2 - \Delta} \right) \right)
\\ = ~ &   \lambda_{\text{min}} \left(I_{\beta \beta} - 2 \frac{I_{\beta \alpha_1} I_{\beta \alpha_1}^\top}{I_{\alpha_1 \alpha_1}}  \right)    -  2M  \pi_{\text{rare}} \left(1 + \frac{\pi_{\text{rare}}}{1/2 - \Delta} \right) \left(  M   
 +   2\frac{ \left\lVert I_{\beta \alpha_1} \right \rVert_2    }{I_{\alpha_1 \alpha_1}}  \right)
\\ \stackrel{(e)}{\geq} ~ &  \lambda_{\text{min}} \left(I_{\beta \beta} - 2 \frac{I_{\beta \alpha_1} I_{\beta \alpha_1}^\top}{I_{\alpha_1 \alpha_1}}  \right)   
 - 2 M   \pi_{\text{rare}} \left(1 + \frac{1}{2} \right) \left(  M   
 +   \frac{M^2}{2} \right)
\\ \stackrel{(f)}{\geq} ~ &  \frac{1}{2}\lambda_{\text{min}} \left(I_{\beta \beta} -  2\frac{I_{\beta \alpha_1} I_{\beta \alpha_1}^\top}{I_{\alpha_1 \alpha_1}}  \right)
,
\end{align*}
where in \((a)\) we used the triangle inequality, in \((b)\) we used that for any \(\boldsymbol{a}, \boldsymbol{b} \in \mathbb{R}^n\) it holds that \(\lVert \boldsymbol{a} \boldsymbol{b}^\top \rVert_{\text{op}} = | \boldsymbol{b}^\top \boldsymbol{a} | \leq \lVert \boldsymbol{a} \rVert_2 \lVert \boldsymbol{b} \rVert_2\) (note that \(\boldsymbol{a} \boldsymbol{b}^\top \) is rank one with eigenvector \(\boldsymbol{a}\)) as well as the triangle inequality, in \((c)\) we used \eqref{e}, \eqref{f}, and \eqref{g}, in \((d)\) we used \eqref{b}, \((e)\) follows from \eqref{h} and
\[
\frac{\pi_{\text{rare}}}{1/2 - \Delta} \leq   \frac{1/2(1/2 - \Delta)(1/2 + \Delta)}{1/2 - \Delta}  =  \frac{1}{2} \left( \frac{1}{2} + \Delta \right)   \leq   \frac{1}{2} \left( \frac{1}{2} + \frac{1}{2} \right)   \leq \frac{1}{2}
\]
(since \(\Delta \leq 1/2\)), and in \((f)\) we used our assumptions that
\[
\lambda_{\text{min}} \left(I_{\beta \beta} -   2 \frac{I_{\beta \alpha_1} I_{\beta \alpha_1}^\top}{I_{\alpha_1 \alpha_1}}  \right)  > 0
\]
and
\begin{align*}
 \pi_{\text{rare}} \leq ~ &  \frac{  \lambda_{\text{min}} \left(I_{\beta \beta} -   2 \frac{I_{\beta \alpha_1} I_{\beta \alpha_1}^\top}{I_{\alpha_1 \alpha_1}}  \right)  }{3 M^2   \left(   2   
 +  M \right) }
\\ \iff \qquad  3 M  \pi_{\text{rare}}  \left(   M   
 +  \frac{M^2}{2} \right) \leq ~ &   \frac{1}{2}   \lambda_{\text{min}} \left(I_{\beta \beta} - 2  \frac{I_{\beta \alpha_1} I_{\beta \alpha_1}^\top}{I_{\alpha_1 \alpha_1}}  \right)  
 .
 \end{align*}
Substituting this into \eqref{final.cov.res} shows that
\[
 \left \lVert \Cov \left( \lim_{n \to \infty}  \sqrt{n} \left[\boldsymbol{\beta} -  \boldsymbol{\hat{\beta}}^{\text{prop. odds}} \right] \right) \right\rVert_{\text{op}} \leq  \frac{2}{\lambda_{\text{min}} \left(I_{\beta \beta} -  2\frac{I_{\beta \alpha_1} I_{\beta \alpha_1}^\top}{I_{\alpha_1 \alpha_1}}  \right)}
 .
 \]

\subsection{Proof of Lemma \ref{lem.ineq.final.thm}}\label{sec.lem.ineq.final.thm}

We omit ${\bf x}$ in integrals, when it is clear (e.g.,
$\pi_1$ stands for $\pi_1({\bf x})$).

  \noindent {\bf Proof of \eqref{a}:}
  
  From \eqref{alpha.beta.block} and \eqref{m.def}, using \( \pi_1(\boldsymbol{x}) +  \pi_2(\boldsymbol{x}) +  \pi_3(\boldsymbol{x}) = 1\) we have

  \begin{align}
    I_{\alpha_2 \alpha_2}&= M_2   \nonumber
    \\ & = \int \left[ \pi_3(\boldsymbol{x}) (1 - \pi_3 (\boldsymbol{x}) ) \right]^2   \left(     \frac{1 }{\pi_{2}(\boldsymbol{x}) }   +  \frac{1}{\pi_{3}(\boldsymbol{x}) }  \right) \ d F(\boldsymbol{x})   \nonumber
    \\ & = \int \left[ \pi_3(\boldsymbol{x}) (1 - \pi_3 (\boldsymbol{x}) ) \right]^2   \left(     \frac{\pi_2(\boldsymbol{x}) + \pi_3(\boldsymbol{x}) }{\pi_{2}(\boldsymbol{x}) \pi_{3}(\boldsymbol{x}) }  \right) \ d F(\boldsymbol{x})   \nonumber
    \\ & = \int \pi_3(\boldsymbol{x}) [1 - \pi_3 (\boldsymbol{x}) ]^2   \left(     \frac{1 - \pi_1(\boldsymbol{x}) }{\pi_{2}(\boldsymbol{x})  }  \right) \ d F(\boldsymbol{x})  \label{m2.id}
    \\
                         &\le \pi_{\text{rare}}\int(1-\pi_3)^2(1-\pi_1)/\pi_2dF  \nonumber\\
                         &= \pi_{\text{rare}}\int\pi_1  (1-\pi_1)^2 \frac{(1-\pi_3)^2}{\pi_2}\left[\frac{1}{\pi_1(1-\pi_1)}\right]dF  \nonumber
                         \\ &\stackrel{(a)}{\leq} \pi_{\text{rare}}\int\pi_1 (1-\pi_1)^2 \frac{1-\pi_3}{\pi_2} \left[\frac{1}{(1/2-\Delta)(1/2+\Delta)}\right]dF  \nonumber
                         \\ & = \pi_{\text{rare}} \int \pi_1^2 (1 - \pi_1)^2 \frac{\pi_1 + \pi_2}{\pi_1 \pi_2}  \frac{1}{1/4 - \Delta^2}  \nonumber\\
                          & = \pi_{\text{rare}} \int \pi_1^2 (1 - \pi_1)^2\ \left( \frac{1}{\pi_1} + \frac{1}{\pi_2} \right) \frac{1}{1/4 - \Delta^2}  \nonumber \\
& = \pi_{\text{rare}} M_1 \frac{1}{1/4 - \Delta^2} \nonumber \\
& = \pi_{\text{rare}}I_{\alpha_1\alpha_1}\frac{1}{1/4 - \Delta^2} \nonumber
,
  \end{align}
where in \((a)\) we used the fact that \(t \mapsto 1/[t(1-t)]\) is nonincreasing in \(t\) for \(t \in (0, 1/2]\), so it is largest when \(t\) is small as possible, and by assumption \(\inf_{\boldsymbol{x} \in \mathcal{S}} \{ \pi_1 \wedge 1 -  \pi_1\} \geq 1/2 - \Delta\).

  \noindent {\bf Proof of \eqref{b}:} Using \eqref{m2.id} we have

  \begin{align*}
    I_{\alpha_2 \alpha_2}&= \int \pi_3(\boldsymbol{x}) [1 - \pi_3 (\boldsymbol{x}) ]^2   \left(     \frac{ \pi_2(\boldsymbol{x}) + \pi_3(\boldsymbol{x}) }{\pi_{2}(\boldsymbol{x})  }  \right) \ d F(\boldsymbol{x}) 
\\ &=\int\pi_3(1-\pi_3)^2(1+\pi_3/\pi_2)dF
\\ & \leq \pi_{\text{rare}} \int(1-0)^2  \left(1+\frac{\pi_{\text{rare}}}{1/2-\Delta}\right) dF
.
  \end{align*}

%

   \noindent {\bf Proof of \eqref{e}:} Using \eqref{m2.id} along with \eqref{alpha.block} and \eqref{m.tilde.def} we have
   
   \begin{align*}
I_{\alpha_2 \alpha_2}= ~ & \int \pi_3(\boldsymbol{x}) [1 - \pi_3 (\boldsymbol{x}) ]^2   \left(     \frac{1 - \pi_1(\boldsymbol{x}) }{\pi_{2}(\boldsymbol{x})  }  \right) \ d F(\boldsymbol{x}) 
\\ = ~ & \int \pi_3(\boldsymbol{x}) [1 - \pi_3 (\boldsymbol{x}) ][  1 - \pi_1(\boldsymbol{x}) ]   \left(     \frac{ \pi_1(\boldsymbol{x}) + \pi_2(\boldsymbol{x}) }{\pi_{2}(\boldsymbol{x})  }  \right) \ d F(\boldsymbol{x}) 
\\ = ~ &\int\pi_3(1-\pi_3)(1-\pi_1) \frac{\pi_1}{\pi_2}dF + \int\pi_3(1-\pi_3)(1-\pi_1)dF
\\ = ~ & \tilde{M}_2 + \int\pi_3(1-\pi_3)(1-\pi_1)dF
\\ = ~ &|I_{\alpha_1\alpha_2}| + \int\pi_3(1-\pi_3)(1-\pi_1)dF
\\ \geq ~  & |I_{\alpha_1\alpha_2}|.
\end{align*}

 \noindent {\bf Proof of \eqref{f}:} From \eqref{alpha.beta.block} and \eqref{j.2.j.tilde.3.sum} we have

\begin{align*}
\left \lVert I_{\beta \alpha_2}
\right\rVert_2 & = \left \lVert   \int  \boldsymbol{x} \left[ \pi_2(\boldsymbol{x}) + \pi_{3}(\boldsymbol{x}) \right] \pi_3 (\boldsymbol{x}) [ 1 - \pi_3 (\boldsymbol{x}) ] \ d F(\boldsymbol{x})  \right\rVert_2
\\ & = \left \lVert   \int  \boldsymbol{x} \left[1 - \pi_1(\boldsymbol{x}) \right] \pi_3 (\boldsymbol{x}) [ 1 - \pi_3 (\boldsymbol{x}) ] \ d F(\boldsymbol{x})  \right\rVert_2
\\ &\le\int\| \boldsymbol{x}\|_2 \pi_3 (\boldsymbol{x})[1-\pi_3(\boldsymbol{x})][1-\pi_1(\boldsymbol{x})]dF
\\ &\le M\int\pi_3(\boldsymbol{x})[1-\pi_3(\boldsymbol{x})][1-\pi_1(\boldsymbol{x})]dF
\\ & \leq M \cdot  \int \pi_3(\boldsymbol{x}) [1 - \pi_3 (\boldsymbol{x}) ] [ 1 - \pi_1(\boldsymbol{x})]   \left(   1 +   \frac{ \pi_3(\boldsymbol{x}) }{\pi_{2}(\boldsymbol{x})  }  \right) \ d F(\boldsymbol{x}) 
\\ & = M \cdot  \int \pi_3(\boldsymbol{x}) [1 - \pi_3 (\boldsymbol{x}) ] [ 1 - \pi_1(\boldsymbol{x})]   \left(     \frac{\pi_2(\boldsymbol{x}) +   \pi_3 (\boldsymbol{x}) }{\pi_{2}(\boldsymbol{x})  }  \right) \ d F(\boldsymbol{x}) 
\\ & = M \cdot  \int \pi_3(\boldsymbol{x}) [1 - \pi_3 (\boldsymbol{x}) ] [ 1 - \pi_1(\boldsymbol{x})]   \left(     \frac{ 1 - \pi_1(\boldsymbol{x}) }{\pi_{2}(\boldsymbol{x})  }  \right) \ d F(\boldsymbol{x}) 
\\ & \leq M \cdot  \int \pi_3(\boldsymbol{x}) [1 - \pi_3 (\boldsymbol{x}) ]^2   \left(     \frac{1 - \pi_1(\boldsymbol{x}) }{\pi_{2}(\boldsymbol{x})  }  \right) \ d F(\boldsymbol{x}) 
\\  & =  M \cdot
 I_{\alpha_2 \alpha_2}
 ,
 \end{align*}
where in the last inequality we used \(\pi_3(\boldsymbol{x}) \leq \pi_{\text{rare}} < 1/2 - \Delta \leq \pi_1(\boldsymbol{x})\) for all \(\boldsymbol{x} \in \mathcal{S}\) and in the last step we used \eqref{m2.id}.
 
  \noindent {\bf Proof of \eqref{g}:} This follows from \eqref{b} and \eqref{f}.
  
  \noindent {\bf Proof of \eqref{h}:}  Using from \eqref{alpha.beta.block}
\begin{align*}
 \lVert I_{\beta \alpha_1} \rVert_2 = ~ & \left\lVert \int \boldsymbol{x} \pi_1^2(\boldsymbol{x})[1 - \pi_1(\boldsymbol{x})] \ dF(\boldsymbol{x}) +  \int \boldsymbol{x} \pi_2(\boldsymbol{x}) \pi_1(\boldsymbol{x})[1 - \pi_1(\boldsymbol{x})] \ dF(\boldsymbol{x}) \right \rVert_2
\\ = ~ &  \left\lVert \int \boldsymbol{x} [\pi_1(\boldsymbol{x}) + \pi_2(\boldsymbol{x})] \pi_1(\boldsymbol{x})[1 - \pi_1(\boldsymbol{x})] \ dF(\boldsymbol{x})   \right \rVert_2
\\ \leq ~ & \int   \left\lVert  \boldsymbol{x}   \right \rVert_2 [\pi_1(\boldsymbol{x}) + \pi_2(\boldsymbol{x})] \pi_1(\boldsymbol{x})[1 - \pi_1(\boldsymbol{x})] \ dF(\boldsymbol{x})  
\\ \leq ~ & M \int  [\pi_1(\boldsymbol{x}) + \pi_2(\boldsymbol{x})] \pi_1(\boldsymbol{x})[1 - \pi_1(\boldsymbol{x})] \ dF(\boldsymbol{x})  
\end{align*}
and from \eqref{alpha.block}
\begin{align}
I_{\alpha_1 \alpha_1} = M_1 = ~ & \int (\pi_1(\boldsymbol{x})[1 - \pi_1(\boldsymbol{x})])^2 \left( \frac{1}{\pi_1(\boldsymbol{x})} + \frac{1}{\pi_2(\boldsymbol{x})} \right) \ d F(\boldsymbol{x}) \nonumber
\\ = ~ & \int (\pi_1(\boldsymbol{x})[1 - \pi_1(\boldsymbol{x})])^2  \frac{\pi_1(\boldsymbol{x}) + \pi_2(\boldsymbol{x})}{\pi_1(\boldsymbol{x}) \pi_2(\boldsymbol{x})} \ d F(\boldsymbol{x})  \nonumber
\\ = ~ & \int \pi_1(\boldsymbol{x})[1 - \pi_1(\boldsymbol{x})] [ \pi_1(\boldsymbol{x}) + \pi_2(\boldsymbol{x})]  \frac{1 - \pi_1(\boldsymbol{x})}{\pi_2(\boldsymbol{x})} \ d F(\boldsymbol{x}) \nonumber
\\ \geq ~ & \int \pi_1(\boldsymbol{x})[1 - \pi_1(\boldsymbol{x})] [ \pi_1(\boldsymbol{x}) + \pi_2(\boldsymbol{x})]   \ d F(\boldsymbol{x}) \nonumber 
,
\end{align}
we have
\[
\frac{ \lVert I_{\beta \alpha_1} \rVert_2^2}{ I_{\alpha_1 \alpha_1}} \leq  M^2 \int  [\pi_1(\boldsymbol{x}) + \pi_2(\boldsymbol{x})] \pi_1(\boldsymbol{x})[1 - \pi_1(\boldsymbol{x})] \ dF(\boldsymbol{x})  \leq     M^2 \int \frac{1}{4} [\pi_1(\boldsymbol{x}) + \pi_2(\boldsymbol{x})]  \ dF(\boldsymbol{x})  \leq \frac{M^2}{4},
\]
where in the second-to-last step we used that \(t \mapsto t(1-t) \leq 1/4\) for all \(t \in [0,1]\). 

\section{Proof of Lemma \ref{asym.matrix}}\label{lemmas.sec}

First we will calculate the Fisher information matrices, then we will show the convergence results. 

Since the logistic regression model can be considered a special case of the proportional odds model with \(K = 2\) categories, we will mostly focus our calculations on the proportional odds model.

\subsection{Calculating the Log Likelihood and Gradients}

In the proportional odds model, the likelihood can be expressed as 
\begin{align}
& \prod_{i=1}^n \prod_{k=1}^K \left( \frac{1}{1 + \exp\left\{ -(\alpha_k + \boldsymbol{\beta}^\top \boldsymbol{x}_{i} ) \right\}} - \frac{1}{1 + \exp\left\{ -(\alpha_{k-1} + \boldsymbol{\beta}^\top \boldsymbol{x}_{i} )\right\}} \right)^{ \mathbbm{1}\{y_i = k\} } \label{prop.odds.likelihood}
,
\end{align}
so the log likelihood can be expressed as
\begin{align*}
\mathcal{L}^{\text{prop. odds}}(\boldsymbol{\alpha}, \boldsymbol{\beta})  & = \sum_{i=1}^n \sum_{k=1}^K \mathbbm{1}\{y_i = k\} \log \left( \frac{1}{1 + \exp\left\{ -(\alpha_k + \boldsymbol{\beta}^\top \boldsymbol{x}_{i} ) \right\}} - \frac{1}{1 + \exp\left\{ -(\alpha_{k-1} + \boldsymbol{\beta}^\top \boldsymbol{x}_{i} )\right\}} \right)
\\ & = \sum_{i=1}^n \sum_{k=1}^K \mathbbm{1}\{y_i = k\} \log \left(p_{ik} - p_{i,k-1} \right)
\\ & = \sum_{i=1}^n \sum_{k=1}^K \mathbbm{1}\{y_i = k\} \log \left(\pi_{ik}  \right),
\end{align*}
where
\begin{align*}
p_{ik} & :=  p_k \left( \boldsymbol{x}_{i} \right) = \mathbb{P}(y_i \leq k \mid \boldsymbol{x}_i),
\\ \pi_{ik}&  := \mathbb{P}(y_i = k \mid \boldsymbol{x}_i ) = p_{ik} - p_{i,k-1} = p_k \left( \boldsymbol{x}_{i} \right)  -  p_{k-1} \left( \boldsymbol{x}_{i} \right) ,
\end{align*}
\(\alpha_0 := - \infty\), and \( \alpha_{K} := \infty\) (while \(\boldsymbol{\alpha} := (\alpha_1, \ldots, \alpha_{K-1})^\top\) are parameters to be estimated). Using
\[
\pderiv{}{t} \frac{1}{1 + \exp\{-t\}}  = \frac{1}{1 + \exp\{-t\}} \left(1 - \frac{1}{1 + \exp\{-t\}}\right)
,
\]
the gradient has entries corresponding to \(\boldsymbol{\beta}\) equal to
\begin{align*}
\nabla_{\boldsymbol{\beta}} \mathcal{L}^{\text{prop. odds}}(\boldsymbol{\alpha}, \boldsymbol{\beta}) & = \sum_{i=1}^n \sum_{k=1}^K \mathbbm{1}\{y_i = k\} \boldsymbol{x}_{i} \left( \frac{p_{ik} (1 - p_{ik}) - p_{i,k-1} (1 - p_{i,k-1})}{p_{ik} - p_{i,k-1}}\right)
\\ & = \sum_{i=1}^n \sum_{k=1}^K \mathbbm{1}\{y_i = k\} \boldsymbol{x}_{i} \left( \frac{p_{ik}  - p_{i,k-1} - \left( p_{ik}^2 -  p_{i,k-1}^2  \right) }{p_{ik} - p_{i,k-1}}\right)
\\ & = \sum_{i=1}^n \sum_{k=1}^K \mathbbm{1}\{y_i = k\} \boldsymbol{x}_{i}  \left( 1 - p_{ik} - p_{i,k-1} \right),
\end{align*}
and using
\begin{align*}
\pderiv{}{\alpha_k} \pi_{ik} & = \pderiv{}{\alpha_k} \left(p_{ik} - p_{i,k-1} \right)= p_{ik}(1 - p_{ik}) \qquad \text{and}
\\  \pderiv{}{\alpha_k} \pi_{i,k+1} & = \pderiv{}{\alpha_k} \left(p_{i,k+1} - p_{ik} \right)= -p_{ik}(1 - p_{ik}), 
\end{align*}
the entries corresponding to \(\boldsymbol{\alpha}\) equal
\begin{align*}
\pderiv{}{\alpha_k}  \mathcal{L}^{\text{prop. odds}}(\boldsymbol{\alpha}, \boldsymbol{\beta}) &  = \sum_{i=1}^n \pderiv{}{\alpha_k} \left(   \mathbbm{1}\{y_i = k\}   \log \left( \pi_{ik} \right) + \mathbbm{1}\{y_i = k + 1\}   \log \left(\pi_{i,k+1} \right) \right)
\\ &  = \sum_{i=1}^n \left( \mathbbm{1}\{y_i = k\}  \frac{p_{ik}(1 - p_{ik})}{\pi_{ik}} -  \mathbbm{1}\{y_i = k + 1\}  \frac{p_{ik}(1 - p_{ik})}{\pi_{i,k+1}}  \right)
\\ \implies \qquad \nabla_{\boldsymbol{\alpha}} \mathcal{L}^{\text{prop. odds}}(\boldsymbol{\alpha}, \boldsymbol{\beta}) &  = \sum_{i=1}^n  \boldsymbol{e}_{k}p_{ik}(1 - p_{ik}) \left(   \frac{\mathbbm{1}\{y_i = k\}}{\pi_{ik}}  -   \frac{ \mathbbm{1}\{y_i = k + 1\}}{\pi_{i,k+1}}  \right) ,
\end{align*}
where \(\boldsymbol{e}_k \in \{0, 1\}^{K-1}\) has the \(k^\text{th}\) entry equal to 1 and the rest equal to 0. (Note that since \(\alpha_0 = -\infty\), \(p_{i0} = 0\), and similarly \(p_{iK} = 1\) as expected.)

\subsection{Calculating the Hessian Matrices}\label{hessian.calc}

The entries of the Hessian corresponding to \(\boldsymbol{\beta}\), \(H_{\beta \beta}^{\text{prop. odds}} \) are
\begin{align*}
\nabla^2_{\boldsymbol{\beta}} \mathcal{L}^{\text{prop. odds}}(\boldsymbol{\alpha}, \boldsymbol{\beta}) = ~ & -\sum_{i=1}^n \sum_{k=1}^K \mathbbm{1}\{y_i = k\} \boldsymbol{x}_{i} \boldsymbol{x}_{i}^\top \left[ p_{ik} (1 - p_{ik}) + p_{i,k-1} \left( 1- p_{i,k-1} \right) \right]
\\ = ~ & -\sum_{i=1}^n  \boldsymbol{x}_{i} \boldsymbol{x}_{i}^\top \sum_{k=1}^{K}   \left( \mathbbm{1}\{y_i = k + 1\} +  \mathbbm{1}\{y_i = k\}\right)    \Var \left( \mathbbm{1}\{y_i \leq k\} \right) .
\end{align*}
Using
\[
\pderiv{}{\alpha_k} p_{ik}(1 -p_{ik})  =  p_{ik}(1 -p_{ik}) -2  p_{ik}^2(1 -p_{ik})  =  p_{ik}(1 -p_{ik}) \left(1 -2  p_{ik}\right),
\]
the entries corresponding to the \(\boldsymbol{\alpha}\) block of the Hessian \(H_{\alpha \alpha}^{\text{prop. odds}} \) are as follows:
\begin{align*}
\pderiv{^2}{\alpha_k^2} \mathcal{L}^{\text{prop. odds}}(\boldsymbol{\alpha}, \boldsymbol{\beta}) = ~ &  \sum_{i=1}^n  \pderiv{}{\alpha_k}  \left( p_{ik}(1 - p_{ik}) \left(   \frac{\mathbbm{1}\{y_i = k\}}{\pi_{ik}}  -   \frac{ \mathbbm{1}\{y_i = k + 1\}}{\pi_{i,k+1}}  \right)  \right)
\\ = ~ &  \sum_{i=1}^n   \bigg(  \mathbbm{1}\{y_i = k\} \frac{ \pi_{ik} p_{ik}(1 -p_{ik}) \left(1 -2  p_{ik}\right) - p_{ik}^2(1 - p_{ik})^2}{\pi_{ik}^2} 
\\ &  -    \mathbbm{1}\{y_i = k + 1\}\frac{\pi_{i,k+1} p_{ik}(1 -p_{ik}) \left(1 -2  p_{ik}\right) + p_{ik}^2(1 - p_{ik})^2}{\pi_{i,k+1}^2}  \bigg)
\\ = ~ &  \sum_{i=1}^n   p_{ik}(1 -p_{ik}) \bigg(  \mathbbm{1}\{y_i = k\} \frac{ \pi_{ik} \left(1 -2  p_{ik}\right) - p_{ik}(1 - p_{ik})}{\pi_{ik}^2} 
\\ &  -    \mathbbm{1}\{y_i = k + 1\}\frac{\pi_{i,k+1} \left(1 -2  p_{ik}\right) + p_{ik}(1 - p_{ik})}{\pi_{i,k+1}^2}  \bigg),
\\ \pderiv{^2}{\alpha_k \partial \alpha_{k-1} } \mathcal{L}^{\text{prop. odds}}(\boldsymbol{\alpha}, \boldsymbol{\beta}) = ~ &  \sum_{i=1}^n p_{ik}(1 - p_{ik})  \pderiv{}{\alpha_{k-1}}  \left(    \frac{\mathbbm{1}\{y_i = k\}}{\pi_{ik}}  -   \frac{ \mathbbm{1}\{y_i = k + 1\}}{\pi_{i,k+1}}  \right) 
\\ = ~ &  - \sum_{i=1}^n  \mathbbm{1}\{y_i = k\}p_{ik}(1 - p_{ik})   \left(    \frac{-p_{i,k-1}(1 - p_{i,k-1})}{\pi_{ik}^2}  \right) 
\\ = ~ &  \sum_{i=1}^n  \mathbbm{1}\{y_i = k\}     \cdot \frac{p_{ik}(1 - p_{ik}) p_{i,k-1}(1 - p_{i,k-1})}{\pi_{ik}^2}  , \qquad \text{and}
\\ \pderiv{^2}{\alpha_k \partial \alpha_{k'} } \mathcal{L}^{\text{prop. odds}}(\boldsymbol{\alpha}, \boldsymbol{\beta}) = ~ &  0, \qquad \text{all other } k \neq k',
\end{align*}
where \(\pi_{i,K+1} = 0\). (Note that \( \pderiv{^2}{\alpha_k \partial \alpha_{k+1} } \mathcal{L}^{\text{prop. odds}}(\boldsymbol{\alpha}, \boldsymbol{\beta}) \) is nonzero as well, but it matches the expression for \(  \pderiv{^2}{\alpha_{\tilde{k}} \partial \alpha_{\tilde{k}-1} } \mathcal{L}^{\text{prop. odds}}(\boldsymbol{\alpha}, \boldsymbol{\beta})  \) with \(\tilde{k} := k + 1\).) 

Finally, the entries corresponding to the \(\boldsymbol{\alpha}\) and \(\boldsymbol{\beta}\) mixed blocks \(H_{\alpha \beta}^{\text{prop. odds}} \) of the Hessian are
\begin{align*}
\pderiv{}{\alpha_k} \nabla_{\boldsymbol{\beta}} \mathcal{L}^{\text{prop. odds}}(\boldsymbol{\alpha}, \boldsymbol{\beta})  = ~ & \sum_{i=1}^n  \bigg( \mathbbm{1}\{y_i = k\} \boldsymbol{x}_{i}  \pderiv{}{\alpha_k}   \left( 1 - p_{ik} - p_{i,k-1} \right) 
\\ &  + \mathbbm{1}\{y_i = k + 1\} \boldsymbol{x}_{i}  \pderiv{}{\alpha_k}   \left( 1 - p_{i,k+1} - p_{i,k} \right)  \bigg)
\\   = ~ &  -\sum_{i=1}^n \boldsymbol{x}_{i}  p_{ik}\left( 1 - p_{ik} \right) \left(   \mathbbm{1}\{y_i = k\}   +   \mathbbm{1}\{y_i = k + 1\}   \right) , ~  k \in [K-1].
\end{align*}
\subsection{Calculation of the Fisher Information Matrices}

Now we find the Fisher information matrices
\[
I^{\text{prop. odds}} (\boldsymbol{\alpha}, \boldsymbol{\beta})= \begin{pmatrix} I_{\alpha \alpha}^{\text{prop. odds}} & I_{\beta \alpha}^{\text{prop. odds}}
\\ \left( I_{\beta \alpha}^{\text{prop. odds}}\right)^\top & I_{\beta \beta}^{\text{prop. odds}}
\end{pmatrix}
\]
and
\[
I^{\text{logistic}} (\boldsymbol{\alpha}, \boldsymbol{\beta})= \begin{pmatrix} I_{\alpha \alpha}^{\text{logistic}} & \left( I_{\beta \alpha}^{\text{logistic}} \right)^\top
\\ I_{\beta \alpha}^{\text{logistic}} & I_{\beta \beta}^{\text{logistic}}
\end{pmatrix}
\]
 by taking the negative expectation of each block (using a single observation). For the \(\alpha\) block, we have
\begin{align*}
-\E \left[ \pderiv{^2}{\alpha_k^2} \mathcal{L}^{\text{prop. odds}}(\boldsymbol{\alpha}, \boldsymbol{\beta})  \right] = ~ &  -\E \bigg[      p_{ik}(1 -p_{ik}) \bigg(  \mathbbm{1}\{y_i = k\} \frac{ \pi_{ik} \left(1 -2  p_{ik}\right) - p_{ik}(1 - p_{ik})}{\pi_{ik}^2} 
\\ &  -    \mathbbm{1}\{y_i = k + 1\}\frac{\pi_{i,k+1} \left(1 -2  p_{ik}\right) + p_{ik}(1 - p_{ik})}{\pi_{i,k+1}^2}  \bigg)
  \bigg]
\\ = ~ &  -\E \bigg[      p_{ik}(1 -p_{ik}) \bigg(  \frac{ \pi_{ik} \left(1 -2  p_{ik}\right) - p_{ik}(1 - p_{ik})}{\pi_{ik}} 
\\ &  -    \frac{\pi_{i,k+1} \left(1 -2  p_{ik}\right) + p_{ik}(1 - p_{ik})}{\pi_{i,k+1}}  \bigg)
\bigg]
\\ = ~ &  \E \bigg[    p_{ik}^2(1 - p_{ik})^2 \left(     \frac{1 }{\pi_{ik}}   +  \frac{1}{\pi_{i,k+1}}  \right)
\bigg]
\\ = ~ &  M_k, \qquad k \in [K-1],
\\ -\E \left[  \pderiv{^2}{\alpha_k \partial \alpha_{k-1} } \mathcal{L}^{\text{prop. odds}}(\boldsymbol{\alpha}, \boldsymbol{\beta})  \right] = ~ & -\E \left[    \mathbbm{1}\{y_i = k\}     \cdot \frac{p_{ik}(1 - p_{ik}) p_{i,k-1}(1 - p_{i,k-1})}{\pi_{ik}^2} \right]
\\ = ~ & -\E \left[   \frac{p_{ik}(1 - p_{ik}) p_{i,k-1}(1 - p_{i,k-1})}{\pi_{ik}} \right]
\\ = ~ & - \tilde{M}_k, \qquad k \in \{2, \ldots, K-1\}, \qquad \text{and}
\\  \E \left[ \pderiv{^2}{\alpha_k \partial \alpha_{k'} } \mathcal{L}^{\text{prop. odds}}(\boldsymbol{\alpha}, \boldsymbol{\beta}) \right] = ~ &  0, \qquad \text{all other } k \neq k',
\end{align*}
where we used the definitions of \(M_k\) and \(\tilde{M}_k\) in \eqref{m.def} and \eqref{m.tilde.def}. Therefore
 \(I_{\alpha \alpha}^{\text{prop. odds}}(\boldsymbol{\alpha}, \boldsymbol{\beta}) \in \mathbb{R}^{(K-1) \times (K-1)}\) has tridiagonal form
\[
I_{\alpha \alpha}^{\text{prop. odds}}(\boldsymbol{\alpha}, \boldsymbol{\beta}) = \begin{pmatrix}
M_1 & -\tilde{M}_2 & 0   & \cdots   & 0  & 0
\\ -\tilde{M}_2 & M_2 &- \tilde{M}_3   & \cdots   & 0  & 0
\\ 0 & -\tilde{M}_3 & M_3   & \cdots   & 0  & 0
\\ \vdots & \vdots & \vdots & \ddots  & \vdots & \vdots
\\ 0 & 0 & 0 &  \cdots & M_{K-2}  &- \tilde{M}_{K-1}
\\ 0 & 0 & 0 &  \cdots   & -\tilde{M}_{K-1}  & M_{K-1}
\end{pmatrix},
\]
verifying \eqref{alpha.block}. Observe that in the case of logistic regression (\(K=2\)), we have
\begin{align}
I_{\alpha \alpha}^{\text{logistic}}(\alpha_1, \boldsymbol{\beta}) = M_1 & = \int \left[ p_{1}(\boldsymbol{x}) (1 - p_{1} (\boldsymbol{x}) ) \right]^2   \left(     \frac{1 }{\pi_{1}(\boldsymbol{x}) }   +  \frac{1}{\pi_{2}(\boldsymbol{x}) }  \right) \ d F(\boldsymbol{x}) \nonumber
\\ & = \int \left[ \pi_{1}(\boldsymbol{x}) (1 - \pi_{1} (\boldsymbol{x}) ) \right]^2   \left(     \frac{1 }{\pi_{1}(\boldsymbol{x}) }   +  \frac{1}{1 - \pi_{1}(\boldsymbol{x}) }  \right) \ d F(\boldsymbol{x}) \nonumber
\\ & = \int \left[ \pi_{1}(\boldsymbol{x}) (1 - \pi_{1} (\boldsymbol{x}) ) \right]^2  \cdot  \frac{1}{ \pi_{1}(\boldsymbol{x})\left(1 - \pi_{1}(\boldsymbol{x}) \right)}  \ d F(\boldsymbol{x}) \nonumber
\\ & = \int \pi_{1}(\boldsymbol{x}) (1 - \pi_{1} (\boldsymbol{x}) )   \ d F(\boldsymbol{x}) \nonumber
\\ & =  M_1^{\text{logistic}}  , \nonumber 
\end{align}
which is \eqref{log.reg.alpha.info}. (This is equivalent to a logistic regression predicting whether \(y_i = 1\) regardless of \(K\).)
\subsection{Mixed block}
For the \(\alpha\)-\(\beta\) mixed block, we have for all \( k \in \{1, \ldots, K-1\}\)
\begin{align*}
 -\E \left[ \pderiv{^2}{\alpha_k \partial \boldsymbol{\beta}} \mathcal{L}^{\text{prop. odds}}(\boldsymbol{\alpha}, \boldsymbol{\beta})  \right]  = ~  & \E \left[     \boldsymbol{X}_1   p_k( \boldsymbol{X}_1) (1 - p_k( \boldsymbol{X}_1)) (\pi_{k} \left( \boldsymbol{X}_1\right)  + \pi_{k+1}\left( \boldsymbol{X}_1\right) ) \right].
\end{align*}
Then
\begin{align*}
 -\E \left[ \pderiv{^2}{\alpha_k \partial \boldsymbol{\beta}} \mathcal{L}^{\text{prop. odds}}(\boldsymbol{\alpha}, \boldsymbol{\beta})  \right]  = ~ & J_k^{\boldsymbol{x}} + \tilde{J}_{k+1}^{\boldsymbol{x}}, \qquad k \in \{1, \ldots, K-1\},
\end{align*}
so \(I_{\beta \alpha}^{\text{prop. odds}}(\boldsymbol{\alpha}, \boldsymbol{\beta}) \in \mathbb{R}^{(K-1) \times p}\) has the form
\[
I_{\beta \alpha}^{\text{prop. odds}}(\boldsymbol{\alpha}, \boldsymbol{\beta}) = \begin{pmatrix}
J_1^{\boldsymbol{x}} + \tilde{J}_{2}^{\boldsymbol{x}} 
\\  J_2^{\boldsymbol{x}} + \tilde{J}_{3}^{\boldsymbol{x}}
\\ \vdots
\\  J_{K-1}^{\boldsymbol{x}} + \tilde{J}_{K}^{\boldsymbol{x}}
\end{pmatrix},
\]
as in \eqref{alpha.beta.block}. In the case of logistic regression predicting whether \(y_i = 1\),
\begin{align*}
I_{\beta \alpha}^{\text{logistic}}(\alpha_1, \boldsymbol{\beta}) & = J_1^{\boldsymbol{x}\text{; logistic}}+\tilde{J}_2^{\boldsymbol{x}\text{; logistic}} 
 =  \int \boldsymbol{x} \pi_1 (\boldsymbol{x}) [ 1 - \pi_1 (\boldsymbol{x}) ] \ d F(\boldsymbol{x})    ,
\end{align*}
verifying \eqref{log.reg.alpha.beta.info}.
\subsection{Beta block}
Finally, for the \(\beta\) block, we have
\begin{align*}
I_{\beta \beta}^{\text{prop. odds}}(\boldsymbol{\alpha}, \boldsymbol{\beta}) & = - \E \left[ \nabla^2_{\boldsymbol{\beta}} \mathcal{L}^{\text{prop. odds}}(\boldsymbol{\alpha}, \boldsymbol{\beta})  \right] 
\\ & =  \E \left[ \E \left[  \boldsymbol{X}_1 \boldsymbol{X}_{1}^\top \sum_{k=1}^{K}   \left( \mathbbm{1}\{y_i = k + 1\} +  \mathbbm{1}\{y_i = k\}\right)    p_k(\boldsymbol{X}_1) (1 - p_k(\boldsymbol{X}_1))   \mid \boldsymbol{X} \right] \right]
\\  & =  \E \left[   \boldsymbol{X}_1 \boldsymbol{X}_{1}^\top \sum_{k=1}^{K}   \left(  \pi_{k + 1} \left(  \boldsymbol{X}_{1 \cdot}\right)  + \pi_{k}  \left(  \boldsymbol{X}_{1 \cdot}\right)  \right)   p_k(\boldsymbol{X}_1) (1 - p_k(\boldsymbol{X}_1))    \right]
\\ 
& = \sum_{k=1}^K \left( J_k^{\boldsymbol{x} \boldsymbol{x}^\top} +  \tilde{J}_k^{\boldsymbol{x} \boldsymbol{x}^\top} \right),
\end{align*}
as in \eqref{beta.block} (recall that \(J_K^{\boldsymbol{x} \boldsymbol{x}^\top} = 0\) and \(\tilde{J}_1^{\boldsymbol{x} \boldsymbol{x}^\top} = 0\) for all \(\boldsymbol{x}\)).
In the case of logistic regression predicting whether \(y_i = 1\) (for any \(K\)),
\begin{align}
I_{\beta \beta}^{\text{logistic}}(\alpha_1, \boldsymbol{\beta}) & =  J_1^{\boldsymbol{x} \boldsymbol{x}^\top\text{; logistic}}  + \tilde{J}_2^{\boldsymbol{x} \boldsymbol{x}^\top\text{; logistic}}   =   \int \boldsymbol{x} \boldsymbol{x}^\top \pi_1 (\boldsymbol{x}) [ 1 - \pi_1 (\boldsymbol{x}) ] \ d F(\boldsymbol{x})  , \nonumber 
\end{align}
matching \eqref{log.reg.beta.info}.

\subsection{Verifying the Asymptotic Distribution of Each Estimator}

By standard maximum likelihood theory (for example, the theorem on p. 145 of \citet[Section 4.2.2, multidimensional generalization on p. 148]{serfling1980}, or Theorem 13.1 in \citealt{wooldridge2010econometric}), the result holds if we can verify three regularity conditions. Before we do, we will define the set \(\boldsymbol{\Theta}\) of feasible parameters \((\boldsymbol{\alpha}, \boldsymbol{\beta})\). Since
\[
\alpha_1 < \alpha_2 < \ldots < \alpha_{K-1},
\]
where the strict inequality follows from our assumption that no class has probability 0 for any \(x \in [-1,1]\),
define the set \(\mathcal{A} \subset \mathbb{R}^{K-1}\) of points satisfying this constraint. Then define \(\boldsymbol{\Theta} := \mathcal{A} \times \mathbb{R}^p\), so that \((\boldsymbol{\alpha}, \boldsymbol{\beta}) \in \boldsymbol{\Theta} \). 

Now we state and verify the needed regularity conditions. 

\begin{enumerate}
\item (R1) \textit{The third derivatives of the log likelihood with respect to each parameter \((\boldsymbol{\alpha}, \boldsymbol{\beta})\) exist for all \(\boldsymbol{x} \in \mathcal{S}\).} This condition holds for both the proportional odds model and logistic regression because every entry of the Hessian matrices in Section \eqref{hessian.calc} are differentiable in every parameter for any \(K \geq 2\).

\item (R2) \textit{For each \((\boldsymbol{\alpha}_0, \boldsymbol{\beta}_0) \in \boldsymbol{\Theta}\), for all \((\boldsymbol{\alpha}, \boldsymbol{\beta})\) in a neighborhood of \((\boldsymbol{\alpha}_0, \boldsymbol{\beta}_0)\) it holds that (i) the element-wise absolute values of the gradients and Hessians of the likelihood are bounded by functions of \(\boldsymbol{x}\) with finite integral over \(\boldsymbol{x} \in \mathcal{S}\), and (ii) the element-wise absolute values of the third derivatives of the log likelihood is bounded by a function of \(\boldsymbol{x}\) with finite expectation with respect to \(\boldsymbol{X}\).} Because \(\boldsymbol{X}\) has bounded support, for these integrals and expectations to be finite it is enough for the bounding functions to over \(\mathcal{S}\) to be finite constants---that is, it is enough to find upper bounds on the absolute values of the gradients and Hessians of the likelihoods and the third derivatives of the log likelihoods. The logistic regression likelihood
\begin{align*}
& \prod_{i=1}^n \frac{  \exp \left\{  - \mathbbm{1} \left\{y_i = 1 \right\} (\alpha_1 + \boldsymbol{\beta}^\top \boldsymbol{x}_i ) \right\}  }{  1 + \exp\{-(\alpha_1 + \boldsymbol{\beta}^\top \boldsymbol{x}_i)\}  }
\end{align*}
has continuous second derivatives and therefore its gradient and Hessian both have finite element-wise absolute value on the bounded support. The same is true of the proportional odds likelihood \eqref{prop.odds.likelihood} when all outcomes have positive probability for all \(\boldsymbol{x} \in \mathcal{S}\), that is, \(\alpha_1 < \ldots < \alpha_{K-1}\). Finally, examining again the Hessian matrices in Section \ref{hessian.calc}, we see that they have continuous derivatives in every parameter for any \(K \geq 2\), so the third derivatives of the log likelihoods are bounded for any \((\boldsymbol{\alpha}_0, \boldsymbol{\beta}_0) \in \boldsymbol{\Theta}\) for all \(\boldsymbol{x} \in \mathcal{S}\).

\item (R3) \textit{The Fisher information matrices exist and are finite and positive definite.} One can see that both of the Fisher information matrices are finite entrywise for all \((\boldsymbol{\alpha}, \boldsymbol{\beta}) \in \boldsymbol{\Theta}\) by examining the matrices and noting that the probabilities for all of the classes are strictly greater than 0 over \(\mathcal{S}\) by assumption. To verify the positive definiteness of the Fisher information matrices
\[
-\E \left[ \pderiv{^2}{\boldsymbol{\theta} \boldsymbol{\theta}^\top} \mathcal{L}(\boldsymbol{\theta})  \right] ,
\]
it is enough to show that the log likelihood for each model is strictly concave, which implies that \(\pderiv{^2}{\boldsymbol{\theta} \boldsymbol{\theta}^\top} \mathcal{L}(\boldsymbol{\theta})\) is almost surely negative definite (since the log likelihood is twice differentiable). Strict concavity of the logistic regression log likelihood
\[
\sum_{i=1}^n \left[ - \mathbbm{1} \left\{y_i = 1 \right\} (\alpha_1 + \boldsymbol{\beta}^\top \boldsymbol{x}_i ) - \log \left( 1 + e^{-(\alpha_1 + \boldsymbol{\beta}^\top \boldsymbol{x}_i)} \right) \right]
\]
is easily seen, and \citet{Pratt1981} provides a proof that the log likelihood for the proportional odds model is strictly concave when the intercepts \(\alpha_1, \ldots, \alpha_{K-1}\) are not equal.

\end{enumerate}

Lastly, the asymptotic \(\mathcal{O}(1/n)\) bias is a consequence of standard maximum likelihood theory; see, for example, \citet{Cordeiro1991}.

\section{Proof of Theorem \ref{cons.thm}}\label{cons.thm.proof}

We prove Theorem \ref{cons.thm} in Section \ref{cons.thm.proof.subsec}, and Section \ref{cons.thm.proof.supp.res} contains proofs of the supporting lemmas.

\subsection{Proof of Theorem \ref{cons.thm}}\label{cons.thm.proof.subsec}

The proof will proceed as follows. First we will show that \textsc{PRESTO} is a member of a class of models described by \citet{Ekvall2022}, which means we can bound the estimation error of the parameters of \textsc{PRESTO} in a finite sample under their Theorem 3 once we show that its assumptions are satisfied. Their result depends on the probability of a particular event \(\mathcal{C}_{\kappa, n, p_n(K-1)} \), and in Proposition \ref{min.eigen.lem} we prove a lower bound on \(\mathbb{P} \left( \mathcal{C}_{\kappa, n, p_n(K-1)} \right) \) that tends towards 1 as \(n \to \infty\). This establishes the consistency of \textsc{PRESTO}.

In the notation of \citet{Ekvall2022}, we can express the objective function for the \textsc{PRESTO} estimator from \eqref{ordinal.pen.opt} as \( R \{b(y_i, \boldsymbol{x}_i, \boldsymbol{\theta})\} - R \{a(y_i, \boldsymbol{x}_i, \boldsymbol{\theta})\} + \lambda_n \lVert \boldsymbol{\theta} \rVert_1\), where \(R(\cdot) = F(\cdot)\), the logistic cumulative distribution function; \(\left(a(y_i, \boldsymbol{x}_i, \boldsymbol{\theta}); b(y_i, \boldsymbol{x}_i, \boldsymbol{\theta})\right)^\top =  \boldsymbol{Z}_i^\top \boldsymbol{\theta} + \boldsymbol{m}_i\) where
\begin{align*}
\boldsymbol{Z}_i  :=  ~ & \begin{pmatrix}
\mathbbm{1}\{y_i \geq 2\}  \boldsymbol{x}_i &  \mathbbm{1}\{y_i \leq K - 1 \}   \boldsymbol{x}_i
\\ \mathbbm{1}\{y_i \geq 3\}  \boldsymbol{x}_i & \mathbbm{1}\{2 \leq y_i \leq K - 1 \}  \boldsymbol{x}_i 
\\ \vdots & \vdots
\\  \mathbbm{1}\{y_i = K  \}  \boldsymbol{x}_i  & \mathbbm{1}\{y_i = K - 1\}  \boldsymbol{x}_i 
\end{pmatrix} \in \mathbb{R}^{p_n(K-1) \times 2},
\\ \boldsymbol{\theta} := ~ & \begin{pmatrix}
\boldsymbol{\beta}_1
\\ \boldsymbol{\psi}_2
\\ \vdots
\\ \boldsymbol{\psi}_{K-1}
\end{pmatrix} \in \mathbb{R}^{p_n(K-1)}, \qquad \text{and}
\\ \boldsymbol{m}_i := ~ & \begin{pmatrix}
\begin{cases} -\infty, & y_i = 1,
\\ \alpha_{k-1}, & y_i = k \ (k \geq 2)
\end{cases}
\\ 
\begin{cases} \alpha_k, & y_i = k \ (k < K - 1)
\\ \infty, & y_i = K
\end{cases}
\end{pmatrix} \in \overline{\mathbb{R}}^2,
\end{align*}
where \( \overline{\mathbb{R}} := \mathbb{R} \cup \{-\infty, \infty\}\) denotes the extended real number system; and, as elsewhere in the paper, \(\boldsymbol{\psi}_k = \boldsymbol{\beta}_k - \boldsymbol{\beta}_{k-1}\) for \(k \in \{2, \ldots, K - 1\}\). Observe that this is a special case of model (1) of \citet{Ekvall2022}. Now, since \(\boldsymbol{\Theta} = \mathbb{R}^{p_n(K-1)}\) is open and the standard logistic density \(r(t) = \exp\{-t\} / \left(1 + \exp\{-t\}\right)^2\) is strictly log-concave, strictly positive, and continuously differentiable on \(\mathbb{R}\), assumptions \((a)\) and \((b)\) of Theorem 3 in \citet{Ekvall2022} are satisfied, assumption \(S(C_4)\) is sufficient for assumption \((c)\), and assumption (e) is satisfied for \(C_3 = c_2\).

We now discuss Assumption 1 from \citet{Ekvall2022}. Note that a decision boundary crosses at some \(\boldsymbol{x}  \in \mathcal{S} \) if and only if for some \(k \in \{2, \ldots, K - 1\}\) it holds that
\begin{align*}
F \left(\alpha_{k} + \boldsymbol{\beta}_{k}^\top \boldsymbol{x} \right)  - F \left( \alpha_{k - 1} + \boldsymbol{\beta}_{k - 1}^\top \boldsymbol{x}  \right)  \leq ~ & 0
\\ \iff \qquad b(k, \boldsymbol{x}, \boldsymbol{\theta}) - a(k, \boldsymbol{x}, \boldsymbol{\theta}) \leq ~ & 0
.
\end{align*}
Observe that for \(k = 1\) it holds that \(m_{i1} = -\infty\) and for \(k = K\) it holds that \(m_{i2} = \infty\); that is, an element of \(\boldsymbol{m}_i\) is infinite. So Assumption \(T(C_4)\) is sufficient for it to either hold that \(\boldsymbol{Z}_i^\top \boldsymbol{\theta} + \boldsymbol{m}_i \in E\) where \(E \subseteq \{\boldsymbol{t} \in \mathbb{R}^2: t_1 < t_2\}\) or an element of \(\boldsymbol{m}_i\) is infinite. To satisfy Assumption 1, it only remains to show that there exists such an \(E\) that is compact. Note that \(\lVert \boldsymbol{\theta}_* \rVert_{\infty}\) is bounded under Assumption \(S(s, C_4)\) and \(\lVert \boldsymbol{\theta}_* \rVert_{0} \leq s\) under our sparsity assumption, so \(\lVert \boldsymbol{\theta}_* \rVert_1\) is bounded due to H\"{o}lder's inequality. Lastly \(\boldsymbol{m}_i\) is bounded because \(\max_{k \in \{1, \ldots, K-1\}} \left| \alpha_k \right| \leq C_4\) under Assumption \(T(C_4)\), so we have shown that we can find such an \(E\) that is bounded. Choose one that is closed and we have \(E\) compact by Bolzano-Weierstrass.

%

We have shown that the assumptions of Theorem 3 of \citet{Ekvall2022} are satisfied, so we have that for \(n\) large enough that \(p_n \geq K -1\), with probability at least \(\mathbb{P}\left(\mathcal{C}_{\kappa, n, p_n(K-1)} \right) - [p_n(K-1)]^{-c_3} \geq \mathbb{P}\left(\mathcal{C}_{\kappa, n, p_n(K-1)} \right) - (p_n^2)^{-c_3}  \)
\begin{align}
\left\lVert \boldsymbol{\hat{\theta}}^{\lambda_n} - \boldsymbol{\theta}_* \right\rVert_1 \leq  c_5 \sqrt{\frac{\log( p_n(K-1))}{n} } \leq c_5 \sqrt{\frac{\log( p_n^2 )}{n} }   = c_5 \sqrt{2} \sqrt{ \frac{\log p_n}{n}}, \label{thm3.prob}
\end{align}
where \(c_3\) and \(c_5\) are constants from \citet{Ekvall2022}, and \(\mathcal{C}_{\kappa, n, p_n(K-1)} \) is defined as follows. For a set \(\mathcal{A} \subseteq \{1, \ldots, p_n(K-1)\}\), define \(\boldsymbol{\theta}_{\mathcal{A}}  \in \mathbb{R}^{p_n(K-1)}\) to have \(j^{\text{th}}\) entry
\[
\left( \boldsymbol{\theta}_{\mathcal{A}} \right)_j := \begin{cases}
\boldsymbol{\theta}_j, & j \in \mathcal{A},
\\ 0, & j \notin \mathcal{A},
\end{cases}
\]
and define \(\boldsymbol{\theta}_{\mathcal{A}^c} := \boldsymbol{\theta} -   \boldsymbol{\theta}_{\mathcal{A}}  \). Then for an \(s\)-sparse (in the sense of Assumption \(S(s, c)\)) \(\boldsymbol{\theta}_*\) with support set \(\mathcal{S}\), define
\[
\mathbb{C}(\mathcal{S}) := \left\{ \boldsymbol{\theta} \in \mathbb{R}^{p_n(K-1)}: \left\lVert  \boldsymbol{\theta}_{\mathcal{S}^c} \right\rVert_1 \leq 3   \left\lVert  \boldsymbol{\theta}_{\mathcal{S}} \right\rVert_1 \right\}.
\]
We interpret \(\mathbb{C}(\mathcal{S}) \) to be a set of approximately \(s\)-sparse vectors (the vectors would be exactly \(s\)-sparse if \(\left\lVert  \boldsymbol{\theta}_{\mathcal{S}^c} \right\rVert_1 = 0\)). Then for \(\kappa > 0 \) and \(n \in \mathbb{N}\), define
\[
\mathcal{C}_{\kappa, n, p_n(K-1)}  := \left\{ (\boldsymbol{X}, \boldsymbol{y}):  \inf_{\{\boldsymbol{\theta} \in \mathbb{C}(\mathcal{S}): \lVert \boldsymbol{\theta} \rVert_2 = 1\}} \left\{ \boldsymbol{\theta}^\top \left(   \frac{1}{n} \sum_{i=1}^n  \boldsymbol{Z}_i \boldsymbol{Z}_i^\top \right) \boldsymbol{\theta} \right\} \geq  \kappa \right\}.
\]
If the set of \(\boldsymbol{\theta}\) over which this condition must hold were \(\mathbb{R}^{p_n(K-1)}\), this would be a minimum eigenvalue condition on \( \frac{1}{n} \sum_{i=1}^n  \boldsymbol{Z}_i \boldsymbol{Z}_i^\top \). This condition is sometimes called a \textit{restricted eigenvalue condition} \citep{PeterBickel2009}. We bound \(\mathbb{P}\left(\mathcal{C}_{\kappa, n, p_n(K-1)} \right)\) in Proposition \ref{min.eigen.lem}.

\begin{proposition}\label{min.eigen.lem}
Suppose the assumptions of Theorem \ref{cons.thm} hold. Let 
\[
\pi_{\text{rare,min}} := \inf_{\boldsymbol{x} \in \mathcal{S}, k \in \{1, \ldots, K\}} \left\{ \mathbb{P}(y_i = k \mid \boldsymbol{x}) \right\},
\]
and observe that \(\pi_{\text{rare,min}}  > 0\) under Assumptions \(X(\mathbb{R}^{p_n})\), \(S(s, C_4)\), and \(T(C_4)\). 
Assume \(n\) is large enough so that 
\begin{equation}\label{a.eigen.eq}
 n \pi_{\text{rare, min}} >  \max \left\{2 \left( C \sqrt{p_n} + \sqrt{\frac{a}{\lambda_{\text{min}}^*}} \right)^2, 2 \right\}
 \end{equation}
  for some \(a > 0\) (recall from the statement of Theorem \ref{cons.thm} that we assumed \(\lambda_{\text{min}}^* := \min_{k \in \{1, \ldots, K\}}  \lambda_{\text{min}} \left(  \E \left[  \boldsymbol{x}_i \boldsymbol{x}_i^\top \mid y_i = k\right]  \right)  > b\) for a fixed \(b>0\)). Then 
\[
\mathbb{P}\left(\mathcal{C}_{a \pi_{\text{rare,min}}/4, n, p_n(K-1)} \right) \geq  1 - 2K \exp \left\{ -\frac{1}{2} \pi_{\text{rare,min}}^2 n\right\} -  2K \exp \left\{-c \left( \sqrt{  \frac{ n \pi_{\text{rare,min}}}{2}}  - C \sqrt{p_n} - \sqrt{\frac{a}{\lambda_{\text{min}}^*}} \right)^2 \right\}  
.
\]

\end{proposition}

\begin{proof} Provided in Section \ref{cons.thm.proof.supp.res}.
\end{proof}

Observe that since \(p_n \leq C_1 n^{C_2}\) for some \(C_1 > 0, C_2 \in (0, 1)\) this probability tends to 1 as \(n \to \infty\). Lemma \ref{theta.beta.bounds}, below, along with \eqref{thm3.prob} then shows that with probability at least 
\begin{equation}\label{prob.bound}
  1  -  p_n^{-C_5}  - 2K \exp \left\{ -\frac{1}{2} \pi_{\text{rare,min}}^2 n\right\} -  2K \exp \left\{-c \left( \sqrt{  \frac{ n \pi_{\text{rare,min}}}{2}}  - C \sqrt{p_n} - \sqrt{\frac{a}{\lambda_{\text{min}}^*}} \right)^2 \right\}
\end{equation}
  it holds that
\[
\left\lVert \boldsymbol{\hat{\beta}}^{\lambda_n} - \boldsymbol{\beta} \right\rVert_2 \leq    \left\lVert  \boldsymbol{\hat{\beta}}^{\lambda_n} - \boldsymbol{\beta}  \right\rVert_1 \leq  C_6 \sqrt{ \frac{\log p_n}{n}}
,
\]
where \(C_5 := 2c_3\) and \(C_6 :=  c_5 \sqrt{2}(K-1) \) (where \(c_5\) depends on the sparsity level \(s\)).
\begin{lemma}\label{theta.beta.bounds}
\(\left\lVert  \boldsymbol{\hat{\beta}}^{\lambda_n} - \boldsymbol{\beta}  \right\rVert_2  \leq \left\lVert   \boldsymbol{\hat{\beta}}^{\lambda_n} - \boldsymbol{\beta} \right\rVert_1  \leq (K-1) \left\lVert  \boldsymbol{\hat{\theta}}^{\lambda_n} - \boldsymbol{\theta}_* \right\rVert_1 \).
\end{lemma}

\begin{proof} Provided in Section \ref{cons.thm.proof.supp.res}.
\end{proof}

Finally, we can now show consistency by showing that the random variable \(\left\lVert  \boldsymbol{\hat{\beta}}^{\lambda_n} - \boldsymbol{\beta}   \right\rVert_2 \) converges in probability to 0. For any \(\epsilon > 0\),
\begin{align*}
& \lim_{n \to \infty} \mathbb{P}\left(  \left\lVert  \boldsymbol{\hat{\beta}}^{\lambda_n} - \boldsymbol{\beta}   \right\rVert_2  < \epsilon \right) 
 \\ \stackrel{(*)}{\geq} ~ &  \lim_{n \to \infty} \mathbb{P}\left(  \left\lVert  \boldsymbol{\hat{\beta}}^{\lambda_n} - \boldsymbol{\beta}  \right\rVert_2< C_6 \sqrt{ \frac{\log p_n}{n}}  \right)  
\\ \geq ~ &   \lim_{n \to \infty} \left(    1  -  p_n^{-C_5}  - 2K \exp \left\{ -\frac{1}{2} \pi_{\text{rare,min}}^2 n\right\} -  2K \exp \left\{-c \left( \sqrt{  \frac{ n \pi_{\text{rare,min}}}{2}}  - C \sqrt{p_n} - \sqrt{\frac{a}{\lambda_{\text{min}}^*}} \right)^2 \right\} \right)
\\ = ~ &  1,
\end{align*}
where \((*)\) follows because for large enough \(n\), \(\epsilon >  C_6 \sqrt{ \log (p_n) /n}\). This establishes Theorem \ref{cons.thm}. All that remains is to provide proofs for the supporting lemmas and proposition, which we do in the following section.

\begin{remark}
Observe that for finite \(n\) we have that for \(n\) large enough, with probability at least equal to the expression in \eqref{prob.bound} it holds that
\[
\sqrt{n} \left\lVert \boldsymbol{\hat{\beta}}^{\lambda_n} - \boldsymbol{\beta} \right\rVert_2 \leq   \sqrt{n}  \left\lVert  \boldsymbol{\hat{\beta}}^{\lambda_n} - \boldsymbol{\beta}  \right\rVert_1 \leq C_6 \sqrt{\log p_n}
.
\]
This establishes a high-probability finite sample relationship between our theoretical guarantee and the asymptotic mean squared error from Definition \ref{def.asym.mse}.

Also note that for any \(\epsilon > 0\), there exists an \(n_\epsilon\) such that the probability in \eqref{prob.bound} is less than \(\epsilon\) for all \(n \geq n_\epsilon\). Since there also exists an \(n_2\) large enough so that for all \(n \geq n_2\) it holds that
\[
\epsilon \geq C_6 \sqrt{\frac{\log p_n}{n}} 
,
\]
 we have for any \(\epsilon > 0\), for all \(n \geq \max \{n_\epsilon, n_2\}\)
\begin{align*}
 \mathbb{P} \left( \frac{  \left\lVert  \boldsymbol{\hat{\beta}}^{\lambda_n} - \boldsymbol{\beta}  \right\rVert_1}{\sqrt{\log( p_n)/n}} \geq C_6 \right)
 \leq  ~ & 1  -  p_n^{-C_5}  - 2K \exp \left\{ -\frac{1}{2} \pi_{\text{rare,min}}^2 n\right\} -  2K \exp \left\{-c \left( \sqrt{  \frac{ n \pi_{\text{rare,min}}}{2}}  - C \sqrt{p_n} - \sqrt{\frac{a}{\lambda_{\text{min}}^*}} \right)^2 \right\}
\\ < ~ & \epsilon
;
\end{align*}
that is, \( \left\lVert  \boldsymbol{\hat{\beta}}^{\lambda_n} - \boldsymbol{\beta}  \right\rVert_1 \in \mathcal{O}_p(\sqrt{\log (p_n)/n})\), and likewise for \(\left\lVert  \boldsymbol{\hat{\beta}}^{\lambda_n} - \boldsymbol{\beta}  \right\rVert_2\). Again, we emphasize that \(C_6\) depends on the sparsity \(s\) (though the above holds since we have assumed that \(s\) is fixed).
\end{remark}

\subsection{Supporting Results for Proof of Theorem \ref{cons.thm}}\label{cons.thm.proof.supp.res}

\begin{proof}[Proof of Proposition \ref{min.eigen.lem}]
We will show that there is a high probability event such that for \(a >0\) it holds that
\[
  \inf_{\{\boldsymbol{\theta} \in \mathbb{C}(\mathcal{S}): \lVert \boldsymbol{\theta} \rVert_2 = 1\}} \left\{ \boldsymbol{\theta}^\top \left(   \frac{1}{n} \sum_{i=1}^n  \boldsymbol{Z}_i \boldsymbol{Z}_i^\top \right) \boldsymbol{\theta} \right\} >  \frac{1}{2}a \pi_{\text{rare,min}}
  .
\]
Let
\[
\boldsymbol{A}^{(1)} = \begin{pmatrix} \boldsymbol{0}_{K-1} & \boldsymbol{e}_1 \end{pmatrix} \in \mathbb{R}^{(K-1)\times 2}
;
\]
\[
 \boldsymbol{A}^{(k)}= \begin{pmatrix} \boldsymbol{1}_{k-1} & \boldsymbol{1}_{k-1}
\\ 0 & 1
\\ \boldsymbol{0}_{K-k - 1} & \boldsymbol{0}_{K-k -1} \end{pmatrix} = \begin{pmatrix} \sum_{k'=1}^{k-1} \boldsymbol{e}_{k'} &  \sum_{k'=1}^{k} \boldsymbol{e}_{k'} \end{pmatrix} \in \mathbb{R}^{(K-1)\times 2}, ~  k \in \{2, \ldots, K-1\};
\]
and
\[
 \boldsymbol{A}^{(K)} =   \begin{pmatrix} \boldsymbol{1}_{K-1} & \boldsymbol{0}_{K-1} \end{pmatrix}  =   \begin{pmatrix} \sum_{k'=1}^{K-1} \boldsymbol{e}_{k'}  & \boldsymbol{0}_{K-1} \end{pmatrix}  \in \mathbb{R}^{(K-1) \times 2},
\]
where \(\boldsymbol{0}_n\) and \(\boldsymbol{1}_n\) are \(n\)-vectors of zeroes and ones (respectively) and \(\boldsymbol{e}_k\) is the standard basis vector in \(\mathbb{R}^{K-1}\) with a 1 in the \(k^{\text{th}}\) entry and zeroes elsewhere, and \(\boldsymbol{A}^{(k)}\in  \mathbb{R}^{(K-1) \times 2}\) for all \(k\). Note that
\begin{equation}\label{z.i.kron}
\boldsymbol{Z}_i = \boldsymbol{A}^{(y_i)} \otimes \boldsymbol{x}_i.
\end{equation}

Let
\[
\boldsymbol{B} = \begin{pmatrix}
1 & 1 &   1 & \cdots & 1
\\ 0 &  1 & 1 &  \cdots & 1
\\ 0 & 0 & 1 & \cdots & 1
\\ \vdots & \vdots & \vdots & \ddots & \vdots
\\ 0 & 0 & 0 & \cdots & 1
\end{pmatrix} = \begin{pmatrix}  \boldsymbol{e}_1 & \sum_{k'=1}^2 \boldsymbol{e}_{k'} & \cdots &  \sum_{k'=1}^{K-1} \boldsymbol{e}_{k'} \end{pmatrix} \in \mathbb{R}^{(K-1)\times(K-1)};
\]
that is, the columns of \(\boldsymbol{B}\) are \(\boldsymbol{B}_k =  \sum_{\ell=1}^k \boldsymbol{e}_{\ell}\). We will make use of the following lemmas, which we prove later in this section:

\begin{lemma}\label{a.id.lem}
\[
\sum_{k=1}^K \frac{n_k}{n} \boldsymbol{A}^{(k)}\left(\boldsymbol{A}^{(k)}\right)^\top = \boldsymbol{B} \boldsymbol{D} \boldsymbol{B}^\top
,
\]
where
\[
\boldsymbol{D} := \operatorname{diag} \left(  \frac{n_1 + n_2}{n}, \ldots, \frac{n_{K-1} + n_K}{n} \right) \in \mathbb{R}^{(K-1)\times(K-1)}
.
\]

\end{lemma}

\begin{lemma}\label{b.op.lem}
\(\sigma_{\text{min}}^2 \left( \boldsymbol{B} \right)  \geq 1/2 \).
\end{lemma}

\begin{lemma}\label{jacob.lem.2}
For arbitrary matrices \(\boldsymbol{\overline{A}}\), \(\boldsymbol{\overline{B}}\), and \(\boldsymbol{\overline{C}}\), if \(\boldsymbol{\overline{B}} \succeq \boldsymbol{\overline{C}}\) and \(\boldsymbol{\overline{A}} \succeq \boldsymbol{0}\) then \(\boldsymbol{\overline{A}} \otimes \boldsymbol{\overline{B}} \succeq \boldsymbol{\overline{A}} \otimes \boldsymbol{\overline{C}}\).
\end{lemma}

\begin{lemma}\label{pos.def.lem}
 With probability at least \( 1 - K \exp \left\{ -\frac{1}{2} \pi_{\text{rare,min}}^2 n\right\}\) it holds that \(\min \{n_1, \ldots, n_K\} > \pi_{\text{rare,min}}n/2  \), where \(n_k := \sum_{i=1}^n \mathbbm{1} \left\{ y_i = k\right\}\) is the number of observations in class \(k\).
\end{lemma}

\begin{lemma}\label{new.a.min.lem}
Under the assumptions of Proposition \ref{min.eigen.lem}, there exist constants \(c, C > 0\) such that \(\lambda_{\text{min}} \left( \frac{1}{n_k}  \sum_{i: y_i = k} \boldsymbol{x}_i \boldsymbol{x}_i^\top \right) \geq a\) for all \(k \in \{1, \ldots, K\}\) with probability at least
\begin{align*}
 1 -  2K \exp \left\{-c \left( \sqrt{  \frac{ n \pi_{\text{rare,min}}}{2}}  - C \sqrt{p_n} - \sqrt{\frac{a}{\lambda_{\text{min}}^*}} \right)^2 \right\}  -   K \exp \left\{ -\frac{1}{2} \pi_{\text{rare,min}}^2 n\right\}
,
\end{align*}
where \(n_k\) is the number of observations in class \(k\) and \(\lambda_{\text{min}}^* := \min_{k \in \{1, \ldots, K\}} \left\{  \lambda_{\text{min}} \left( \E \left[ \boldsymbol{X}^\top \boldsymbol{X} \mid y = k \right]  \right) \right\}\).
\end{lemma}

Using a union bound, there exists an event with probability at least
\[
1 - 2K \exp \left\{ -\frac{1}{2} \pi_{\text{rare,min}}^2 n\right\} -  2K \exp \left\{-c \left( \sqrt{  \frac{ n \pi_{\text{rare,min}}}{2}}  - C \sqrt{p_n} - \sqrt{\frac{a}{\lambda_{\text{min}}^*}} \right)^2 \right\}  
\]
on which the conclusions of all of the above lemmas hold for \(n\) large enough. On this event we have for \(n\) sufficiently large
\begin{align*}
\inf_{\{\boldsymbol{\theta} \in \mathbb{C}(\mathcal{S}): \lVert \boldsymbol{\theta} \rVert_2 = 1\}} \left\{ \boldsymbol{\theta}^\top \left(   \frac{1}{n} \sum_{i=1}^n  \boldsymbol{Z}_i \boldsymbol{Z}_i^\top \right) \boldsymbol{\theta} \right\}   \geq ~ & \inf_{\{\boldsymbol{\theta} \in \mathbb{R}^{p_n(K-1)}: \lVert \boldsymbol{\theta} \rVert_2 = 1\}} \left\{ \boldsymbol{\theta}^\top \left(   \frac{1}{n} \sum_{i=1}^n  \boldsymbol{Z}_i \boldsymbol{Z}_i^\top \right) \boldsymbol{\theta} \right\}  
\\ = ~ & \lambda_{\text{min}} \left( \frac{1}{n} \sum_{i=1}^n \boldsymbol{Z}_i \boldsymbol{Z}_i^\top \right) 
\\  \stackrel{(a)}{=} ~ & \lambda_{\text{min}} \left( \frac{1}{n}  \sum_{i=1}^n  \left( \boldsymbol{A}^{(y_i)} \otimes \boldsymbol{x}_i \right) \left(  \boldsymbol{A}^{(y_i)} \otimes \boldsymbol{x}_i \right)^\top \right) 
\\ =  ~ & \lambda_{\text{min}} \left( \frac{1}{n}   \sum_{i=1}^n \left( \boldsymbol{A}^{(y_i)}  \left(\boldsymbol{A}^{(y_i)}\right)^\top  \right) \otimes \left( \boldsymbol{x}_i \boldsymbol{x}_i^\top  \right)\right) 
\\ =  ~ & \lambda_{\text{min}} \left( \frac{1}{n}  \sum_{k=1}^{K}  \sum_{i:y_i=k}\left( \boldsymbol{A}^{(y_i)}  \left(\boldsymbol{A}^{(y_i)}\right)^\top  \right) \otimes \left( \boldsymbol{x}_i \boldsymbol{x}_i^\top  \right)\right) 
%
\\ =  ~ & \lambda_{\text{min}} \left( \frac{1}{n}  \sum_{k=1}^{K}  \left( \boldsymbol{A}^{(k)}  \left( \boldsymbol{A}^{(k)} \right)^\top  \right) \otimes \left(  \sum_{i:y_i=k} \boldsymbol{x}_i \boldsymbol{x}_i^\top  \right)\right) 
\\ =  ~ &  \lambda_{\text{min}} \left(  \sum_{k=1}^K \left( \frac{n_k}{n} \boldsymbol{A}^{(k)}\left(\boldsymbol{A}^{(k)}\right)^\top \right)  \otimes \left( \frac{1}{n_k} \sum_{i: y_i = k} \boldsymbol{x}_i \boldsymbol{x}_i^\top  \right) \right)
\\ \stackrel{(b)}{\geq}  ~ &  \lambda_{\text{min}} \left(  \sum_{k=1}^K \left( \frac{n_k}{n} \boldsymbol{A}^{(k)}\left(\boldsymbol{A}^{(k)}\right)^\top \right)  \otimes \left( a \boldsymbol{I}_p \right) \right)
\\ \stackrel{(c)}{=} ~ &    \lambda_{\text{min}} \left(  \sum_{k=1}^K \frac{n_k}{n} \boldsymbol{A}^{(k)}\left(\boldsymbol{A}^{(k)}\right)^\top \right) \lambda_{\text{min}} \left(a \boldsymbol{I}_p \right)  
\\ \stackrel{(d)}{=} ~ &  \lambda_{\text{min}} \left(  \boldsymbol{B} \boldsymbol{D} \boldsymbol{B}^\top \right)  \lambda_{\text{min}} \left( a \boldsymbol{I}_p\right) 
\\ \stackrel{(e)}{\geq} ~ &  a  \min_{k \in \{1, \ldots, K-1\}} \left\{ \frac{n_k + n_{k+1}}{n}  \right\}\lambda_{\text{min}} \left(  \boldsymbol{B} \boldsymbol{B}^\top \right)
\\ \stackrel{(f)}{\geq} ~ &  \frac{a}{4} \min_{k \in \{1, \ldots, K-1\}} \left\{ \frac{n_k + n_{k+1}}{n}  \right\}
\\ \stackrel{(g)}{>} ~ &  \frac{a\pi_{\text{rare,min}}}{4}
,
\end{align*}
where in \((a)\) we used \eqref{z.i.kron}, \((b)\) holds on the high probability event by Lemmas \ref{jacob.lem.2} and \ref{new.a.min.lem} because \(\lambda_{\text{min}} \left( \frac{1}{n_k}  \sum_{i: y_i = k} \boldsymbol{x}_i \boldsymbol{x}_i^\top \right) \geq a\) implies \(\frac{1}{n_k} \sum_{i: y_i = k} \boldsymbol{x}_i \boldsymbol{x}_i^\top  \succeq a \boldsymbol{I}_p\) for all \(k\), in \((c)\) we used the fact that for matrices \(\boldsymbol{M}, \boldsymbol{N}\) the eigenvalues of \(\boldsymbol{M} \otimes \boldsymbol{N}\) are the products of the eigenvalues of \(\boldsymbol{M}\) and \(\boldsymbol{N}\) and both of the factor matrices are positive semidefinite, in \((d)\) we applied Lemma \ref{a.id.lem}, \((e)\) follows from \(  \boldsymbol{B} \boldsymbol{D} \boldsymbol{B}^\top \succeq  \min_{k \in \{1, \ldots, K-1\}} \left\{ \frac{n_k + n_{k+1}}{n}  \right\}\boldsymbol{B} \boldsymbol{B}^\top\), in \((f)\) we applied Lemma \ref{b.op.lem}, and \((g)\) follows from Lemma \ref{pos.def.lem}. That is, on this high probability event it holds that 
\[
\inf_{\{\boldsymbol{\theta} \in \mathbb{C}(\mathcal{S}): \lVert \boldsymbol{\theta} \rVert_2 = 1\}} \left\{ \boldsymbol{\theta}^\top \left(   \frac{1}{n} \sum_{i=1}^n  \boldsymbol{Z}_i \boldsymbol{Z}_i^\top \right) \boldsymbol{\theta} \right\}  >   \frac{a\pi_{\text{rare,min}}}{4}
,
\]
and the conclusion follows.

\end{proof}

\begin{proof}[Proof of Lemma \ref{theta.beta.bounds}] 
For convenience, denote \(\boldsymbol{\theta}_1 = \boldsymbol{\beta}_1\) and \(\boldsymbol{\theta}_k = \boldsymbol{\psi}_k\), \(k \in \{2, \ldots, K-1\}\). Let
\begin{align*}
 \boldsymbol{\hat{\epsilon}}_{\boldsymbol{\theta}}^{\lambda_n} = ~ & \left( \left(   \boldsymbol{\hat{\epsilon}}_{\boldsymbol{\theta},1}^{\lambda_n} \right)^\top, \left(  \boldsymbol{\hat{\epsilon}}_{\boldsymbol{\theta},2}^{\lambda_n} \right)^\top, \ldots, \left(  \boldsymbol{\hat{\epsilon}}_{\boldsymbol{\theta},K-1}^{\lambda_n} \right)^\top \right)^\top
 :=   \left( \left( \boldsymbol{\hat{\theta}}_1^{\lambda_n} - \boldsymbol{\theta}_1\right)^\top, \left( \boldsymbol{\hat{\theta}}_2^{\lambda_n} - \boldsymbol{\theta}_2\right)^\top, \ldots, \left( \boldsymbol{\hat{\theta}}_{K-1}^{\lambda_n} - \boldsymbol{\theta}_{K-1}\right)^\top \right)^\top
 .
 \end{align*}
Note that \(\boldsymbol{\beta}_k = \sum_{k' \leq k} \boldsymbol{\theta}_{k'}\). Let \( \boldsymbol{\beta} := (\boldsymbol{\beta}_1^\top, \boldsymbol{\beta}_2^\top, \ldots, \boldsymbol{\beta}_{K-1}^\top)^\top\), and denote by
 \[
 \boldsymbol{\hat{\beta}}_k^{\lambda_n} := \sum_{k' \leq k} \boldsymbol{\hat{\theta}}_{k'}, \qquad k \in \{1, \ldots, K-1\}
 \]
 the estimates of \(\boldsymbol{\beta}\) yielded from the estimates of \(\boldsymbol{\theta}\). Let
\begin{align*}
 \boldsymbol{\hat{\epsilon}}_{\boldsymbol{\beta}}^{\lambda_n}  = ~ & \left( \left(   \boldsymbol{\hat{\epsilon}}_{\boldsymbol{\beta},1}^{\lambda_n} \right)^\top, \left(  \boldsymbol{\hat{\epsilon}}_{\boldsymbol{\beta},2}^{\lambda_n} \right)^\top, \ldots, \left(  \boldsymbol{\hat{\epsilon}}_{\boldsymbol{\beta},K-1}^{\lambda_n} \right)^\top \right)^\top
 \\ := ~ &  \left( \left( \boldsymbol{\hat{\beta}}_1^{\lambda_n} - \boldsymbol{\beta}_1\right)^\top, \left( \boldsymbol{\hat{\beta}}_2^{\lambda_n} - \boldsymbol{\beta}_2\right)^\top, \ldots, \left( \boldsymbol{\hat{\beta}}_{K-1}^{\lambda_n} - \boldsymbol{\beta}_{K-1}\right)^\top \right)^\top
 ,
\end{align*}
 and observe that
 \begin{align}
 \boldsymbol{\hat{\epsilon}}_{\boldsymbol{\beta}}^{\lambda_n} = ~ &  \left( \left( \boldsymbol{\hat{\beta}}_1^{\lambda_n} - \boldsymbol{\beta}_1\right)^\top, \left( \boldsymbol{\hat{\beta}}_2^{\lambda_n} - \boldsymbol{\beta}_2\right)^\top, \ldots, \left( \boldsymbol{\hat{\beta}}_{K-1}^{\lambda_n} - \boldsymbol{\beta}_{K-1}\right)^\top \right)^\top \nonumber
 \\  = ~ &  \left( \left( \boldsymbol{\hat{\theta}}_1^{\lambda_n} - \boldsymbol{\theta}_1\right)^\top, \left(  \sum_{k' \leq 2} \boldsymbol{\hat{\theta}}_{k'} -   \sum_{k' \leq 2} \boldsymbol{\theta}_{k'} \right)^\top, \ldots, \left(  \sum_{k' \leq K -1} \boldsymbol{\hat{\theta}}_{k'} -  \sum_{k' \leq K-1} \boldsymbol{\theta}_{k'} \right)^\top \right)^\top \nonumber
\\  = ~ &  \left( \left(   \boldsymbol{\hat{\epsilon}}_{\boldsymbol{\theta},1}^{\lambda_n} \right)^\top, \left(  \sum_{k' \leq 2} \boldsymbol{\hat{\epsilon}}_{\boldsymbol{\theta},k'}^{\lambda_n} \right)^\top, \ldots, \left(  \sum_{k' \leq K -1}  \boldsymbol{\hat{\epsilon}}_{\boldsymbol{\theta},k'}^{\lambda_n} \right)^\top \right)^\top \label{theta.sum}
.
 \end{align}
Then
%
\begin{align*}
\left\lVert \boldsymbol{\hat{\epsilon}}_{\boldsymbol{\beta},1}^{\lambda_n} \right\rVert_1 
= ~ &  \sum_{k=1}^{K-1}   \left\lVert \boldsymbol{\hat{\epsilon}}_{\boldsymbol{\beta},k}^{\lambda_n} \right\rVert_1 
\stackrel{(a)}{=} \sum_{k=1}^{K-1}   \left\lVert   \sum_{k' \leq k} \boldsymbol{\hat{\epsilon}}_{\boldsymbol{\theta},k'}^{\lambda_n} \right\rVert_1
\stackrel{(b)}{\leq}  \sum_{k=1}^{K-1}  \sum_{k' \leq k}   \left\lVert   \boldsymbol{\hat{\epsilon}}_{\boldsymbol{\theta},k'}^{\lambda_n} \right\rVert_1 = \sum_{k=1}^{K-1}  (K - k) \left\lVert   \boldsymbol{\hat{\epsilon}}_{\boldsymbol{\theta},k}^{\lambda_n} \right\rVert_1  
\\ \leq ~ &   (K - 1)  \sum_{k=1}^{K-1}  \left\lVert   \boldsymbol{\hat{\epsilon}}_{\boldsymbol{\theta},k}^{\lambda_n} \right\rVert_1 = (K-1) \left\lVert \boldsymbol{\hat{\epsilon}}_{\boldsymbol{\theta}}^{\lambda_n} \right\rVert_1.
\end{align*}
where in \((a)\) we used \eqref{theta.sum} and in \((b)\) we used the triangle inequality. Lastly, \( \left\lVert  \boldsymbol{\hat{\beta}}^{\lambda_n} - \boldsymbol{\beta}  \right\rVert_2  \leq \left\lVert   \boldsymbol{\hat{\beta}}^{\lambda_n} - \boldsymbol{\beta} \right\rVert_1\) is a property of the \(\ell_1\) and \(\ell_2\) norms.

\end{proof}

\begin{proof}[Proof of Lemma \ref{a.id.lem}]
\begin{align*}
& \sum_{k=1}^K \frac{n_k}{n} \boldsymbol{A}^{(k)}\left(\boldsymbol{A}^{(k)}\right)^\top 
\\ = ~ & \frac{n_1}{n}  \boldsymbol{A}^{(1)} \left(\boldsymbol{A}^{(1)}\right)^\top + \sum_{k=2}^{K-1} \frac{n_k}{n} \boldsymbol{A}^{(k)}\left(\boldsymbol{A}^{(k)}\right)^\top + \frac{n_K}{n} \boldsymbol{A}^{(K)} \left(\boldsymbol{A}^{(K)}\right)^\top
%
%
\\ = ~ & \frac{n_1}{n}  \begin{pmatrix} \boldsymbol{0}_{K-1} & \boldsymbol{e}_1 \end{pmatrix} \begin{pmatrix} \boldsymbol{0}_{K-1}^\top \\ \boldsymbol{e}_1^\top \end{pmatrix}
 + \sum_{k=2}^{K-1} \frac{n_k}{n} \begin{pmatrix} \sum_{k'=1}^{k-1} \boldsymbol{e}_{k'} & \sum_{k'=1}^{k} \boldsymbol{e}_{k'}  \end{pmatrix}  \begin{pmatrix} \sum_{k'=1}^{k-1} \boldsymbol{e}_{k'}^\top \\ \sum_{k'=1}^{k} \boldsymbol{e}_{k'}^\top  \end{pmatrix} 
\\ & + \frac{n_K}{n} \begin{pmatrix} \sum_{k'=1}^{K-1} \boldsymbol{e}_{k'} & \boldsymbol{0}_{K-1} \end{pmatrix} \begin{pmatrix} \sum_{k'=1}^{K-1} \boldsymbol{e}_{k'}^\top \\ \boldsymbol{0}_{K-1}^\top \end{pmatrix}
\\ = ~ & \frac{n_1}{n}  \left( \boldsymbol{0}_{K-1} \boldsymbol{0}_{K-1}^\top +  \boldsymbol{e}_1  \boldsymbol{e}_1^\top  \right)
+ \sum_{k=2}^{K-1} \frac{n_k}{n} \left(  \sum_{k'=1}^{k-1} \boldsymbol{e}_{k'}  \sum_{k'=1}^{k-1} \boldsymbol{e}_{k'}^\top +  \sum_{k'=1}^{k} \boldsymbol{e}_{k'}  \sum_{k'=1}^{k} \boldsymbol{e}_{k'}^\top \right)
\\ & + \frac{n_K}{n}  \left( \sum_{k'=1}^{K-1} \boldsymbol{e}_{k'}  \sum_{k'=1}^{K-1} \boldsymbol{e}_{k'}^\top  + \boldsymbol{0}_{K-1} \boldsymbol{0}_{K-1}^\top \right)
%
%
\\ = ~ & \frac{n_1}{n}  \boldsymbol{B}_1  \boldsymbol{B}_1^\top  
+ \sum_{k=2}^{K-1} \frac{n_k}{n} \left(  \boldsymbol{B}_{k-1}  \boldsymbol{B}_{k-1}^\top +  \boldsymbol{B}_{k}   \boldsymbol{B}_{k}^\top \right)
 + \frac{n_K}{n}  \boldsymbol{B}_{K-1} \boldsymbol{B}_{K-1}^\top 
\\ =  ~ &  \sum_{k=1}^{K-1} \frac{n_k + n_{k+1}}{n} \boldsymbol{B}_k \boldsymbol{B}_k^\top
\\ = ~ &  \boldsymbol{B} \boldsymbol{D} \boldsymbol{B}^\top.
\end{align*}
\end{proof}

\begin{proof}[Proof of Lemma \ref{b.op.lem}]
\(\boldsymbol{B}\) is full rank with inverse
\[
\boldsymbol{B}^{-1} = \begin{pmatrix}
1 & -1 & 0 &  0  & \cdots & 0 
\\ 0 & 1 & -1  & 0 & \cdots & 0 
\\ 0 & 0 & 1 & -1   & \cdots & 0 
\\ \vdots & \vdots & \vdots & \vdots & \ddots & \vdots 
\\ 0 & 0 & 0 & 0 & \cdots & 1 
\end{pmatrix}.
\]
We have
\[
\lVert \boldsymbol{B}^{-1} \rVert_{\text{op}} \leq \sqrt{  \lVert \boldsymbol{B}^{-1} \rVert_{1}  \lVert \boldsymbol{B}^{-1} \rVert_{\infty} } = \sqrt{2 \cdot 2} = 2,
\]
where we have used that the \(\ell_1\) norm of each column and row is at most 2. Then the result follows from \(\sigma_{\text{min}} \left( \boldsymbol{B} \right) = 1/ \lVert \boldsymbol{B}^{-1} \rVert_{\text{op}}\).
\end{proof}

\begin{proof}[Proof of Lemma \ref{jacob.lem.2}]
The matrix \(\boldsymbol{\overline{A}} \otimes (\boldsymbol{\overline{B}} - \boldsymbol{\overline{C}})\) is the Kronecker product of two positive semidefinite matrices and is therefore positive semidefinite.
\end{proof}

\begin{proof}[Proof of Lemma \ref{pos.def.lem}]

First we state a lemma we will use.

\begin{lemma}\label{bin.prob.bound}
Suppose that for all \(\boldsymbol{x} \in \mathcal{S}\) and \(k \in \{1, \ldots, K\}\) it holds that \(\mathbb{P}(y_i = k \mid \boldsymbol{x}) \geq \pi_{\text{rare,min}}\). Let \(n_k\) be the number of observations in class \(k\) from a data set of size \(n\). Then for any \(q \in \{1, \ldots, \lfloor n\pi_{\text{rare,min}} \rfloor\}\),
\[
\mathbb{P} \left(  \bigcap_{k=1}^K \left\{ n_k >   q  \right\} \right)  \geq   1 - K \exp \left\{ -2n \left( \pi_{\text{rare,min}} - \frac{q}{n} \right)^2\right\}
.
\]

\end{lemma}

\begin{proof} Provided later in this section.
\end{proof}



Note that \(n \pi_{\text{rare,min}}/2 \leq \lfloor n \pi_{\text{rare,min}}  \rfloor\) because \(x/2 \leq \lfloor x \rfloor\) for all \(x \geq 2\) and \(n  \pi_{\text{rare,min}} \geq 2\) due to assumption \eqref{a.eigen.eq}. So the assumptions of Lemma \ref{bin.prob.bound} are satisfied for \(q = n \pi_{\text{rare,min}}/2\), and there are more than \(n \pi_{\text{rare,min}}/2\) observations in each class with probability at least 
\begin{align*}
 1 - K \exp \left\{ -2n \left( \pi_{\text{rare,min}} - \pi_{\text{rare,min}}/2\} \right)^2\right\}  = ~ & 1 - K \exp \left\{ -\frac{1}{2} \pi_{\text{rare,min}}^2 n\right\}
 .
 \end{align*}
\end{proof}

\begin{proof}[Proof of Lemma \ref{new.a.min.lem}]

We will prove the result by using a concentration inequality on the minimum singular value of a random matrix (which will correspond to the square root of the minimum eigenvalue of \(\frac{1}{n_k}  \sum_{i: y_i = k} \boldsymbol{x}_i \boldsymbol{x}_i^\top\); recall that the eigenvalues of \(\boldsymbol{A}^\top \boldsymbol{A}\) are the squares of the singular values of \(\boldsymbol{A}\)). However, the result we use applies only to random matrices with isotropic second moment matrices, so we need to standardize \(\frac{1}{n_k}  \sum_{i: y_i = k} \boldsymbol{x}_i \boldsymbol{x}_i^\top\) by its inverse square root second moment matrix. We can relate this quantity to the minimum eigenvalue of \(\frac{1}{n_k}  \sum_{i: y_i = k} \boldsymbol{x}_i \boldsymbol{x}_i^\top\) by our claim (that we will verify later) that for a symmetric positive semidefinite matrix \(\boldsymbol{S}\) and symmetric positive definite \(\boldsymbol{\Sigma}\) it holds that
\begin{equation}\label{lam.min.claim}
\lambda_{\text{min}}\left( \boldsymbol{S} \right)  \geq \lambda_{\text{min}} \left( \boldsymbol{\Sigma}^{-1/2} \boldsymbol{S} \boldsymbol{\Sigma}^{-1/2} \right) \lambda_{\text{min}} \left( \boldsymbol{\Sigma} \right).
\end{equation}
Then substituting \(\boldsymbol{S} =   \frac{1}{n_k}  \sum_{i: y_i = k } \boldsymbol{x}_i \boldsymbol{x}_i^\top  \) and \( \boldsymbol{\Sigma} = \boldsymbol{\Sigma}_k = \E \left[  \boldsymbol{x}_i \boldsymbol{x}_i^\top \mid y_i = k\right] \) (which we assumed has a strictly postive minimum eigenvalue and is therefore invertible) into \eqref{lam.min.claim} yield that on an event where \(n_k = \sum_{i=1}^n \mathbbm{1} \left\{y_i = k\right\} \geq 1\), we have that almost surely
\begin{align}
\lambda_{\text{min}}\left(   \frac{1}{n_k}  \sum_{i: y_i = k } \boldsymbol{x}_i \boldsymbol{x}_i^\top  \right)  \geq ~ &  \lambda_{\text{min}} \left(  \frac{1}{n_k}  \boldsymbol{\Sigma}_k^{-1/2}     \sum_{i: y_i = k } \boldsymbol{x}_i \boldsymbol{x}_i^\top  \boldsymbol{\Sigma}_k^{-1/2} \right) \lambda_{\text{min}} \left( \boldsymbol{\Sigma}_k \right) \nonumber
\\  \geq ~ &  \lambda_{\text{min}} \left(  \frac{1}{n_k}  \boldsymbol{\Sigma}_k^{-1/2}     \sum_{i: y_i = k } \boldsymbol{x}_i \boldsymbol{x}_i^\top  \boldsymbol{\Sigma}_k^{-1/2} \right) \lambda_{\text{min}}^*  \label{min.eigen}
,
\end{align}
where the second line uses our assumption from the statement of Theorem \ref{cons.thm} and the fact that \(  \frac{1}{n_k}  \boldsymbol{\Sigma}_k^{-1/2}     \sum_{i: y_i = k } \boldsymbol{x}_i \boldsymbol{x}_i^\top  \boldsymbol{\Sigma}_k^{-1/2}  \) is almost surely positive semidefinite. Next we lower-bound the minimum eigenvalue of the random matrix \( \frac{1}{n_k}  \boldsymbol{\Sigma}_k^{-1/2}     \sum_{i: y_i = k } \boldsymbol{x}_i \boldsymbol{x}_i^\top  \boldsymbol{\Sigma}_k^{-1/2} \) using a concentration inequality from \citet{vershynin_2012}.
Observe that (as \citealt{vershynin2018high} points out in Example 2.5.8(c)) the \(\psi_2\) norm of a bounded random vector \(\boldsymbol{T} \in \mathbb{R}^p\) can be upper-bounded as follows:
\begin{align*}
\lVert \boldsymbol{T} \rVert_{\psi_2} = ~ & \sup_{\boldsymbol{v}: \boldsymbol{v} \in \mathbb{R}^p, \lVert \boldsymbol{v} \rVert_2 = 1} \left\{ \lVert \boldsymbol{T}^\top \boldsymbol{v} \rVert_{\psi_2} \right\}
\\ = ~ & \sup_{\boldsymbol{v}: \boldsymbol{v} \in \mathbb{R}^p, \lVert \boldsymbol{v} \rVert_2 = 1} \left\{  \inf \left\{ t > 0: \E \exp \left(  \frac{\left[ \boldsymbol{T}^\top \boldsymbol{v}\right]^2}{t^2}  \right) \leq 2 \right\} \right\}
\\ \leq ~ & \inf \left\{ t > 0: \exp \left(  \frac{\lVert \boldsymbol{T} \rVert_\infty^2}{t^2}  \right) \leq 2 \right\} 
\\ = ~ &   \frac{\lVert \boldsymbol{T} \rVert_\infty}{ \sqrt{\log 2} }
.
\end{align*}
Therefore under our assumptions \(\boldsymbol{x}_i \mid \boldsymbol{y}\) is bounded and therefore subgaussian, with \(\psi_2\) norm at most \(\lVert \boldsymbol{x}_i \rVert_{\psi_2} \leq \lVert \boldsymbol{x} \rVert_\infty   /\sqrt{\log 2} =  C_4/\sqrt{\log 2}\) for all \(k\). So for any \(k \in \{1, \ldots, K\}\), we have from Theorem 5.39 in \citet{vershynin_2012} that for any \(t \geq 0\) there exists \(c > 0\) such that the event
\begin{equation}\label{cond.prod.thm.5.39.guarantee}
\mathbb{P} \left( \sigma_{\text{min}} \left( \frac{1}{n_k}  \sum_{i: y_i = k} \boldsymbol{\Sigma}_k^{-1/2} \boldsymbol{x}_i   \right)  \geq \sqrt{n_k} - C \sqrt{p_n} - t  \mid \boldsymbol{y} \right) \geq 1 - 2 \exp\{-ct^2\} 
\end{equation}
holds almost surely for \(C = C_4^2/(\log 2) \sqrt{\log(9)/c_1}\), where \(C_4\) is as defined in the statement of Theorem \ref{cons.thm}, \(c_1\) is a constant from \citet{vershynin_2012}, \(\sigma_{\text{min}} ( \cdot)\) denotes the minimum singular value, and if the set \(\{i: y_i = k\}\) is empty we define \(\frac{1}{n_k}  \sum_{i: y_i = k} \boldsymbol{\Sigma}_k^{-1/2} \boldsymbol{x}_i \) to equal \(\boldsymbol{0}_p\) (and note that the inequality then trivially holds with probability one in this case, because \(n_k = 0\) so the right side is nonpositive). 
For \(k \in \{1, \ldots, K\}\) and \(t \geq 0\), define the event \(\mathcal{E}_k(t)\) by 
\[
\mathcal{E}_k(t) := \left\{ \sigma_{\text{min}} \left( \frac{1}{n_k}  \sum_{i: y_i = k} \boldsymbol{\Sigma}_k^{-1/2} \boldsymbol{x}_i   \right)  \geq \sqrt{n_k} - C \sqrt{p_n} - t  \right\}
,
\]
and observe that 
\begin{equation}\label{e.k.t.prob}
\mathbb{P}(\mathcal{E}_k(t) ) \geq 1 - 2 \exp\{-ct^2\}
\end{equation}
 by marginalizing \eqref{cond.prod.thm.5.39.guarantee} over \(\boldsymbol{y}\).
Now we consider a particular choice of \(t\). 
From our assumption \eqref{a.eigen.eq} we have
\begin{align*}
n \pi_{\text{rare, min}} > ~ & 2 \left( C \sqrt{p_n} + \sqrt{\frac{a}{\lambda_{\text{min}}^*}} \right)^2
 \\ \iff \qquad \sqrt{\frac{n \pi_{\text{rare,min}}}{2} } - C \sqrt{p_n}  - \sqrt{\frac{a}{\lambda_{\text{min}}^*}} > ~ &    0,
 \end{align*}
so we can choose \(t = \sqrt{ n \pi_{\text{rare,min}}  /2}  - C \sqrt{p_n} - \sqrt{\frac{a}{\lambda_{\text{min}}^*}}\), yielding that on 
\[
\mathcal{E}_k \left( \sqrt{\frac{n \pi_{\text{rare,min}}}{2} }  - C \sqrt{p_n} - \sqrt{\frac{a}{\lambda_{\text{min}}^*}} \right) \cap \{n_k \geq 1\}
\]
we have
\begin{align}
\sigma_{\text{min}} \left( \frac{1}{n_k}  \sum_{i: y_i = k} \boldsymbol{\Sigma}_k^{-1/2} \boldsymbol{x}_i   \right)  \geq ~ & \sqrt{n_k} - C \sqrt{p_n} - \sqrt{\frac{n \pi_{\text{rare,min}}}{2} }+ C \sqrt{p_n} + \sqrt{\frac{a}{\lambda_{\text{min}}^*}}  \nonumber
\\ = ~ & \sqrt{n_k} -  \sqrt{\frac{n \pi_{\text{rare,min}}}{2} }  + \sqrt{\frac{a}{\lambda_{\text{min}}^*}} \label{sing.val.ineq}
.
\end{align}
That is, on this event we can lower-bound the minimum eigenvalue of \( \frac{1}{n_k}  \boldsymbol{\Sigma}_k^{-1/2}     \sum_{i: y_i = k } \boldsymbol{x}_i \boldsymbol{x}_i^\top  \boldsymbol{\Sigma}_k^{-1/2} \) provided that \(n_k\) is at least 1 and large enough that the right side of \eqref{sing.val.ineq} is nonnegative. Next we work on lower-bounding \(n_k\) with high probability. Consider the event
\begin{align*}
\mathcal{N} :=  ~ & \left\{ n_k \geq  \frac{n \pi_{\text{rare,min}}  }{2} \qquad \forall k \in \{1, \ldots, K\}  \right\}
.
\end{align*}
Notice that on \(\mathcal{N}\) we have that the lower bound in \eqref{sing.val.ineq} is at least \(\sqrt{a/\lambda_{\text{min}}^*}\) and \(n_k  \geq 1\) for all \(k\) due to \eqref{a.eigen.eq}. That is, for any \(k\), on \(\mathcal{N} \cap \mathcal{E}_k \left(  \sqrt{ n \pi_{\text{rare,min}}  /2}  - C \sqrt{p_n} - \sqrt{\frac{a}{\lambda_{\text{min}}^*}} \right)\) inequality \eqref{sing.val.ineq} holds and yields \(\sigma_{\text{min}} \left( \frac{1}{n_k}  \sum_{i: y_i = k} \boldsymbol{\Sigma}_k^{-1/2} \boldsymbol{x}_i   \right) \geq  \sqrt{a/\lambda_{\text{min}}^*}  \). Since the eigenvalues of \(\frac{1}{n_k}  \sum_{i: y_i = k}  \boldsymbol{\Sigma}_k^{-1/2}  \boldsymbol{x}_i \boldsymbol{x}_i^\top  \boldsymbol{\Sigma}_k^{-1/2} \) are the squares of the singular values of \(\frac{1}{n_k}  \sum_{i: y_i = k} \boldsymbol{\Sigma}_k^{-1/2} \boldsymbol{x}_i \), on \(\mathcal{N} \cap  \mathcal{E}_k \left(  \sqrt{ n \pi_{\text{rare,min}}  /2}  - C \sqrt{p_n} - \sqrt{\frac{a}{\lambda_{\text{min}}^*}} \right)\) we have
\[
\lambda_{\text{min}} \left( \frac{1}{n_k}  \sum_{i: y_i = k} \boldsymbol{\Sigma}_k^{-1/2} \boldsymbol{x}_i \boldsymbol{x}_i^\top  \boldsymbol{\Sigma}_k^{-1/2}  \right) \geq  \frac{a}{\lambda_{\text{min}}^*} > 0.
\]
Finally, substituting this into \eqref{min.eigen} we have that on \(\left( \bigcap_{k=1}^K  \mathcal{E}_k \left(  \sqrt{ n \pi_{\text{rare,min}}  /2}  - C \sqrt{p_n} - \sqrt{\frac{a}{\lambda_{\text{min}}^*}} \right) \right) \cap  \mathcal{N} \) 
\begin{align*}
\lambda_{\text{min}} \left( \frac{1}{n_k}  \sum_{i: y_i = k} \boldsymbol{x}_i \boldsymbol{x}_i^\top   \right) 
\geq ~ &  \frac{a}{\lambda_{\text{min}}^*}  \lambda_{\text{min}}^* = a \qquad \forall k \in \{1, \ldots, K\}.
\end{align*}
Using a union bound, this holds with probability at least
\begin{align*}
& \mathbb{P} \left( \left( \bigcap_{k=1}^K  \mathcal{E}_k \left(  \sqrt{\frac{n \pi_{\text{rare,min}}}{2} }  - C \sqrt{p_n} - \sqrt{\frac{a}{\lambda_{\text{min}}^*}} \right) \right) \cap   \mathcal{N} \right) 
%
\\ \geq ~ & 1-  \sum_{k=1}^K \mathbb{P} \left(   \mathcal{E}_k^c \left(  \sqrt{\frac{n \pi_{\text{rare,min}}}{2} }  - C \sqrt{p_n} - \sqrt{\frac{a}{\lambda_{\text{min}}^*}} \right) \right)  - \mathbb{P} \left(   \mathcal{N}^c \right) 
\\ \geq ~ & 1- 2 K  \exp \left\{-c \left( \sqrt{ \frac{ n \pi_{\text{rare,min}}  }{2}}  - C \sqrt{p_n} - \sqrt{\frac{a}{\lambda_{\text{min}}^*}} \right)^2 \right\}\  -  K \exp \left\{ -\frac{1}{2} \pi_{\text{rare,min}}^2 n\right\}
,
\end{align*}
where in the last step we applied \eqref{e.k.t.prob} and Lemma \ref{pos.def.lem}.

Finally, we show \eqref{lam.min.claim}. Observe that for any \(\boldsymbol{v} \) with \( \lVert \boldsymbol{v} \rVert_2 = 1\)
\begin{align*}
\boldsymbol{v}^\top \boldsymbol{S} \boldsymbol{v} = ~ & \left( \boldsymbol{\Sigma}^{1/2} \boldsymbol{v}  \right)^\top  \boldsymbol{\Sigma}^{-1/2} \boldsymbol{S} \boldsymbol{\Sigma}^{-1/2}  \left( \boldsymbol{\Sigma}^{1/2} \boldsymbol{v}  \right)
\\ \geq ~ & \lambda_{\text{min}} \left(  \boldsymbol{\Sigma}^{-1/2} \boldsymbol{S} \boldsymbol{\Sigma}^{-1/2} \right) \lVert \boldsymbol{\Sigma}^{1/2} \boldsymbol{v} \rVert_2^2
\\ = ~ & \lambda_{\text{min}} \left(  \boldsymbol{\Sigma}^{-1/2} \boldsymbol{S} \boldsymbol{\Sigma}^{-1/2} \right)  \boldsymbol{v}^\top \boldsymbol{\Sigma} \boldsymbol{v}
\\ \geq ~ & \lambda_{\text{min}} \left(  \boldsymbol{\Sigma}^{-1/2} \boldsymbol{S} \boldsymbol{\Sigma}^{-1/2} \right)  \lambda_{\text{min}} \left(    \boldsymbol{\Sigma} \right),
\end{align*}
proving the claim. 

\end{proof}

\begin{proof}[Proof of Lemma \ref{bin.prob.bound}]


 By Hoeffding's inequality we have that for any \(k \in \{1, \ldots, K\}\) and any \(q \in \{1, \ldots, \lfloor n\pi_{\text{rare,min}} \rfloor\}\), \(\mathbb{P} \left( n_k \leq q \right) \leq   \exp \left\{ -2n \left( \pi_{\text{rare,min}} - \frac{q}{n} \right)^2\right\}\). Then using a union bound, we have
\begin{align*}
 \mathbb{P} \left(  \bigcap_{k=1}^K \left\{ n_k >   q  \right\} \right) 
 = ~ & 1 - \mathbb{P} \left(  \bigcup_{k=1}^K \left\{ n_k \leq   q  \right\} \right) 
\\ \geq ~ &  1 - \sum_{k=1}^K \mathbb{P} \left(  n_k \leq   q \right) 
\\ \geq ~ &  1 - K \exp \left\{ -2n \left( \pi_{\text{rare,min}} - \frac{q}{n} \right)^2\right\}
.
\end{align*}

\end{proof}

\section{Estimating \textsc{PRESTO}}\label{sec.est.presto}

Code generating all plots and tables that appear in the paper is available in the public GitHub repo at \url{https://github.com/gregfaletto/presto}. Executing the code in that repo generates the exact plots that appear in this version of the paper.

In all data experiments, we implemented PRESTO by slightly modifying the code for the R \texttt{ordinalNet} package (version 2.12), as we discuss below.

The synthetic data experiments from Sections \ref{sparse.sim} and \ref{dense.sim}, as well as Simulation Studies A and B in Appendix \ref{main.thm.sec}, were conducted in R Version 4.2.2 running on macOS 10.15.7 on an iMac with a 3.5 GHz Quad-Core Intel Core i7 processor and 32 GB or RAM. We used the R packages \texttt{MASS} (version 7.3.58.1), \texttt{simulator} (version 0.2.4), \texttt{ggplot2} (version 3.3.6), \texttt{cowplot} (version 1.1.1), and \texttt{stargazer} (version 5.2.3), all available for download on CRAN, as well as the base \texttt{parallel} package (version 4.2.2).

The real data experiments from Section \ref{real.data} and Appendix \ref{real.data.2} were conducted in R Version 4.3.0 running on macOS Ventura 13.3.1 on a MacBook Pro with a 2.3 GHz Quad-Core Intel Core i5 processor and 16 GB or RAM. We used the R packages \texttt{MASS} (version 7.3.58.4), \texttt{simulator} (version 0.2.4), \texttt{ggplot2} (version 3.4.2), \texttt{cowplot} (version 1.1.1), and \texttt{stargazer} (version 5.2.3), all available for download on CRAN, as well as the base \texttt{parallel} package (version 4.3.0). The data for the real data experiment from Section \ref{real.data} is the \texttt{soup} data set from the R \texttt{ordinal} package (version 2022.11.16). The data for the real data experiment from Appendix \ref{real.data.2} is the \texttt{PreDiabetes} data set from the R \texttt{MLDataR} package (version 1.0.1).

In the remainder of this section, we describe our implementation of PRESTO. We will use the notation and terminology of \citet{ordnet}, with the exception of continuing to use our convention of \(K\) total categories. We fit \textsc{PRESTO} by reparameterizing the efficient coordinate descent algorithm used to estimate \(\ell_1\)-penalized ordinal regression models in the R \texttt{ordinalNet} \citep{ordnet} package, in much the same way that the generalized lasso can be implemented using reparameterization; see \citet[Section 3]{tibshirani2011solution} or Section 4.5.1.1 of \citet{hastie2015statistical} for a textbook-level discussion.

The \texttt{ordinalNet} package implements elastic net penalized ordinal regression, including an \(\ell_1\)-penalized (or \(\ell_2\)-penalized, as we explored in Section \ref{ridge.supp}) relaxation of the proportional odds model. The parameters we seek to model are \(\boldsymbol{\eta}_i \in \mathbb{R}^{K-1}\), which model the probabilities of the \(K\) outcomes by the relation \(\boldsymbol{\eta}_i = g(\boldsymbol{p}_i)\), where \(\boldsymbol{p}_i \in \mathcal{S}^{K-1}\) (where \(\mathcal{S}^{K-1} := \{ \boldsymbol{p}: \boldsymbol{p} \in (0, 1)^{K-1}, \lVert p \rVert_1 < 1\}\) are the probabilities of outcomes \(\{1, \ldots, K -1\}\); in particular, \(p_{ik} = \mathbb{P}(y_i = k)\) for \(k \in \{1, \ldots, K-1\}\) and \(\mathbb{P}(y_i = K) = 1 - \sum_{k=1}^{K-1} p_{ik}\) ) and \(g: \mathcal{S}^{K-1} \to \mathbb{R}^{K-1} \) is an invertible function. For our model, the forwards cumulative probability model, where \(p_{ik} = \mathbb{P}(y_i \leq k)\),
\[
 \begin{bmatrix} g(\boldsymbol{p}_i) \end{bmatrix}_k = \log \left( \frac{\sum_{k'=1}^k p_{ik'} }{ 1 - \sum_{k'=1}^k p_{ik'}} \right), \qquad k \in \{1, \ldots, K\}.
\]
We choose the nonparallel model \(\boldsymbol{\eta}_i = \boldsymbol{c} + \boldsymbol{B}^\top \boldsymbol{x}_{i}\), where \(\boldsymbol{x}_{i}\) is the \(i^\text{th}\) row of \(\boldsymbol{X}\), \(\boldsymbol{c} \in \mathbb{R}^{K-1}\) is a vector of intercepts, and
\[
\boldsymbol{B} = \begin{bmatrix}  \boldsymbol{B}_{\cdot 1} & \cdots &   \boldsymbol{B}_{\cdot K - 1} \end{bmatrix}
\]
is a \(p \times (K-1)\) matrix of coefficients. Observe that we can simply write \(\boldsymbol{\eta}_i =  \boldsymbol{\tilde{X}}_{i} \boldsymbol{\beta}\), for \(\boldsymbol{\beta} \in \mathbb{R}^{(K-1)(p+1)}\) defined by
\[
\boldsymbol{\beta}  = \begin{pmatrix}
\boldsymbol{c}
\\ \boldsymbol{B}_{\cdot 1}
\\ \vdots
\\ \boldsymbol{B}_{\cdot K - 1}
\end{pmatrix}
\]
and \(\boldsymbol{\tilde{X}}_{i} \in \mathbb{R}^{(K-1) \times (K-1)(p+1)} \) defined by
\[
\boldsymbol{\tilde{X}}_{i} = \begin{pmatrix}
& \boldsymbol{x}_i^\top & \boldsymbol{0}_{p}^\top & \cdots & \boldsymbol{0}_p^\top & \boldsymbol{0}_p^\top
\\ & \boldsymbol{0}_{p}^\top &\boldsymbol{x}_i^\top & \cdots  & \boldsymbol{0}_p^\top
 & \boldsymbol{0}_p^\top
\\ \boldsymbol{I}_{K-1} &  \vdots & \vdots & \ddots & \vdots & \vdots
\\ &    \boldsymbol{0}_{p}^\top&    \boldsymbol{0}_{p}^\top  & \cdots &   \boldsymbol{x}_i^\top  &   \boldsymbol{0}_{p}^\top
\\ &   \boldsymbol{0}_{p}^\top  &  \boldsymbol{0}_{p}^\top& \cdots &  \boldsymbol{0}_{p}^\top &   \boldsymbol{x}_i^\top 
\end{pmatrix} = \begin{pmatrix} \boldsymbol{I}_{K-1}  & \boldsymbol{I}_{K-1} \otimes   \boldsymbol{x}_i^\top  \end{pmatrix} .
\]
The R \texttt{ordinalNet} package solves the convex \citep{wurm2017regularized, Pratt1981} optimization problem
\begin{align*}
& \underset{\boldsymbol{\beta}}{\arg \min}
\left\{ - \frac{1}{n} \sum_{i=1}^n \ell_i \left( g^{-1} \left(  \boldsymbol{\tilde{X}}_{i}^\top \boldsymbol{\beta} \right) \right)  +  \lambda \sum_{q=K}^{(K-1)(p+1)} \left| \beta_q \right| \right\}
\end{align*}
where
\[
\ell_i(\boldsymbol{p}_i) = \sum_{k=1}^{K-1} \mathbbm{1} \{y_i = k \} \log p_{ik} +  \mathbbm{1} \{y_i = K  \}  \log \left( 1 - \sum_{k=1}^{K-1}  p_{ik} \right).
\]
We would like to place an \(\ell_1\) penalty on the first differences \(B_{j,k+1} - B_{jk}\) for all \(k \in \{1, \ldots, K-2\}\). We can do this through the parameterization \(\boldsymbol{\Psi} \in \mathbb{R}^{p \times (K-1)}\) defined by
\[
\Psi_{jk} = \begin{cases}
B_{j1}, & j \in [p], k = 1,
\\  B_{jk} - B_{j,k-1}, & j \in [p], k \in \{2, \ldots, K - 1\}.
\end{cases} 
\]
Observe that these matrices are related by
\[
B_{jk} = \sum_{k'=1}^k \Psi_{jk'},
\]
so
\[
\eta_{ik} = c_{0k} + \boldsymbol{B}_{\cdot k}^\top \boldsymbol{x}_i =  c_{0k} + \sum_{k'=1}^k \boldsymbol{\Psi}_{ \cdot k'}^\top \boldsymbol{x}_i, \qquad k \in \{1, \ldots, K -1\}.
\]
Therefore we can simply write
\[
\boldsymbol{\eta}_i =  \boldsymbol{\tilde{X}}_{i}' \boldsymbol{\beta}',
\]
for \(\boldsymbol{\beta}' \in \mathbb{R}^{(K-1)(p+1)}\) defined by
\[
\boldsymbol{\beta}'  = \begin{pmatrix}
\boldsymbol{c}
\\ \boldsymbol{\Psi}_{ \cdot 1}
\\ \vdots
\\ \boldsymbol{\Psi}_{ \cdot K -1}
\end{pmatrix}
\]
and \(\boldsymbol{\tilde{X}}_{i}' \in \mathbb{R}^{(K-1) \times (K-1)(p+1)} \) defined by
\begin{equation}\label{mod.x}
\boldsymbol{\tilde{X}}_{i}' = \begin{pmatrix}
& \boldsymbol{x}_i^\top & \boldsymbol{0}_{p}^\top & \cdots & \boldsymbol{0}_p^\top & \boldsymbol{0}_p^\top
\\ & \boldsymbol{x}_i^\top &\boldsymbol{x}_i^\top & \cdots  & \boldsymbol{0}_p^\top
 & \boldsymbol{0}_p^\top
\\ \boldsymbol{I}_{K-1} &  \vdots & \vdots & \ddots & \vdots & \vdots
\\ &   \boldsymbol{x}_i^\top  &   \boldsymbol{x}_i^\top  & \cdots &   \boldsymbol{x}_i^\top  &   \boldsymbol{0}_{p}^\top
\\ &   \boldsymbol{x}_i^\top  &   \boldsymbol{x}_i^\top  & \cdots &   \boldsymbol{x}_i^\top  &   \boldsymbol{x}_i^\top 
\end{pmatrix}.
\end{equation}
We then seek to solve the slightly modified optimization problem (modification highlighted)
\begin{align*}
& \underset{\boldsymbol{\beta}'}{\arg \min}
\left\{ - \frac{1}{n} \sum_{i=1}^n \ell_i \left( g^{-1} \left(   \underbrace{\left(\boldsymbol{\tilde{X}}_{i}' \right)^\top}_{(*)} \boldsymbol{\beta}' \right) \right) +  \lambda \sum_{q=K}^{(K-1)(p+1)}  \left| \beta_q' \right|.  \right\}
\end{align*}
Though this can not be implemented within the framework of the existing \texttt{ordinalNet} package, the above modification only requires changing a handful of lines of the publicly available source code in the \texttt{ordinalNet} package. Though our implementation is simple, using the modified design matrix (the change above) could cause convergence of parameter estimation to be slow, because the resulting lasso problem effectively has highly correlated features, which slows the convergence of coordinate descent. See Section 4.5.1.1 of \citet{hastie2015statistical} for a textbook-level discussion of this point. Though this is a convenient way to estimate PRESTO, it is not the best approach for estimating PRESTO in practice, and the limitations of our estimation approach are not limitations of PRESTO itself. Known efficient strategies for solving optimization problems like PRESTO are discussed in \citet[Section 4.5]{hastie2015statistical} and \citet{Arnold2016}. Though PRESTO does not strictly lie within the class of optimization problems considered by \citet{xin2014efficient}, some of their ideas might also be helpful for better estimating PRESTO. An ADMM approach like the one proposed by \cite{zhu2017augmented} for generalized lasso problems may also be useful for estimating PRESTO. \citet{ko2019easily} propose their own methods for solving similar optimization problems.

\end{document}